\documentclass[reqno]{amsart}
\usepackage{fullpage,amsthm,amssymb,amsfonts,graphicx,mathrsfs,textcomp,color,fouridx}
\newtheorem{rhp}{Riemann-Hilbert Problem}
\newtheorem{ass}{Assumption}

\newtheorem{proposition}{Proposition}
\newtheorem{theorem}{Theorem}
\newtheorem{corollary}{Corollary}
\newtheorem{definition}{Definition}
\newcommand{\trans}[1]{{#1}^{\ensuremath{\mathsf{T}}}}
\newcommand{\breather}{{\sf{B}}}
\newcommand{\kink}{{\sf{K}}}
\newcommand{\grazing}{{\sf{G}}}
\newcommand{\bo}{\mathcal{O}}
\newcommand{\lo}{\mathfrak{o}}
\newcommand{\pu}{\mathcal{U}}
\newcommand{\pv}{\mathcal{V}}
\newcommand{\pw}{\mathcal{W}}
\newcommand{\pz}{\mathcal{Z}}
\newcommand{\ind}{m}
\newcommand{\sv}{\tau}
\numberwithin{equation}{section}
\numberwithin{theorem}{section}
\numberwithin{proposition}{section}
\numberwithin{lemma}{section}
\numberwithin{corollary}{section}
\numberwithin{ass}{section}
\numberwithin{rhp}{section}
\numberwithin{figure}{section}
\numberwithin{definition}{section}
\title{The Sine-Gordon Equation in the Semiclassical Limit:  Critical Behavior near a Separatrix}
\author{Robert J. Buckingham and Peter D. Miller}
\date{\today\\The authors thank A. B. J. Kuijlaars and A. R. Its for useful discussions.  R. J. Buckingham was partially supported by the Charles Phelps Taft Research Foundation.  P. D. Miller was partially supported by the National Science Foundation under grant DMS-0807653.  }
\begin{document}
\begin{abstract}
We study the Cauchy problem for the sine-Gordon equation in the semiclassical limit with pure-impulse
initial data of sufficient strength to generate both high-frequency rotational motion near the peak of the impulse profile and also high-frequency librational motion in the tails.  Subject to suitable conditions of a general nature, we analyze the fluxon condensate solution approximating the given initial data  for small time
near points where the initial data crosses the separatrix of the phase portrait of the simple pendulum.  We show that the solution is locally constructed as a universal curvilinear grid of 
superluminal kinks and grazing collisions thereof, with the grid curves being determined from rational solutions of the Painlev\'e-II system.
\end{abstract}
\maketitle
\tableofcontents
\section{Introduction}
This paper is concerned with a detailed local analysis of the solution of the Cauchy initial-value problem for the sine-Gordon equation
\begin{equation}
\epsilon^2u_{tt}-\epsilon^2u_{xx}+\sin(u)=0,\quad u(x,0)=F(x),\quad \epsilon u_t(x,0)=G(x),
\quad x\in\mathbb{R}.
\label{eq:sGCauchy}
\end{equation}
We will consider the number $\epsilon>0$ to be a small parameter.  This type of scaling can be physically motivated in the situation that the sine-Gordon equation is used to model the propagation of magnetic flux along superconducting Josephson junctions \cite{ScottCR76}.  
The sine-Gordon equation can also be derived in the continuum limit as a model for an array of
coaxial pendula with nearest-neighbor torsion coupling \cite{BaroneEMS71}.  This latter application is particularly
useful for the purposes of visualization of solutions.

A dramatic separation of scales occurs in the \emph{semiclassical limit}  $\epsilon\downarrow 0$ 
if the initial data $F(\cdot)$ and $G(\cdot)$ are held fixed.  As can be seen in Figures~\ref{fig:exactplots4}--\ref{fig:exactplots16}, the semiclassical dynamics apparently consists of well-defined (asymptotically independent of $\epsilon$) spacetime regions containing oscillations on space and time scales proportional to $\epsilon$ but modulated over longer scales originating with the $\epsilon$-independent initial conditions.  An important role is played by the
$x$-parametrized curve $(F,G)=(F(x),G(x))$ in the phase portrait of the simple pendulum, and its
relation to the separatrix curve $(1-\cos(F))+\tfrac{1}{2}G^2=2$.  Indeed, in the specific context of suitable
initial data of \emph{pure impulse} type, that is, for which $F(x)\equiv 0$, the following dichotomy has recently been established \cite{BuckinghamMelliptic} regarding the asymptotic behavior
of the solution $u(x,t)$  of \eqref{eq:sGCauchy}
for small time $t$ independent of $\epsilon$.   If $(F,G)=(0,G(x))$
lies inside the separatrix, then $u(x,t)$ 
is accurately modeled by a modulated train of superluminal librational waves; but if $(F,G)=(0,G(x))$ lies outside the separatrix, then $u(x,t)$ is instead accurately modeled by a
modulated train of superluminal rotational waves.  If $x\in\mathbb{R}$ is a point lying exactly
on the separatrix curve, then the approximation theorems proved in \cite{BuckinghamMelliptic}
fail to provide a uniform description of the asymptotic behavior of the solution $u(x,t)$ near such $x$.
It is clear from plots of exact solutions shown in Figures~\ref{fig:exactplots4}--\ref{fig:exactplots16} that some essentially
different asymptotic behavior is generated by separatrix crossings in the pure-impulse initial
data.  In particular, a different and more complicated kind of waveform than modulated traveling waves appears to spread in time away from specific points where separatrix crossings occur in the
initial data.  The coupled pendulum interpretation is useful here:  if the pendula are all initially at rest in the gravitationally stable configuration and are given a spatially-localized initial impulse
of sufficient strength, then some pendula have sufficient energy to rotate completely about the axis a number of times, while the pendula in the ``wings'' experience very little initial impulse and only have energy for small oscillatory motions near equilibrium (so-called librational motion).  Clearly
this situation leads to kink generation near the transition points and a more complicated type of dynamics as the kinks struggle to separate from one another.
\begin{figure}[h]
\begin{center}
\includegraphics{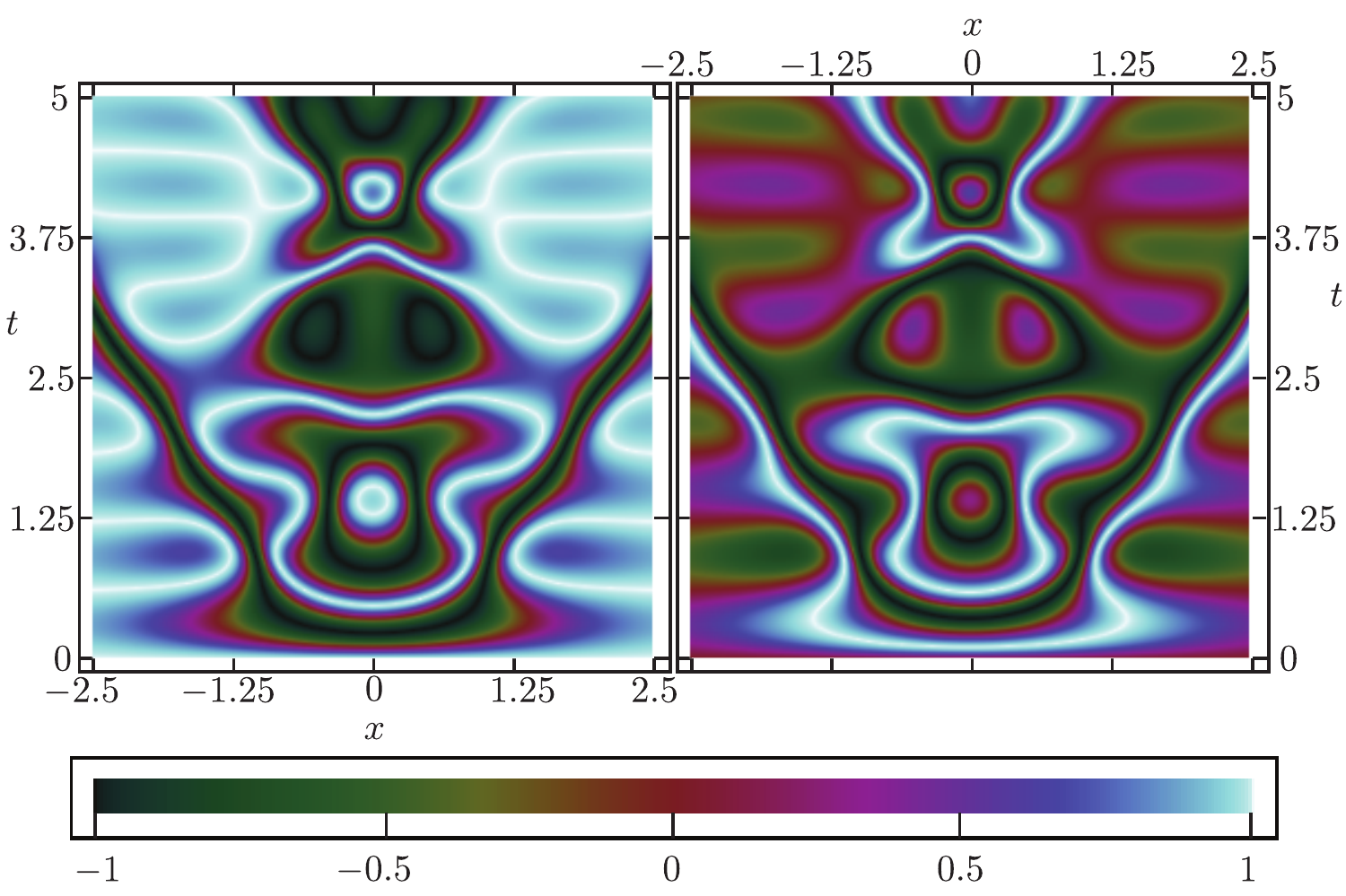}
\end{center}
\caption{\emph{Density plots of $\cos(u(x,t))$ (left) and $\sin(u(x,t))$ (right) for the exact solution of the Cauchy problem
\eqref{eq:sGCauchy} with $\epsilon=0.1875$ for $F(x)\equiv 0$ and $G(x)=-3\,\mathrm{sech}(x)$.  For this data, separatrix crossings occur for $x=\pm\,\mathrm{arcsech}(\tfrac{2}{3})\approx \pm 0.962$.}}
\label{fig:exactplots4}
\end{figure}
\begin{figure}[h]
\begin{center}
\includegraphics{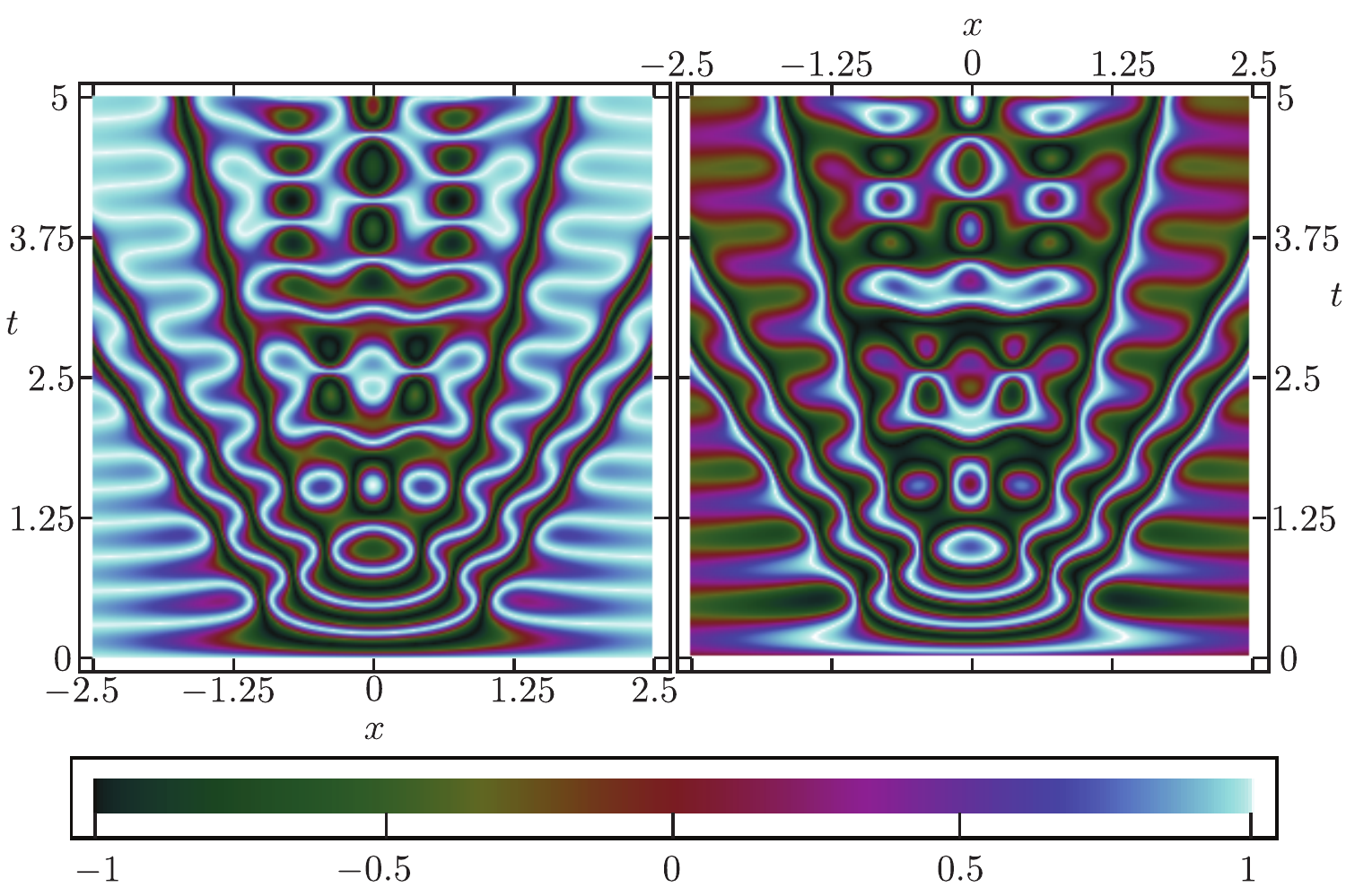}
\end{center}
\caption{\emph{Same as Figure~\ref{fig:exactplots4} but with $\epsilon=0.09375$.}}
\label{fig:exactplots8}
\end{figure}
\begin{figure}[h]
\begin{center}
\includegraphics{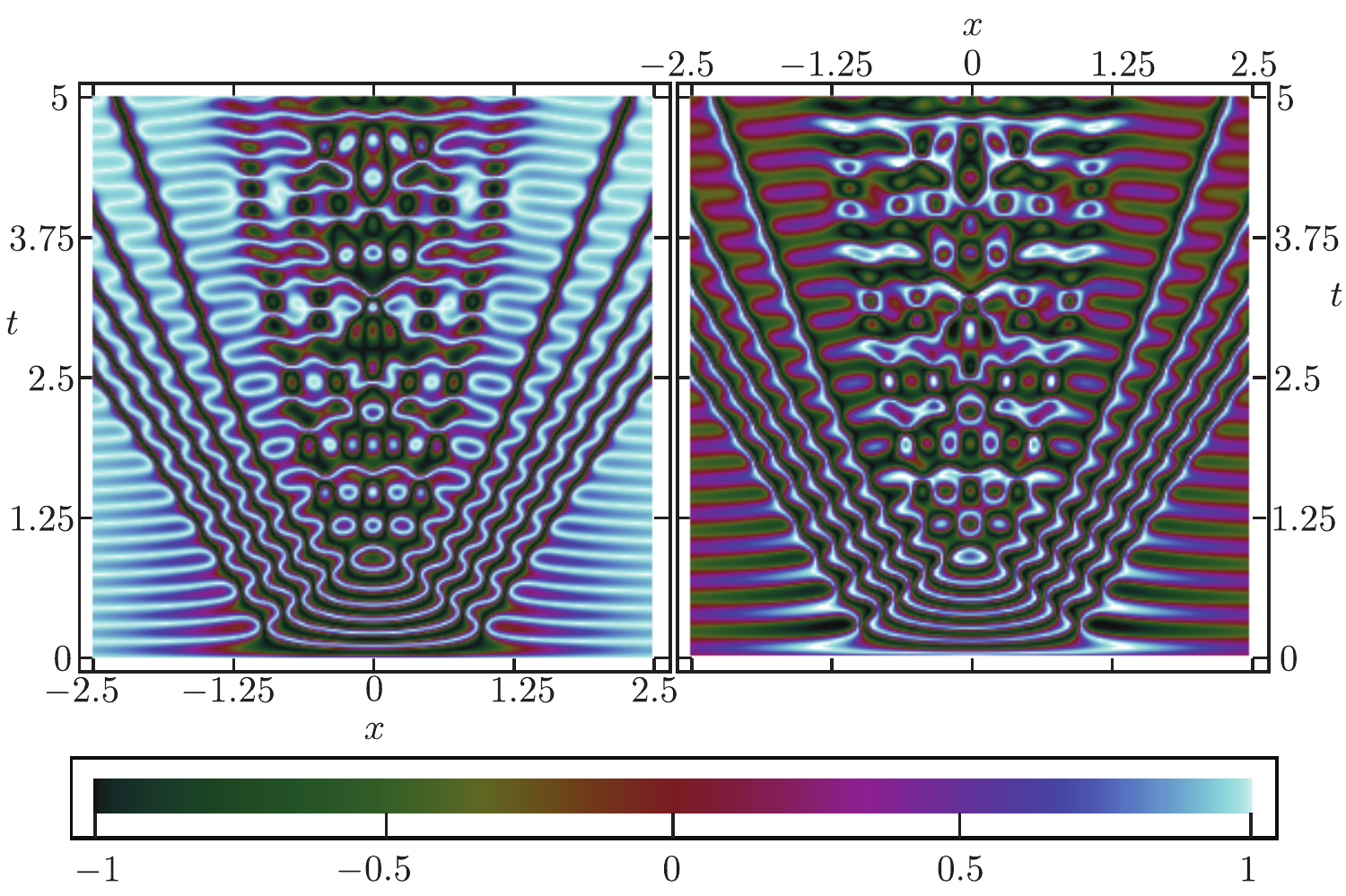}
\end{center}
\caption{\emph{Same as Figure~\ref{fig:exactplots4} but with $\epsilon=0.046875$.}}
\label{fig:exactplots16}
\end{figure}


To better understand the reason for the locally complicated behavior, it is useful to  view the sine-Gordon equation as
a perturbation of the simple pendulum ODE:
\begin{equation}
u_{TT}+\sin(u)=\epsilon^2u_{xx},\quad u(x,0)=0,\quad u_T(x,0)=G(x),\quad
t=\epsilon T.
\end{equation}
We might expect that the right-hand side should be negligible for quite a long
rescaled time $T$ as long as in the unperturbed system ($\epsilon=0$) pendula
located at nearby values of $x$ follow nearby orbits of the pendulum phase 
portrait.  This will be the case unless $|G(x)|\approx 2$, the condition allowing for nearby values
of $x$ to correspond to topologically dissimilar orbits, leading to large
relative displacements in $u$ over finite $T$ and causing $u_{xx}$ to become
very large very quickly.  See Figure~\ref{fig:sepplot}.
\begin{figure}[h]
\begin{center}
\includegraphics[width=0.3\linewidth]{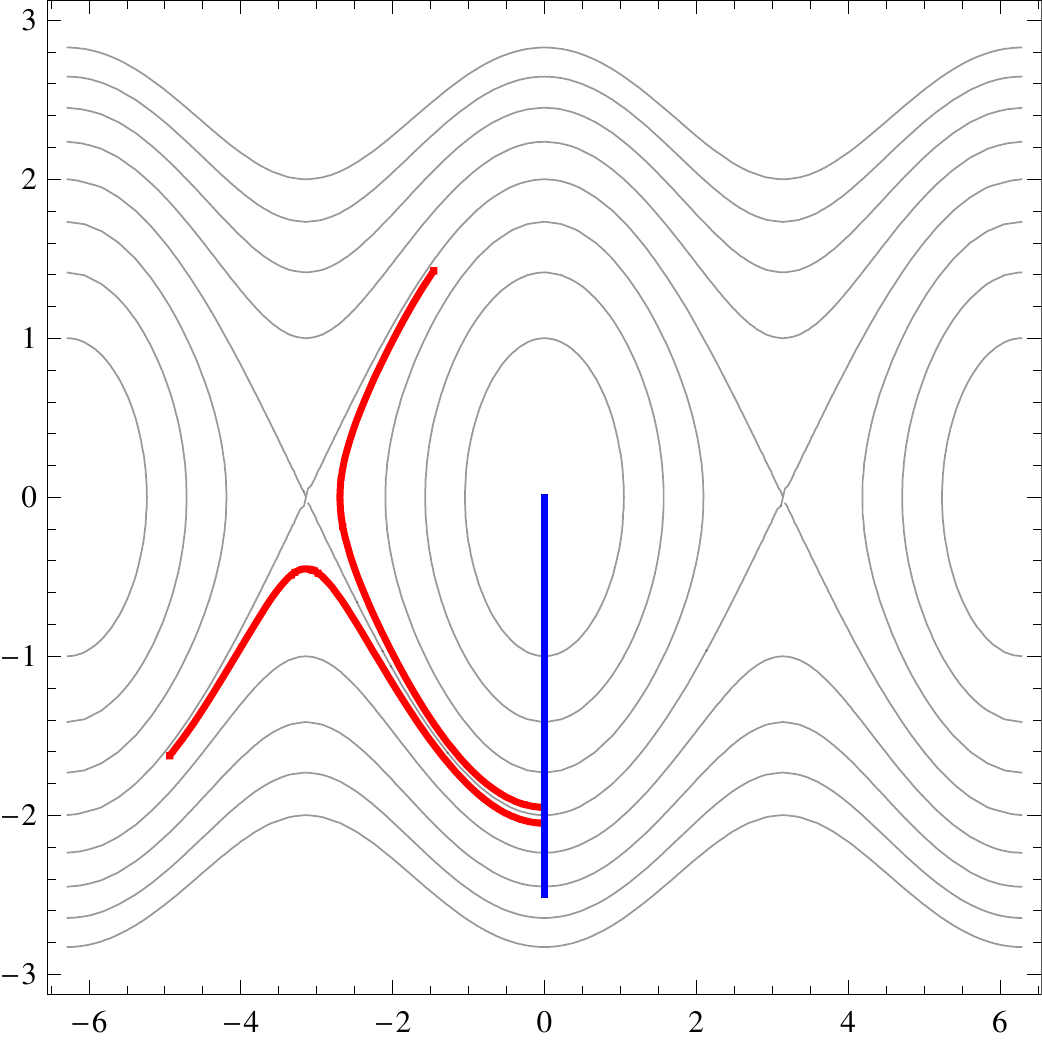}
\end{center}
\caption{\emph{The initial data $u(x,0)=0$ and $u_T(x,0)=G(x)$ with $G(0)<-2$
plotted parametrically in the $(u,u_T)$-plane (blue).  Orbits of the simple
pendulum are shown with gray curves.  Orbits corresponding to nearby values of
$x$ (red curves) can diverge after a descent to a neighborhood of a saddle 
point in finite $T$, if $x\approx \pm x_{\mathrm{crit}}=G^{-1}(-2)$.}}
\label{fig:sepplot}
\end{figure}
To leading order, and for times $t$ of order $\epsilon$ we expect that 
orbits near $x=x_{\mathrm{crit}}$ should follow the pendulum separatrix:
$\cos(u/2)\approx \mathrm{sech}(t/\epsilon)$, 
$\sin(u/2)\approx -\tanh(t/\epsilon)$, and 
$\epsilon u_t\approx -2\,\mathrm{sech}(t/\epsilon)$.
Of course as $T=t/\epsilon\to\infty$, the perturbation term $\epsilon^2u_{xx}$
will become important and will prevent the rapid divergence of trajectories
shown in Figure~\ref{fig:sepplot}.  

The subject of this paper is the asymptotic analysis, in the semiclassical limit $\epsilon\downarrow 0$, of the solution of the Cauchy problem \eqref{eq:sGCauchy} for the sine-Gordon equation
in the case that $t$ is small and 
$x$ is near
a value where the curve $(F,G)=(0,G(x))$ crosses the separatrix (at its midpoint, as in Figure~\ref{fig:sepplot}).  This region of the $(x,t)$-plane is blown up with what turn out to be the correct scalings for easier viewing in Figure~\ref{fig:CriticalZoom}.  
\begin{figure}[h]
\begin{center}
\includegraphics{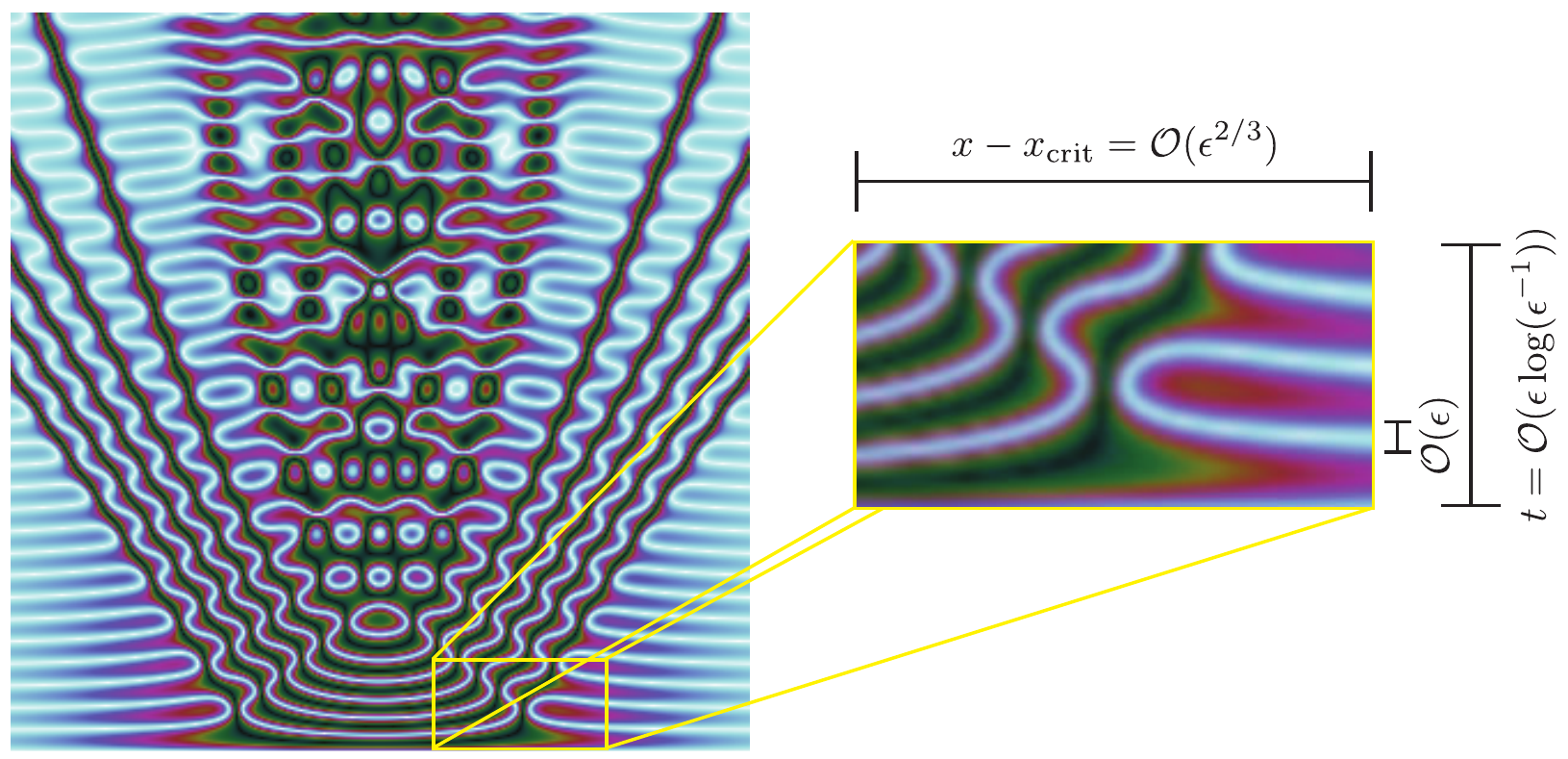}
\end{center}
\caption{\emph{The region of the $(x,t)$-plane in which the fluxon condensate is analyzed.}}
\label{fig:CriticalZoom}
\end{figure}
Therefore we are near the boundary of the regions where
the solution is described in terms of modulated rotational and librational
solutions of the pendulum equation.  We will show that, under suitable further assumptions, the
asymptotic behavior of $u(x,t)$ in this situation is universal, and is described by an essentially
multiscale formula that is expressible in terms of modern special functions, specifically solutions of certain nonlinear ordinary differential equations of Painlev\'e type.

To formulate our results precisely requires that we first set up some background material; we hope that the reader will bear with us until \S\ref{sec:results} where the full details will be presented.  In the meantime, we can describe the semiclassical asymptotics of $u(x,t)$ in the region illustrated
in Figure~\ref{fig:CriticalZoom} by saying that this region contains a curvilinear network of 
isolated kink-type solutions of the sine-Gordon equation (with the approximation $u(x,t)\approx \pi \pmod{2\pi}$ holding in between the kinks) lying along the graphs of certain rational solutions
of inhomogeneous Painlev\'e-II equations.  At  spacetime points associated with poles and zeros of these solutions, the kinks collide in a grazing fashion as can be modeled by a more complicated exact solution of the sine-Gordon equation corresponding to a double soliton.  The phenomenon of an asymptotically universal wave pattern being described by simple ``solitonic'' solutions located in space and time according to a rather more transcendental solution of a Painlev\'e-type equation was only recently observed for the first time
in another problem by Bertola and Tovbis \cite{BertolaT10cusp}.  


\subsection{Notation and terminology.}
All power functions $z^p$ will be assumed to be defined for nonintegral real $p$ as the principal branch with branch cut $z<0$ and with $-|p|\pi<\arg(z^p)<|p|\pi$.

If $G(\cdot)$ is a function of $x$ for which the equation $G(x)=-2$ has a unique solution
for $x>0$, then we denote this solution by $x=x_\mathrm{crit}$.
Let us use the term \emph{criticality} to describe the point $x=x_\mathrm{crit}$
and $t=0$, and \emph{near criticality} to mean $x\approx x_\mathrm{crit}$
and $t\approx 0$.

For any rational function $\mathcal{R}=\mathcal{R}(y)$, let $\mathscr{P}(\mathcal{R})$ denote
the finite set of real poles of $\mathcal{R}$ and let $\mathscr{Z}(\mathcal{R})$ denote the finite
set of real zeros of $\mathcal{R}$.
Finally, for a finite set $\mathscr{S}\subset\mathbb{R}$ and a real number $y\in\mathbb{R}$ denote by
\begin{equation}
|y-\mathscr{S}|:=\min_{y_0\in\mathscr{S}}|y-y_0|
\end{equation}
the distance between $y$ and $\mathscr{S}$.

If $\Sigma$ is an oriented contour in the complex plane and $f$ is a function analytic in the complement of $\Sigma$, we will use subscripts $f_+(\xi)$ and $f_-(\xi)$ to refer to the boundary
values taken by $f(w)$ as $w\to\xi\in\Sigma$ nontangentially from the left and right respectively.

We will make frequent use of the Pauli matrices:
\begin{equation}
\sigma_1:=\begin{bmatrix}0 & 1\\ 1 & 0\end{bmatrix},\quad
\sigma_2:=\begin{bmatrix}0 & -i\\ i & 0\end{bmatrix},\quad
\sigma_3:=\begin{bmatrix}1 & 0 \\ 0 & -1\end{bmatrix}.
\end{equation}
With the sole exception of these three, we write all matrices with boldface capital letters.

We use the Landau notation for most estimates, with ``big-oh'' written $\bo$ and ``little-oh'' written $\lo$.
Also, if $q_1,q_2,\dots,q_m$ are some quantities, we will use the shorthand notation
$\bo(q_1,\dots,q_m)$ to represent a quantity that is bounded by a linear combination
of $|q_1|,\dots,|q_m|$, that is, $\bo(q_1,\dots,q_m)=\bo(q_1)+\cdots +\bo(q_m)$.

We will be dealing with several matrix functions in which the matrix symbol carries subscripts and superscripts, so we will use a special notation for Taylor/Laurent expansion coefficients of such matrix functions:  if $\mathbf{M}_a^b(\sv)$ is such a matrix function and $\sv_0$ is a point about which this matrix is to be expanded, we write the Taylor expansion in the form
\begin{equation}
\mathbf{M}_a^b(\sv)=\sum_{k=0}^\infty [\smash{\fourIdx{\sv_0}{k}{b}{a}{\mathbf{M}}}](\sv-\sv_0)^k.
\end{equation}
We also use this notation in the case that $\sv_0=\infty$ with obvious modifications.

\subsection{Assumptions and definition of fluxon condensates.}
We study the Cauchy initial-value problem \eqref{eq:sGCauchy} under exactly the same assumptions
used in our earlier work \cite{BuckinghamMelliptic}:
\begin{ass} The initial conditions for \eqref{eq:sGCauchy} satisfy $F(x)\equiv 0$.  
\label{ass:impulse}
\end{ass}
Assumption~\ref{ass:impulse} asserts that the initial data is of pure-impulse type.  This is
important because it implies that the direct scattering problem for the solution of the Cauchy
problem by the inverse-scattering transform reduces from the Faddeev-Takhtajan eigenvalue problem to the better-understood nonselfadjoint Zakharov-Shabat eigenvalue problem.
\begin{ass} The function $G$ is a nonpositive function of Klaus-Shaw type, that is,
$G\in L^1(\mathbb{R})\cap C^1(\mathbb{R})$ and $G$ has a unique local (and global)
minimum.
\label{ass:KlausShaw}
\end{ass}
As was shown by Klaus and Shaw in \cite{KlausS02},  Assumption~\ref{ass:KlausShaw} provides a useful and important confinement property of the discrete spectrum of the nonselfadjoint Zakharov-Shabat eigenvalue problem associated to the potential $G$.  This allows us to formulate
Riemann-Hilbert Problem~\ref{rhp:basic} below on a system of contours very close to the unit circle and a transecting negative real interval, a set whose image under the function $E(\cdot)$ defined in \eqref{eq:ED} below is a segment of the imaginary axis.
\begin{ass}
The function $G$ is even in $x$:  $G(-x)=G(x)$, placing the unique minimum of $G$ at $x=0$.
\label{ass:even}
\end{ass}
Assumption~\ref{ass:even} is admittedly less important (we believe that with appropriate modifications our results all hold true without it) but it allows for a substantial simplification
of our analysis.  
We point out that under Assumptions~\ref{ass:KlausShaw} and \ref{ass:even} the function $G$ restricted to $\mathbb{R}_+$ has a well-defined inverse function $G^{-1}:(G(0),0)\to\mathbb{R}_+$.
\begin{ass}
The function $G$ is strictly increasing and real-analytic at each $x>0$, and the positive and
real-analytic function
\begin{equation}
\mathscr{G}(m):=\frac{\sqrt{m}\sqrt{G(0)^2-m}}{2G'(G^{-1}(-\sqrt{m}))},\quad 0<m<G(0)^2
\label{eq:hmdefcrit}
\end{equation}
can be analytically continued to neighborhoods of $m=0$ and $m=G(0)^2$ with 
$\mathscr{G}(0)>0$ and $\mathscr{G}(G(0)^2)>0$.
\label{ass:scriptGnice}
\end{ass}
The analyticity of $G$ for $x>0$ and that of $\mathscr{G}$ up to the endpoints of the interval $[0,G(0)^2]$ as guaranteed by Assumption~\ref{ass:scriptGnice} are both absolutely crucial to our
method of analysis.  It is the analyticity of $G$ that implies that of $\Psi$ defined in \eqref{eq:Psidefinecrit} below, and hence of $\theta_0$ defined by \eqref{eq:ED} and \eqref{eq:theta0defcrit}, and these are used to convert a ``primordial'' Riemann-Hilbert problem of inverse scattering that involves a large number (inversely proportional to $\epsilon>0$) of
pole singularities into the simpler Riemann-Hilbert Problem~\ref{rhp:basic} to be formulated below.  The latter problem has no poles, but only jump discontinuities along contours, and hence is amenable to the Deift-Zhou steepest descent technique \cite{DeiftZ93} of rigorous asymptotic analysis.
\begin{ass}
The small number $\epsilon$ lies in the infinite sequence
\begin{equation}
\epsilon=\epsilon_N:=\frac{\|G\|_1}{4\pi N},\quad N=1,2,3,\dots.
\end{equation}
\label{ass:epsilon}
\end{ass}
Assumption~\ref{ass:epsilon} is important because it minimizes the effect of \emph{spectral singularities}, events occurring infinitely often as $\epsilon\downarrow 0$ at which discrete eigenvalues are born from the continuous spectrum.  When spectral singularities occur, the reflection coefficient has poles in the continuous spectrum, and without Assumption~\ref{ass:epsilon} (or some suitable approximation thereof) the reflection coefficient cannot be neglected uniformly on the continuous spectrum.  Assumption~\ref{ass:epsilon} is important because our approach is based on neglecting the
reflection coefficient entirely.  (It can be shown to be small in the semiclassical limit without Assumption~\ref{ass:epsilon} except near points where spectral singularities can occur.)
\begin{ass}
The function $G(x)$ satisfies $G(0)<-2$.
\label{ass:kinks}
\end{ass}
It is this last assumption that ensures that there exist exactly two points $x=\pm x_\mathrm{crit}$
at which the initial data curve $(F,G)=(0,G(x))$ crosses the separatrix of the simple pendulum phase portrait.  Therefore, Assumption~\ref{ass:kinks} guarantees that the phenomenon we wish
to study in this paper actually occurs for the initial data in question.

Given initial data for the Cauchy problem \eqref{eq:sGCauchy} satisfying these assumptions, we construct a sequence of exact solutions $u=u_N(x,t)$ of the sine-Gordon equation $\epsilon^2 u_{tt}-\epsilon^2u_{xx}+\sin(u)=0$ for $\epsilon=\epsilon_N$, $N=1,2,3,\dots$ called the \emph{fluxon condensate} associated with the given initial data.  While $u_N(x,t)$ is an exact solution of the partial differential equation for each $N$, it generally does not satisfy exactly
the given initial conditions.  However, it has been proved  \cite[Corollary 1.1]{BuckinghamMelliptic}  that $u_N(x,0)=\bo(\epsilon_N)$ holds modulo $4\pi$ and $\epsilon_Nu_{N,t}(x,0)=G(x)+\bo(\epsilon_N)$ where the error estimates are valid pointwise for $x\neq 0$ and $|x|\neq x_\mathrm{crit}$, and uniformly on compact subsets of the set of pointwise accuracy.  We strengthen this convergence to include the points $x$ near where $|x|=x_\mathrm{crit}$ in Theorem~\ref{thm:tzero} below.

The fluxon condensate is
constructed as follows.  First one defines a function $\Psi(\lambda)$ by setting
\begin{equation}
\Psi(\lambda):=\frac{1}{2}\int_0^{G^{-1}(-v)}\sqrt{G(s)^2-v^2}\,ds,\quad\lambda = \frac{iv}{4},
\quad 0<v<-G(0).
\label{eq:Psidefinecrit}
\end{equation}
Note that $\Psi$ is a strictly decreasing function of $v$ where defined.  This fact allows us to
define a sequence of numbers $\{\lambda_{N,k}^0\}_{k=0}^{N-1}$ by solving the equation
\begin{equation}
\Psi(\lambda_{N,k}^0) = \pi\epsilon_N\left(k +\frac{1}{2}\right),\quad k=0,1,2,\dots,N-1.
\label{eq:BScrit}
\end{equation}
(These positive imaginary numbers are approximate eigenvalues for the Zakharov-Shabat eigenvalue problem associated with the Klaus-Shaw potential $G$, and \eqref{eq:BScrit} is a kind of \emph{Bohr-Sommerfeld quantization rule} for that nonselfadjoint problem.)  Setting
\begin{equation}
E(w):=\frac{i}{4}\left[(-w)^{1/2}+(-w)^{-1/2}\right]\quad\text{and}\quad D(w):=\frac{i}{4}\left[(-w)^{1/2}-(-w)^{-1/2}\right],\quad |\arg(-w)|<\pi,
\label{eq:ED}
\end{equation}
we define
\begin{equation}
Q(w;x,t):=E(w)x+D(w)t,\quad |\arg(-w)|<\pi,
\end{equation}
and
\begin{equation}
\Pi_N(w):=\prod_{k=0}^{N-1}\frac{E(w)+\lambda_{N,k}^0}{E(w)-\lambda_{N,k}^0},\quad |\arg(-w)|<\pi.
\end{equation}
This function is meromorphic and has no zeros where defined.  It has $2N$ poles (counted with multiplicity), the set of which we denote as $P_N$.  Generically (with respect to deformations of $G$) the poles are all simple, and there
are $2N_\breather$ of them in complex-conjugate pairs on the unit circle in the $w$-plane,
and $2N_\kink$ of them in pairs on the negative real axis in involution with respect to the map
$w\mapsto 1/w$.  As $N\uparrow\infty$ (or $\epsilon_N\downarrow 0$) the poles accumulate
on the whole unit circle and the interval $[\mathfrak{a},\mathfrak{b}]$, where
\begin{equation}
\mathfrak{a}:= -\frac{1}{4}\left(\sqrt{G(0)^2-4}-G(0)\right)^2,\quad
\mathfrak{b}:=-\frac{1}{4}\left(\sqrt{G(0)^2-4}+G(0)\right)^2=\frac{1}{\mathfrak{a}}.
\end{equation}
Clearly, both $\mathfrak{a}$ and $\mathfrak{b}$ are independent of $N$, and $\mathfrak{a}<-1<\mathfrak{b}<0$.  Were Assumption~\ref{ass:kinks} not satisfied, we would have $N_\kink=0$,
and the poles of $\Pi_N$ would accumulate only in a complex-conjugate symmetric incomplete arc of the
unit circle containing the point $w=1$.  In nongeneric cases it can happen that there is a double pole at $w=-1$.    We consider the accumulation locus of poles, $P_\infty:=[\mathfrak{a},\mathfrak{b}]\cup S^1$ to be an oriented contour, with orientation of the two intervals $(\mathfrak{a},-1)$ and $(-1,\mathfrak{b})$ both toward $w=-1$, and with orientation of the upper
and lower semicircles from $w=-1$ toward $w=1$.  See Figure~\ref{fig:critNcontour}.
We also use the notation
\begin{equation}
\theta_0(w):=\Psi(E(w)).
\label{eq:theta0defcrit}
\end{equation}
Then, define a function $L(w)$ by a Cauchy integral:
\begin{equation}
L(w):=\frac{(-w)^{1/2}}{\pi}\int_{P_\infty}\frac{\theta_0(y)}{(-y)^{1/2}}\frac{dy}{y-w},\quad w\in\mathbb{C}\setminus (P_\infty\cup\mathbb{R}_+),
\label{eq:Ldefine}
\end{equation}
and then set
\begin{equation}
Y_N(w):=\Pi_N(w)e^{-L(w)/\epsilon_N},\quad w\in\mathbb{C}\setminus (P_\infty\cup\mathbb{R}_+).
\end{equation}
Finally, writing
\begin{equation}
\overline{L}(\xi):=\frac{1}{2}(L_+(\xi)+L_-(\xi)),\quad\xi\in P_\infty,
\end{equation}
a function that turns out to have a well-defined analytic continuation to a full neighborhood of the self-intersection point of $w=-1$ of the contour $P_\infty$ (upon continuation from any of the four intersecting arcs), we set
\begin{equation}
T_N(\xi):=2\Pi_N(\xi)\cos(\epsilon_N^{-1}\theta_0(\xi))e^{-\overline{L}(\xi)/\epsilon_N},\quad
\xi\in P_\infty.
\end{equation}
It can be shown \cite[Proposition 3.1]{BuckinghamMelliptic} that $Y_N(w)=1+\bo(\epsilon_N)$ holds uniformly on compact subsets of the open 
domain of definition, and that $Y_N(w)=\bo(1)$ holds if $w\to\mathfrak{a}$ or $w\to\mathfrak{b}$
nontangentially to the real axis.  Similarly, $T_N(\xi)$ extends from $P_\infty$ to an analytic
function $T_N(w)$ on a neighborhood of $w=-1$, and $T_N(w)=1+\bo(\epsilon_N)$
holds uniformly on this neighborhood as well as on $P_\infty$, as long as $w$ is bounded away
from $w=\mathfrak{a}$ and $w=\mathfrak{b}$; however the estimate $T_N(\xi)=\bo(1)$ holds
uniformly for $\xi\in P_\infty$.

Consider the contour $\Sigma_{\mathbf{N}}$ illustrated in Figure~\ref{fig:critNcontour}.
\begin{figure}[h]
\begin{center}
\includegraphics{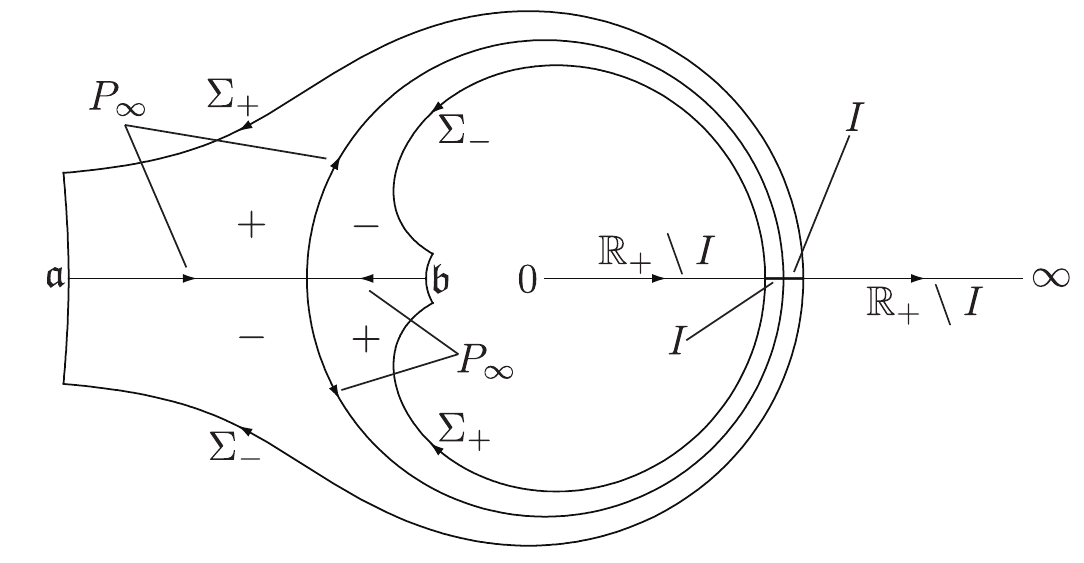}
\end{center}
\caption{\emph{The contour $\Sigma_{\mathbf{N}}$ of discontinuity for the sectionally analytic
function $\mathbf{N}(w)$.  With the exception of the two components of $\mathbb{R}_+\setminus I$ which are oriented left-to-right, the contour arcs are oriented with the regions labeled ``$+$'' on the left, and with the regions labeled ``$-$'' on the right.  The image of $\Sigma_\mathbf{N}\setminus\mathbb{R}_+$ under
$w\mapsto E(w)$ with $E$ defined by \eqref{eq:ED} consists of the four straight line segments:
(i) $\Re\{E\}=-\delta$, $0<\Im\{E\}\le -G(0)/4$, (ii) $\Re\{E\}=0$, $0<\Im\{E\}\le -G(0)/4$, (iii) $\Re\{E\}=\delta$, $0<\Im\{E\}\le -G(0)/4$, and (iv) $|\Re\{E\}|\le\delta$, $\Im\{E\}=-G(0)/4$.  Here $\delta>0$ is a sufficiently small number.  The self-intersection point $w=-1$ is the critical point of $E$.}}
\label{fig:critNcontour}
\end{figure}
Let $g(w)$ denote a function analytic in the domain $w\in\mathbb{C}\setminus(P_\infty\cup\mathbb{R}_+)$, satisfying the symmetry $g(w^*)=g(w)^*$ and the  conditions that $g(0)=0$ and
$g(w)\to 0$ as $w\to\infty$.  We assume also that $g$ takes well-defined boundary values on $P_\infty\cup\mathbb{R}_+$ in the classical sense (H\"older continuity up to the boundary), and that the boundary values satisfy $g_+(\xi)+g_-(\xi)=0$ for $\xi\in\mathbb{R}_+$.  Set
\begin{equation}
\theta(\xi):=-i(g_+(\xi)-g_-(\xi))\quad\text{and}\quad\phi(\xi;x,t):=2iQ(\xi;x,t)+\overline{L}(\xi)-g_+(\xi)-g_-(\xi),\quad\xi\in P_\infty.
\end{equation}
Then, to determine the fluxon condensate, one solves the following Riemann-Hilbert problem.  
\begin{rhp}[Basic Problem of Inverse Scattering]
Given $g$ as above, find a $2\times 2$ matrix function $\mathbf{N}(w)=\mathbf{N}^g_N(w;x,t)$ of the complex variable
$w$ with the following properties:
\begin{itemize}
\item[]\textbf{Analyticity:}  $\mathbf{N}(w)$ is analytic for $w\in\mathbb{C}\setminus\Sigma_{\mathbf{N}}$.
\item[]\textbf{Jump condition:}  $\mathbf{N}(w)$ is H\"older continuous up to $\Sigma_{\mathbf{N}}$.
The boundary values $\mathbf{N}_\pm(\xi)$ taken by $\mathbf{N}(w)$ on the various arcs of $\Sigma_{\mathbf{N}}$ are required to satisfy the following jump conditions 
\begin{equation}
\mathbf{N}_+(\xi)=\sigma_2\mathbf{N}_-(\xi)\sigma_2,\quad\xi\in\mathbb{R}_+\setminus I,
\end{equation}
\begin{equation}
\mathbf{N}_+(\xi)=\sigma_2\mathbf{N}_-(\xi)\sigma_2\begin{bmatrix}1+e^{i(\theta_{0+}(\xi)-\theta_{0-}(\xi))/\epsilon_N} & 
iY_{N-}(\xi)e^{-i\theta_{0-}(\xi)/\epsilon_N}e^{-F(\xi;x,t)/\epsilon_N}\\
-Y_{N+}(\xi)e^{i\theta_{0+}(\xi)/\epsilon_N}e^{F(\xi;x,t)/\epsilon_N} & 1
\end{bmatrix}
,\quad\xi\in I,
\end{equation}
where $F(\xi;x,t):=2iQ_+(\xi;x,t)+L_+(\xi)-2g_+(\xi)$,
\begin{equation}
\mathbf{N}_+(\xi)=\mathbf{N}_-(\xi)\begin{bmatrix}
e^{-i\theta(\xi)/\epsilon_N} & 0\\ -iT_N(\xi)e^{\phi(\xi;x,t)/\epsilon_N} & e^{i\theta(\xi)/\epsilon_N}
\end{bmatrix},\quad \xi\in P_\infty,
\end{equation}
and,
\begin{equation}
\mathbf{N}_+(\xi)=\mathbf{N}_-(\xi)\begin{bmatrix}
1 & 0 \\
-iY_N(\xi)e^{(2iQ(\xi;x,t)+L(\xi)\pm i\theta_0(\xi)-2g(\xi))/\epsilon_N} & 1\end{bmatrix},\quad
\xi\in\Sigma_\pm.
\end{equation}
\item[]\textbf{Normalization:}  The following normalization condition holds:
\begin{equation}
\lim_{w\to\infty}\mathbf{N}(w)=\mathbb{I},
\end{equation}
where the limit is uniform with respect to angle for $|\arg(-w)|<\pi$.
\end{itemize}
\label{rhp:basic}
\end{rhp}
It turns out that for each choice of the function $g$, each $N=1,2,3,\dots$, and each $(x,t)\in\mathbb{R}^2$ there exists a unique solution of this Riemann-Hilbert problem.  Moreover,
the product $\mathbf{M}(w)=\mathbf{M}_N(w;x,t):=\mathbf{N}_N^g(w;x,t)e^{g(w)\sigma_3/\epsilon_N}$ does not depend on the choice of the function $g$.  The solution
$\mathbf{N}^g_N(w;x,t)$ has convergent expansions near $w=0$ and $w=\infty$ of the form
\begin{equation}
\mathbf{N}^g_N(w;x,t)=\mathbb{I} + \sum_{k=1}^\infty[\fourIdx{\infty}{k}{g}{N}{\mathbf{N}}](x,t)(-w)^{-k/2},\quad \text{for $|w|$ sufficiently large,}
\label{eq:Nexpandwlarge}
\end{equation}
and
\begin{equation}
\mathbf{N}^g_N(w;x,t)=\sum_{k=0}^\infty[\fourIdx{0}{k}{g}{N}{\mathbf{N}}](x,t)(-w)^{k/2},\quad 
\text{for $|w|$ sufficiently small.}
\label{eq:Nexpandwsmall}
\end{equation}
\begin{definition}[Fluxon condensates]
Given a function $G$ satisfying Assumptions~\ref{ass:impulse} through \ref{ass:kinks}, 
the fluxon condensate associated with $G$ is the family of functions $\{u_N(x,t)\}_{N=1}^\infty$
defined modulo $4\pi$ by the equations
\begin{equation}
\cos(\tfrac{1}{2}u_N(x,t)):=\left[[\fourIdx{0}{0}{g}{N}{\mathbf{N}}](x,t)\right]_{11}
\quad\text{and}\quad
\sin(\tfrac{1}{2}u_N(x,t)):=\left[[\fourIdx{0}{0}{g}{N}{\mathbf{N}}](x,t)\right]_{21}.
\label{eq:fluxoncondensate}
\end{equation}
Note that these are independent of $g$ because $g(0)=0$.
\label{def:fluxoncondensate}
\end{definition}
While it is possible to extract formulae for derivatives of $u_N(x,t)$ with respect to $x$ and $t$
by the chain rule, it is preferable to have formulae that do not require differentiation of $\mathbf{N}^g_N(w;x,t)$, as this will require control of derivatives of error terms.  However, it can be shown
also that the following formula holds:
\begin{equation}
\epsilon_N\frac{\partial u_N}{\partial t}(x,t)=\left[[\fourIdx{\vphantom{0}\infty}{\vphantom{0}1}{g}{N}{\mathbf{N}}](x,t)\right]_{12}
+\left[[\fourIdx{0}{0}{g}{N}{\mathbf{N}}](x,t)^{-1}[\fourIdx{0}{1}{g}{N}{\mathbf{N}}](x,t)\right]_{12},
\label{eq:epsuNt}
\end{equation}
and this does not require differentiation of $\mathbf{N}^g_N(w;x,t)$ with respect to $t$ (it also turns out to be independent of choice of $g$).
Each function $u_N(x,t)$ of the fluxon condensate is an exact solution of the sine-Gordon equation with $\epsilon=\epsilon_N$, and $u_N(-x,t)=u_N(x,t)$.

While Riemann-Hilbert Problem~\ref{rhp:basic} is the most convenient starting point for our
local analysis near criticality, it is not the most fundamental representation of $u_N(x,t)$.  Indeed,
$u_N(x,t)$ is a \emph{reflectionless potential} for the direct scattering problem of the Lax pair
for the sine-Gordon equation, and this means that it can be obtained from a purely ``discrete''
Riemann-Hilbert problem whose solution is a rational function on the Riemann surface of $y^2=-w$, with poles on both sheets over the points of $P_N$.  Since the number of poles is increasing with $N$, some preparations are required to recast the problem in a suitable form for addressing the limit $N\to\infty$.  These preparations are detailed in \cite{BuckinghamMelliptic},
and they take the form of a sequence of explicit transformations resulting in the equivalent 
Riemann-Hilbert problem~\ref{rhp:basic}.  In general, a number of choices are made along the
way because the transformations that are required to enable the subsequent asymptotic analysis
depend on $(x,t)$; however for $x\approx x_\mathrm{crit}>0$ and $t$ small only the simplest
of the choices detailed in \cite{BuckinghamMelliptic} are required.  For readers familiar with the terminology of that paper, we are assuming that $\Delta=\emptyset$, and correspondingly,
$Y_N(w)$ is the function called $Y^\nabla(w)$ in \cite{BuckinghamMelliptic} while $T_N(w)$ is
the function called $T^\nabla(w)$ in \cite{BuckinghamMelliptic}.  Also, for our purposes we need
make no distinction between the contour called $\Sigma^\nabla\cup\Sigma^\Delta$ (which would just be $\Sigma^\nabla$ in the case that $\Delta=\emptyset$) and the contour $P_\infty$, and this implies that the set $Z\subset\mathbb{C}$ that lies between these two contours is empty.
Finally, to derive Riemann-Hilbert Problem~\ref{rhp:basic} from the results of \cite{BuckinghamMelliptic} we used the facts that for $\xi\in I$, $Q_+(\xi;x,t)+Q_-(\xi;x,t)=L_+(\xi)+L_-(\xi)=0$.

The solution of the sine-Gordon equation by inverse-scattering methods is a subject with a
long history.  When the sine-Gordon equation is written in terms of characteristic coordinates 
$x\pm t$ it fits naturally into the hierarchy of the Ablowitz-Kaup-Newell-Segur or Zakharov-Shabat scattering problem.  The \emph{characteristic} Cauchy problem for the sine-Gordon equation
was integrated in this framework by Ablowitz, Kaup, Newell, and Segur \cite{AblowitzKNS73}
and by Zakharov, Takhtajan, and Faddeev \cite{ZakharovTF74}.  The more physically-relevant
problem of the Cauchy problem in laboratory coordinates as posed in \eqref{eq:sGCauchy}
required new methodology, and the solution of this Cauchy problem by the inverse scattering
method was first outlined by Kaup \cite{Kaup75}.  An account of the solution of the Cauchy problem in laboratory coordinates is given in the textbook of Faddeev and Takhtajan
\cite{FaddeevT87}.  Some further analytical details needed to make the theory completely
rigorous were supplied by Zhou \cite{Zhou95} and later by Cheng, Venakides, and Zhou
\cite{ChengVZ99}.  A self-contained account of the Riemann-Hilbert formulation of the inverse-scattering solution of the Cauchy problem \eqref{eq:sGCauchy} assuming only that at each
instant of time the solution has $L^1$-Sobolev regularity can be found in our paper 
\cite[Appendix A]{BuckinghamM08}, and a direct proof that the sine-Gordon equation preserves this degree
of regularity if it is initially present is given in \cite[Appendix B]{BuckinghamM08}.  In our recent paper
\cite{BuckinghamMelliptic} we found that to describe the semiclassical limit for the Cauchy problem \eqref{eq:sGCauchy} it is useful to reformulate the Riemann-Hilbert problem in the complex plane of a square root of the spectral variable used in \cite[Appendix A]{BuckinghamM08}; this makes it easier to express the asymptotic solutions in terms of
Riemann theta functions of the lowest possible genus.  As we view the current paper as
a continuation of our work in \cite{BuckinghamMelliptic} we use the same formulation here.

\subsection{Statement of results.}
\label{sec:results}
Let $x_\mathrm{crit}$ be defined by
\begin{equation}
x_\mathrm{crit}:=G^{-1}(-2)>0,
\label{eq:xcritdefinecrit}
\end{equation}
and define a positive constant $\nu$ by
\begin{equation}
\nu:=\frac{1}{12G'(x_\mathrm{crit})}>0.
\label{eq:nudefinecrit}
\end{equation}
Also, set
\begin{equation}
\Delta x:=x-x_\mathrm{crit}.
\label{eq:deltax}
\end{equation}
All of our results concern the asymptotic behavior of the fluxon condensate $u_N(x,t)$ in the small region near criticality where $t = \bo(\epsilon_N\log(\epsilon_N^{-1}))$ and $\Delta x = \bo(\epsilon_N^{2/3})$ as shown in Figure~\ref{fig:CriticalZoom}.

Our first result is concerned with the relevance of the fluxon condensate associated with $G(x)$
to the Cauchy initial-value problem.
\begin{theorem}[Initial accuracy of fluxon condensates]
Suppose that $x\pm x_\mathrm{crit}=\bo(\epsilon_N^{2/3})$.  Then uniformly for such $x$,
\begin{equation}
u_N(x,0)= \bo(\epsilon_N^{1/3}) \pmod{4\pi} \quad\text{and}\quad
\epsilon_N\frac{\partial u_N}{\partial t}(x,0)=G(x) +\bo(\epsilon_N^{1/3}).
\end{equation}
\label{thm:tzero}
\end{theorem}
This result extends that of \cite[Corollary~1.1]{BuckinghamMelliptic} to suitable neighborhoods
of the points $x=\pm x_\mathrm{crit}$.

\begin{theorem}[Main approximation theorem]
There exist multiscale approximating functions $\dot{C}(x,t;\epsilon_N)$ and $\dot{S}(x,t;\epsilon_N)$ (defined in detail
in \S\ref{sec:ProofMain}, and depending on initial data only through the constant $\nu>0$) such that
\begin{equation}
\begin{split}
\cos(\tfrac{1}{2}u_N(x,t))&=\dot{C}(x,t;\epsilon_N)+\bo(\epsilon_N^{1/6})\\
\sin(\tfrac{1}{2}u_N(x,t))&=\dot{S}(x,t;\epsilon_N)+\bo(\epsilon_N^{1/6})
\end{split}
\end{equation}
with the error terms being uniform for $\Delta x=\bo(\epsilon_N^{2/3})$ and $t=\bo(\epsilon_N\log(\epsilon_N^{-1}))$.
\label{thm:main}
\end{theorem}
In fact, we will really show that the error term is significantly smaller, namely $\bo(\epsilon_N^{1/3})$, over most of the small region of the $(x,t)$-plane where the above Theorem provides an asymptotic
description of the dispersive breakup of the pendulum separatrix.  Now, the multiscale model provided by Theorem~\ref{thm:main} serves to establish universality of the behavior near the critical point,
but the formulae for $\dot{C}(x,t;\epsilon_N)$ and $\dot{S}(x,t;\epsilon_N)$ are rather complicated, so it is useful
to give some more detailed information by focusing on smaller parts of the $(x,t)$-plane near
criticality.  

To render our results in a more elementary fashion, we need to first define a certain hierarchy of rational functions.  First define
\begin{equation}
\pu_0(z):=1\quad\text{and}\quad \pv_0(z):=-\frac{1}{6}z.
\end{equation}
Then define $\{(\pu_\ind,\pv_\ind)\}_{\ind\in\mathbb{Z}}$ by the forward recursion
\begin{equation}
\pu_{\ind+1}(z):=-\frac{1}{6}z\pu_\ind(z)-\frac{\pu_\ind'(z)^2}{\pu_\ind(z)} +\frac{1}{2}\pu_\ind''(z)\quad\text{and}\quad\pv_{\ind+1}(z):=\frac{1}{\pu_\ind(z)}
\label{eq:Baecklundplusintro}
\end{equation}
and the backward recursion
\begin{equation}
\pu_{m-1}(z):=\frac{1}{\pv_\ind(z)}\quad\text{and}\quad
\pv_{m-1}(z):=\frac{1}{2}\pv_\ind''(z)-\frac{\pv_\ind'(z)^2}{\pv_\ind(z)}-\frac{1}{6}z\pv_\ind(z).
\label{eq:Baecklundminusintro}
\end{equation}
Up to constant factors, $\pu_\ind(z)$ and $\pv_\ind(z)$ are ratios of consecutive \emph{Yablonskii-Vorob'ev} polynomials (see \cite{Clarkson03,ClarksonM03}) and logarithmic derivatives of $\pu_\ind(z)$ and $\pv_\ind(z)$ are the unique \cite{Murata85} rational solutions of the inhomogeneous
Painlev\'e-II equations.  These observations are not necessary for us to state our results, and
we will make further comments later at an appropriate point in the paper.
The following proposition 
characterizes the behavior of these rational functions for large $z$.
\begin{proposition}
For each $\ind\in\mathbb{Z}$,
\begin{equation}
\pu_\ind(z)=\left(-\frac{z}{6}\right)^\ind(1+\bo(z^{-1}))\quad\text{and}\quad
\pv_\ind(z)=\left(-\frac{z}{6}\right)^{1-\ind}(1+\bo(z^{-1}))
\label{eq:unvnasympy}
\end{equation}
as $z\to\infty$.  In particular, $\mathrm{sgn}(\pu_\ind(z))=\mathrm{sgn}(\pv_\ind(z))=1$ for sufficiently large negative $z$ while
$\mathrm{sgn}(\pu_\ind(z))=-\mathrm{sgn}(\pv_\ind(z))=(-1)^\ind$ for sufficiently large positive $z$.
\label{prop:unvnasympy}
\end{proposition}
\begin{proof}
It is obvious that the recursions \eqref{eq:Baecklundplusintro} and \eqref{eq:Baecklundminusintro} preserve rationality of the input functions $(\pu_\ind,\pv_\ind)$, so for each
$\ind\in\mathbb{Z}$, $\pu_\ind$ is a rational function of $z$.  Therefore, $\pu_\ind$ has a Laurent expansion
for large $z\in\mathbb{C}$ of the form
\begin{equation}
\pu_\ind(z)=c_\ind z^{k_\ind}(1+\bo(z^{-1}))
\end{equation}
where $c_\ind$ is a complex constant and $k_\ind\in\mathbb{Z}$ is an exponent to be determined, and moreover this expansion is differentiable any number of times with respect to $z$.  In particular, it follows that 
\begin{equation}
\frac{\pu_\ind'(z)^2}{\pu_\ind(z)} = k_\ind^2c_\ind z^{k_\ind-2}(1+\bo(z^{-1}))\quad\text{and}\quad
\pu_\ind''(z)=k_\ind(k_\ind-1)c_\ind z^{k_\ind-2}(1+\bo(z^{-1}))
\end{equation}
as $z\to\infty$.  This shows that the final two terms on the right-hand side of the formula for $\pu_{\ind+1}(z)$ given in \eqref{eq:Baecklundplusintro}
are subdominant compared to the term $-z\pu_\ind(z)/6$, and therefore,
\begin{equation}
\pu_{\ind+1}(z)=-\frac{z}{6}c_\ind z^{k_\ind}(1+\bo(z^{-1}))=-\frac{z}{6}\pu_\ind(z)(1+\bo(z^{-1})),\quad z\to\infty.
\end{equation}
This formula gives the recurrence relation for the large-$z$ asymptotics of $\pu_\ind$.  The asymptotic formula
for $\pu_\ind(z)$ given in \eqref{eq:unvnasympy} follows by solving the recurrence with the base case $\pu_0(z)=1$, and that for $\pv_\ind(z)$ follows from the identity $\pv_\ind(z)=1/\pu_{\ind-1}(z)$.
\end{proof}

\begin{theorem}[Superluminal kink asymptotics]
Fix an integer $\ind$, and suppose that $\Delta x = \bo(\epsilon_N^{2/3})$ while 
$|t-\tfrac{2}{3}\ind\epsilon_N\log(\epsilon_N^{-1})|\le\tfrac{1}{3}\epsilon_N\log(\epsilon_N^{-1})$.  
Then 
\begin{equation}
\begin{split}
\cos(\tfrac{1}{2}u_N(x,t))&=(-1)^\ind\,\mathrm{sgn}(\pu_\ind(z))\,\mathrm{sech}(T_\kink) + E_{\cos}(x,t;\epsilon_N)\\
\sin(\tfrac{1}{2}u_N(x,t))&=(-1)^{\ind+1}\tanh(T_\kink)+E_{\sin}(x,t;\epsilon_N),
\end{split}
\end{equation}
where the phase is
\begin{equation}
T_\kink:=\frac{t}{\epsilon_N}-2\ind\log\left(\frac{4\nu^{1/3}}{\epsilon_N^{1/3}}\right) +
\log|\pu_\ind(z)|,
\label{eq:kinkTdefine}
\end{equation}
and a rescaled spatial coordinate is given by
\begin{equation}
z:=\frac{\Delta x}{2\nu^{1/3}\epsilon_N^{2/3}}.
\label{eq:zDeltax}
\end{equation}
The error terms satisfy
\begin{equation}
\lim_{\epsilon_N\downarrow 0}E_{\cos}(x,t;\epsilon_N) = \lim_{\epsilon_N\downarrow 0}
E_{\sin}(x,t;\epsilon_N) = 0
\end{equation}
unless (i) both $\epsilon_N^{-2\ind/3}e^{-t/\epsilon_N}=\bo(\epsilon_N^{1/3})$ and 
$|z-\mathscr{Z}(\pu_\ind)|=\bo(\epsilon_N^{1/3})$ or (ii) both $\epsilon_N^{2\ind/3}e^{t/\epsilon_N}=\bo(\epsilon_N^{1/3})$ and $|z-\mathscr{P}(\pu_\ind)|=\bo(\epsilon_N^{1/3})$.
Moreover, given any interval $[z_-,z_+]$ on which $\log |\pu_\ind(z)|$ is a bounded
function, we have
\begin{equation}
E_{\cos}(x,t;\epsilon_N)=\bo(\epsilon_N^{1/3})\quad\text{and}
\quad
E_{\sin}(x,t;\epsilon_N)=\bo(\epsilon_N^{1/3})\quad
\end{equation}
holding uniformly for 
$z\in[z_-,z_+]$ and $|t-\tfrac{2}{3}\ind\epsilon_N\log(\epsilon_N^{-1})|\le\tfrac{1}{3}\epsilon_N\log(\epsilon_N^{-1})$.
\label{thm:kinkapprox}
\end{theorem}

Wherever the error terms vanish in the limit, it follows that $\cos(u_N(x,t))\to 2\,\mathrm{sech}^2(T_\kink)-1$ and that $\sin(u_N(x,t))\to-2\sigma\,\mathrm{sech}(T_\kink)\tanh(T_\kink)$, where $\sigma:=\mathrm{sgn}(\pu_\ind(z))$.  If 
we define a function $u(T)$
modulo $2\pi$ by 
\begin{equation}
\label{superluminal-antikink-formula}
\begin{split}
\cos(u(T)) & := 2\,\mathrm{sech}^2(T)-1 \\
\sin(u(T)) & := -2\sigma\,\mathrm{sech}(T)\tanh(T),
\end{split}
\end{equation}
then it is easy to check 
that $u(T)$ is an $X$-independent solution of the unscaled sine-Gordon equation
\begin{equation}
\frac{\partial^2u}{\partial T^2}-\frac{\partial^2u}{\partial X^2}+\sin(u)=0.
\label{eq:unscaledsG}
\end{equation}
\begin{figure}[h]
\begin{center}
\includegraphics{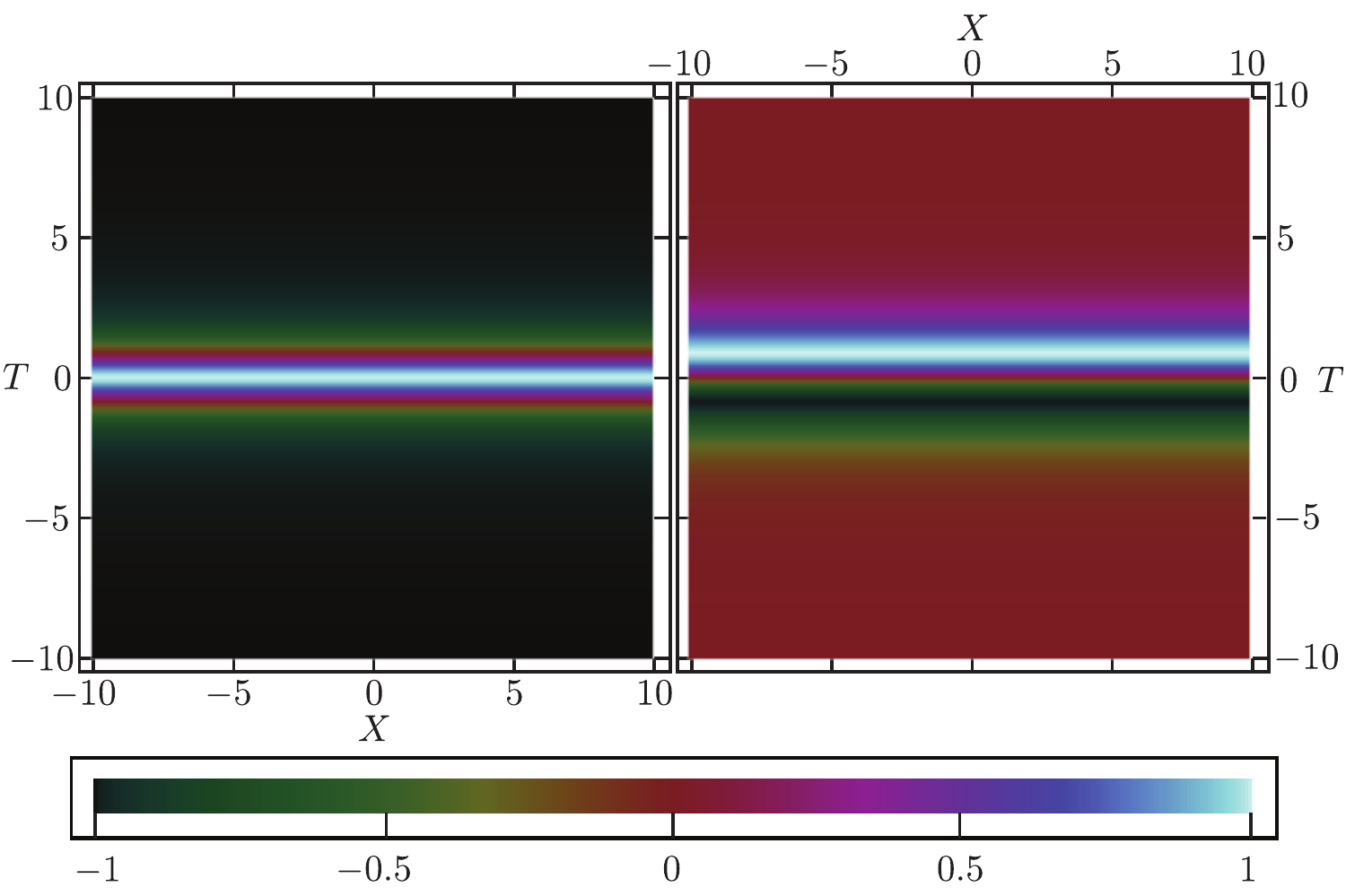}
\end{center}
\caption{\emph{(Superluminal antikink).  The cosine (left) and sine (right) of $u(X,T)=u(T)$ defined by \eqref{superluminal-antikink-formula} with $\sigma=1$.  For a superluminal kink ($\sigma=-1$), the plots are simply reflected through the horizontal $X$-axis. }}
\label{fig:KinkPicture}
\end{figure}
This exact solution represents a superluminal (infinite velocity) kink with unit magnitude \emph{topological charge}  $\sigma$; if $\sigma=1$ (respectively $\sigma=-1$) then $u$ decreases (respectively increases) by $2\pi$ as $T$ varies from $T=-\infty$ to $T=+\infty$.  Sometimes one refers to $u(T)$ as a kink in the case $\sigma=-1$ and as  an \emph{antikink} in the case $\sigma=1$.  Another way to describe $u(T)$
is to say that it is a solution of the simple pendulum equation $u''+\sin(u)=0$ that is homoclinic to
the unstable equilibrium of a stationary inverted pendulum.  See Figure~\ref{fig:KinkPicture}.

The dominant part of the phase variable $T_\kink$ defined by \eqref{eq:kinkTdefine} is indeed a recentering and rescaling by $\epsilon_N$ of $t$; however $T_\kink$ also contains weak $x$-dependence through the 
function $\log|\pu_\ind(z)|$ where $z$ is proportional to $x-x_\mathrm{crit}$.  Therefore, the approximation of $u_N(x,t)$ described in Theorem~\ref{thm:kinkapprox} is not an exact kink, but rather is one that is slowly modulated in the direction parallel to the wavefront.  Indeed, the center of the approximating kink (where the pendulum angle is zero) corresponds to $T_\kink=0$, which is a curve in
the $(x,t)$-plane that is a scaled and translated version of the graph of the function $-\log|\pu_\ind(z)|$.
See Figure~\ref{fig:kinkcurves}.  An additional important observation is that the period of each approximating kink is proportional to $\epsilon_N$, but the strip in which it lives (indexed by the integer $\ind$) corresponds to a time interval of length proportional to $\epsilon_N\log(\epsilon_N^{-1})$; hence the kinks are in reality widely separated from each other, and therefore
in ``most'' of the domain $\Delta x =\bo(\epsilon_N^{2/3})$ and $t=\bo(\epsilon_N\log(\epsilon_N^{-1}))$ covered by Theorem~\ref{thm:kinkapprox} it is fair to say that the approximate formula $u_N(x,t)\approx \pi\pmod{2\pi}$ holds, which in the context of the coupled pendulum interpretation of sine-Gordon means that ``most'' of the pendula are approximately in the unstable inverted configuration.
\begin{figure}[h]
\begin{center}
\includegraphics[height=1.8 in]{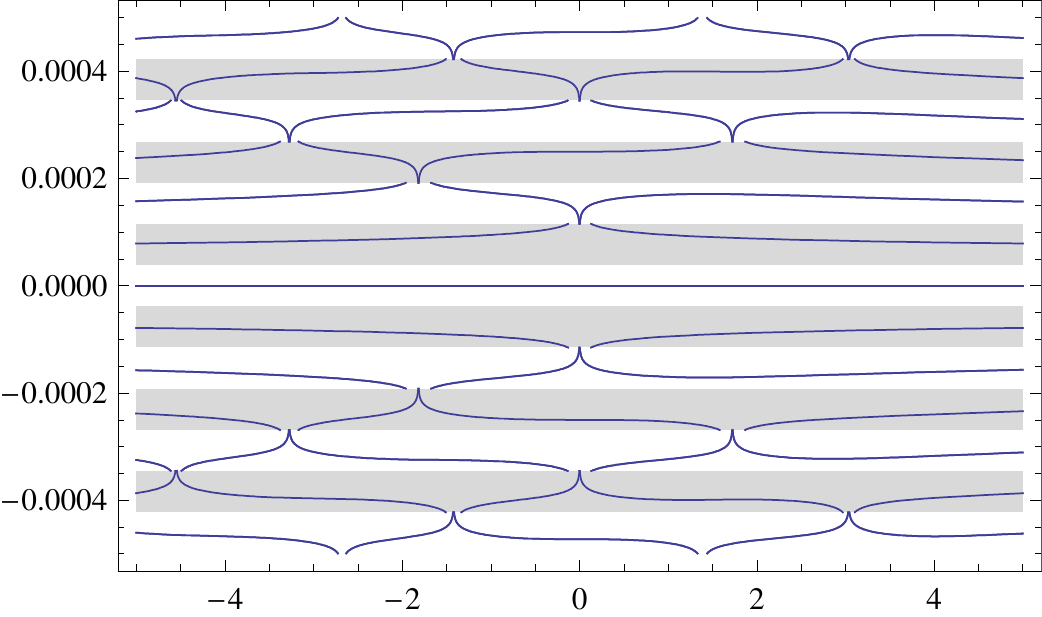}\hfill%
\includegraphics[height=1.8 in]{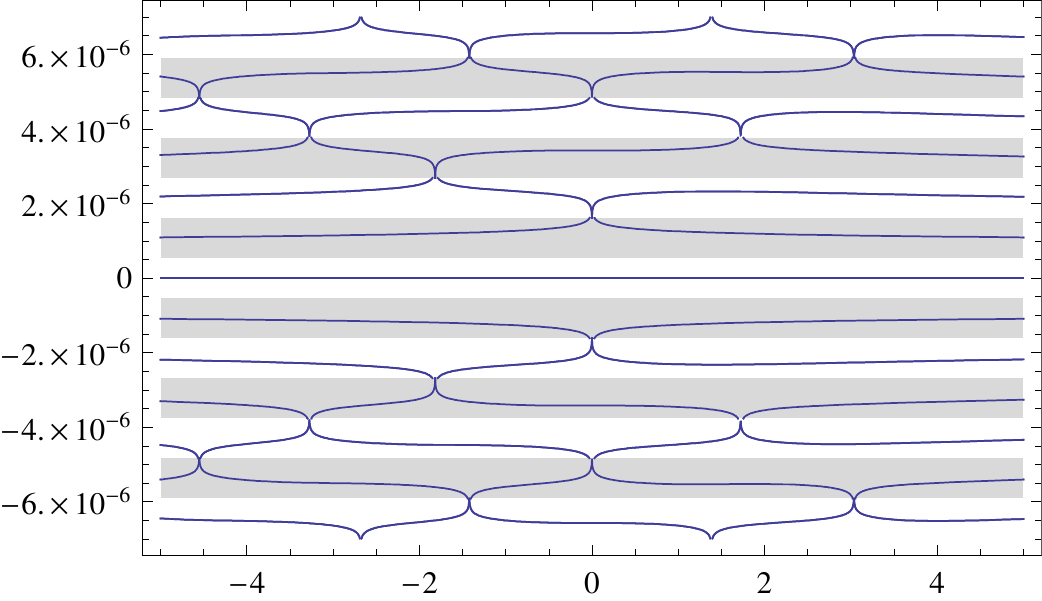}
\end{center}
\caption{\emph{The horizontal strips in the $(z,t)$-plane in each of which Theorem~\ref{thm:kinkapprox} provides asymptotics for $u_N(x,t)$, shown with alternating gray and white shading, with the center $T_\kink=0$ of each corresponding superluminal kink shown with a superimposed curve.  Left:  $\epsilon_N=10^{-5}$.  Right:  $\epsilon_N=10^{-7}$.  In both
cases $4\nu^{1/3}=1$.  For different values of $\nu>0$, the curves can ``march out'' of the confining strips as $|t|$ increases, unless $\epsilon_N$ is small enough.}}
\label{fig:kinkcurves}
\end{figure}

We may now formulate some observations about the kink centered at $T_\kink=0$ that follow
from Proposition~\ref{prop:unvnasympy}.  Indeed, for each
$\ind\in\mathbb{Z}$ the following is true: for $z$ sufficiently negative, $\sigma:=\mathrm{sgn}(\pu_\ind(z))=1$, while for $z$ sufficiently positive, $\sigma=(-1)^\ind$.  This then implies that as $z\to -\infty$
we have a series of timelike antikinks (all of the same topological charge) consistent with nearly synchronous rotational motion of pendula.  Similarly, as $z\to +\infty$
we have alternation between kink and antikink with no net topological charge consistent with nearly synchronous librational motion of pendula.  These facts demonstrate that the asymptotic behavior near criticality matches appropriately with the established asymptotic formulae \cite{BuckinghamMelliptic} valid away from criticality that model $u_N(x,t)$ as a modulated train of superluminal rotational traveling waves to the left of $x_\mathrm{crit}$ and as a modulated train of superluminal librational traveling waves to the right of $x_\mathrm{crit}$.  It also follows from Proposition~\ref{prop:unvnasympy} that if $\ind>0$ then $\pu_\ind(z)$ blows up as $z\to \pm\infty$ while if $\ind<0$ then $\pu_\ind(z)\to 0$
as $z\to\pm\infty$.  This implies that for $\ind>0$ the corresponding kinks 
follow a logarithmic ``frown'' for large $|z|$, while for $\ind<0$ they instead follow a logarithmic ``smile'' for large $|z|$.  These features can be seen in Figure~\ref{fig:kinkcurves}.

When the horizontal strips of the $(x,t)$-plane corresponding to different integral values of $m$ are put together, one sees that the kink-like asymptotics given by Theorem~\ref{thm:kinkapprox} are valid throughout the region where $\Delta x = \bo(\epsilon_N^{2/3})$ and $t = \bo(\epsilon_N\log(\epsilon_N^{-1}))$ with the sole exception of small sub-regions near the top and bottom edges of each strip where the curve $T_\kink=0$ tries to exit the strip.  When the strips are put together
there are obvious mismatches of the curves from neighboring strips in these small sub-regions
as can be seen from the left-hand plot in Figure~\ref{fig:kinkcurves}.  Our final result corrects the
kink asymptotics in these regions and therefore removes the mismatches.

\begin{theorem}[Grazing kink collisional asymptotics]
Fix an integer $\ind$ and let $z_0$ denote any of the (necessarily simple) real zeros of $\pu_{\ind-1}(z)$ (equivalently a simple zero of $\pv_{\ind+1}(z)$).  Given a sufficiently small positive number $\mu>0$,
suppose that $|t-(\tfrac{2}{3}\ind-\tfrac{1}{3})\epsilon_N\log(\epsilon_N^{-1})|\le \tfrac{1}{3}\epsilon_N\log(\epsilon_N^{-1})$, and that 
$|z-z_0|\le\mu \epsilon_N^{1/6} \exp(|t-(\tfrac{2}{3}m-\tfrac{1}{3})\epsilon_N\log(\epsilon_N^{-1})|/(2\epsilon_N))$, where $z$ is defined by \eqref{eq:zDeltax}.
%
Then
\begin{equation}
\begin{split}
\cos(\tfrac{1}{2}u_N(x,t))&= (-1)^{\ind-1}\,\mathrm{sgn}(\pu_{\ind-1}'(z_0))\,\frac{2X_\grazing\mathrm{sech}(T_\grazing)}{1+X_\grazing^2\mathrm{sech}^2(T_\grazing)}+E_{\cos}(x,t;\epsilon_N)\\
\sin(\tfrac{1}{2}u_N(x,t))&= (-1)^{\ind-1}\frac{1-X_\grazing^2\mathrm{sech}^2(T_\grazing)}{1+X_\grazing^2\mathrm{sech}^2(T_\grazing)}+E_{\sin}(x,t;\epsilon_N),
\end{split}
\label{eq:grazingcossinuover2}
\end{equation}
where 
\begin{equation}
X_\grazing:=2\left(\frac{\nu}{\epsilon_N}\right)^{1/3}\left(\frac{\Delta x}{2\nu^{1/3}\epsilon_N^{2/3}}-z_0\right)\quad\text{and}\quad
T_\grazing:=\frac{t}{\epsilon_N}-(2\ind-1)\log\left(\frac{4\nu^{1/3}}{\epsilon_N^{1/3}}\right) +
\log|\pu_{\ind-1}'(z_0)|.
\label{eq:grazingXT}
\end{equation}
The error terms satisfy 
\begin{equation}
\lim_{\epsilon_N\downarrow 0}E_{\cos}(x,t;\epsilon_N)=\lim_{\epsilon_N\downarrow 0}E_{\sin}(x,t;\epsilon_N)=0
\end{equation}
whenever $|z-z_0|\ll 1$ as $\epsilon_N\downarrow 0$ (which excludes only the four extreme corners of the hourglass-shaped region under consideration).  Moreover, we have
\begin{equation}
E_{\cos}(x,t;\epsilon_N)=\bo(\epsilon_N^{1/3})\quad\text{and}\quad
E_{\sin}(x,t;\epsilon_N)=\bo(\epsilon_N^{1/3})
\end{equation}
holding uniformly for $|z-z_0|=\bo(\epsilon_N^{1/3})$ and $|t-(\tfrac{2}{3}m-\tfrac{1}{3})\epsilon_N\log(\epsilon_N^{-1})|\le\tfrac{1}{3}\epsilon_N\log(\epsilon_N^{-1})$.
\label{thm:grazingapprox}
\end{theorem}

Wherever the error terms vanish in the limit, it follows from \eqref{eq:grazingcossinuover2} that
$\cos(u_N(x,t))\to\cos(u(X_\grazing,T_\grazing))$ and $\sin(u_N(x,t))\to\sin(u(X_\grazing,T_\grazing))$,
where with $\kappa:=\mathrm{sgn}(\pu_{\ind-1}'(z_0))$, a function $u(X,T)$ is defined modulo $2\pi$ by
\begin{equation}
\label{grazing-kink-formula}
\begin{split}
\cos(u(X,T))&:=\frac{8X^2\mathrm{sech}^2(T)}{(1+X^2\mathrm{sech}^2(T))^2} -1\\
\sin(u(X,T))&:=4\kappa X\mathrm{sech}(T)\frac{1-X^2\mathrm{sech}^2(T)}{(1+X^2\mathrm{sech}^2(T))^2}.
\end{split}
\end{equation}
Then it is again easy to confirm that $u(X,T)$ is an exact solution of the unscaled sine-Gordon
equation \eqref{eq:unscaledsG}. 
Plots of $\cos(u(X,T))$ and $\sin(u(X,T))$ are shown in Figure~\ref{fig:GrazingCollision}.
\begin{figure}[h]
\begin{center}
\includegraphics{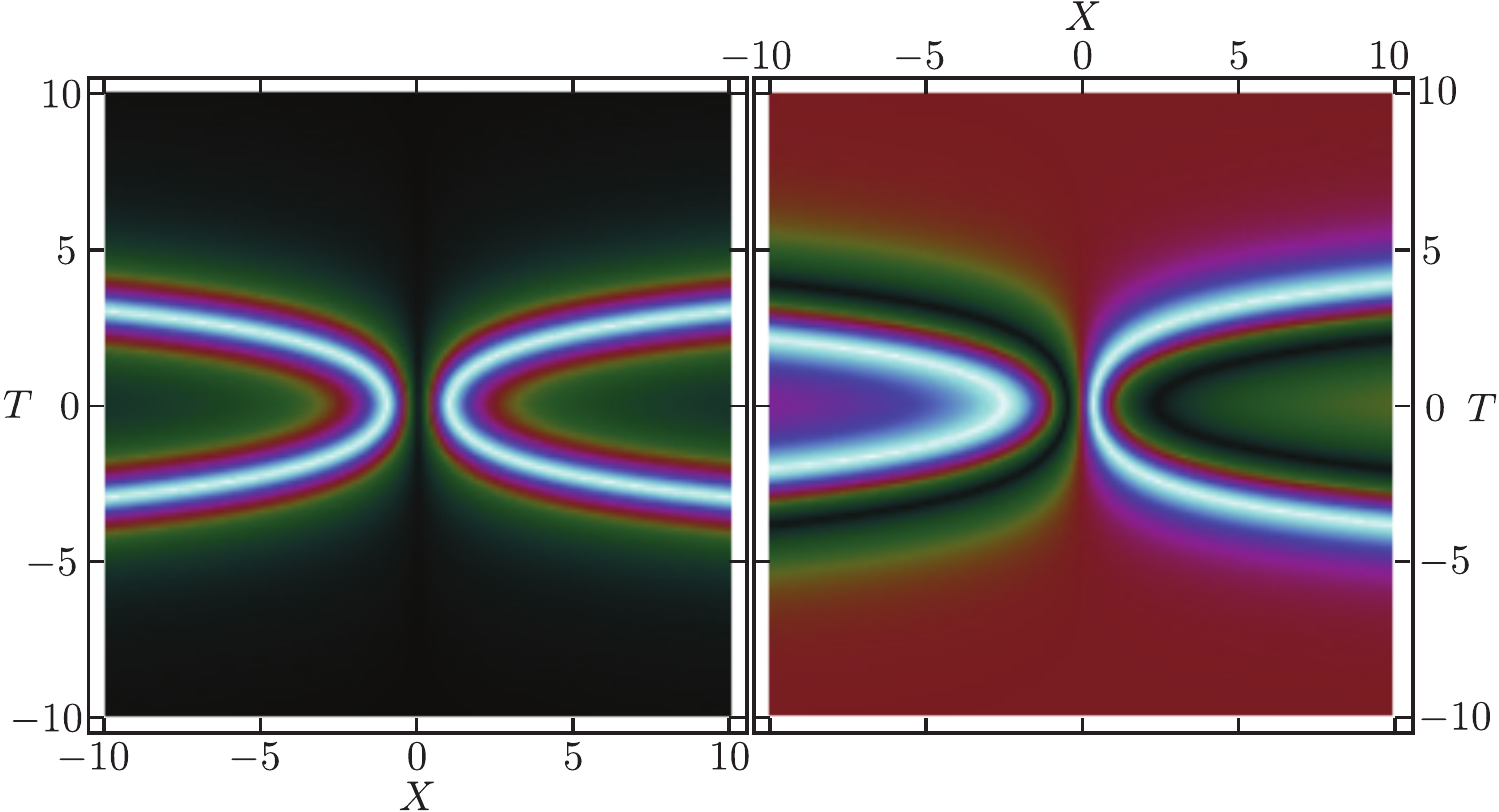}
\end{center}
\caption{\emph{(Grazing collision of superluminal kinks).  The cosine (left) and sine (right) of $u(X,T)$ defined by \eqref{grazing-kink-formula} with $\kappa=1$.  For $\kappa=-1$ the plots are simply reflected through the vertical $T$-axis. }}
\label{fig:GrazingCollision}
\end{figure}
This is a particular solution of the sine-Gordon equation that corresponds to boundary conditions of the form
$u\to \pi\pmod {2\pi}$ as $T\to\pm\infty$.  In the proper version of scattering theory corresponding to 
these boundary conditions, the solution at hand is a reflectionless potential associated to a \emph{double
eigenvalue}.  In the context of the Riemann-Hilbert problem of inverse scattering such an object is
encoded as a double pole of the matrix unknown.  Such solutions were noted in the earliest days of inverse scattering by Zakharov and Shabat \cite{ZS} in their study of the focusing nonlinear Schr\"odinger (NLS) equation.  They pointed out
that such solutions describe \emph{grazing collisions} of solitons.  The trajectories of the solitons
emerging from the interaction region are not asymptotically straight lines, but rather are logarithmic,
with relative velocities of the solitons tending to zero.  In the context of the sine-Gordon equation we have instead a grazing collision of superluminal kinks.  These structures serve to smooth out
the mismatches of curves shown in Figure~\ref{fig:kinkcurves} at horizontal strip boundaries.  In the language of matched asymptotic expansions, they function as internal transition layers of ``corner'' type.

The regions of validity of the asymptotic formulae given in Theorems~\ref{thm:kinkapprox}
and \ref{thm:grazingapprox} are compared in Figure~\ref{fig:RegionsSchematic}, which illustrates the complementary nature of these two results.
\begin{figure}[h]
\begin{center}
\includegraphics{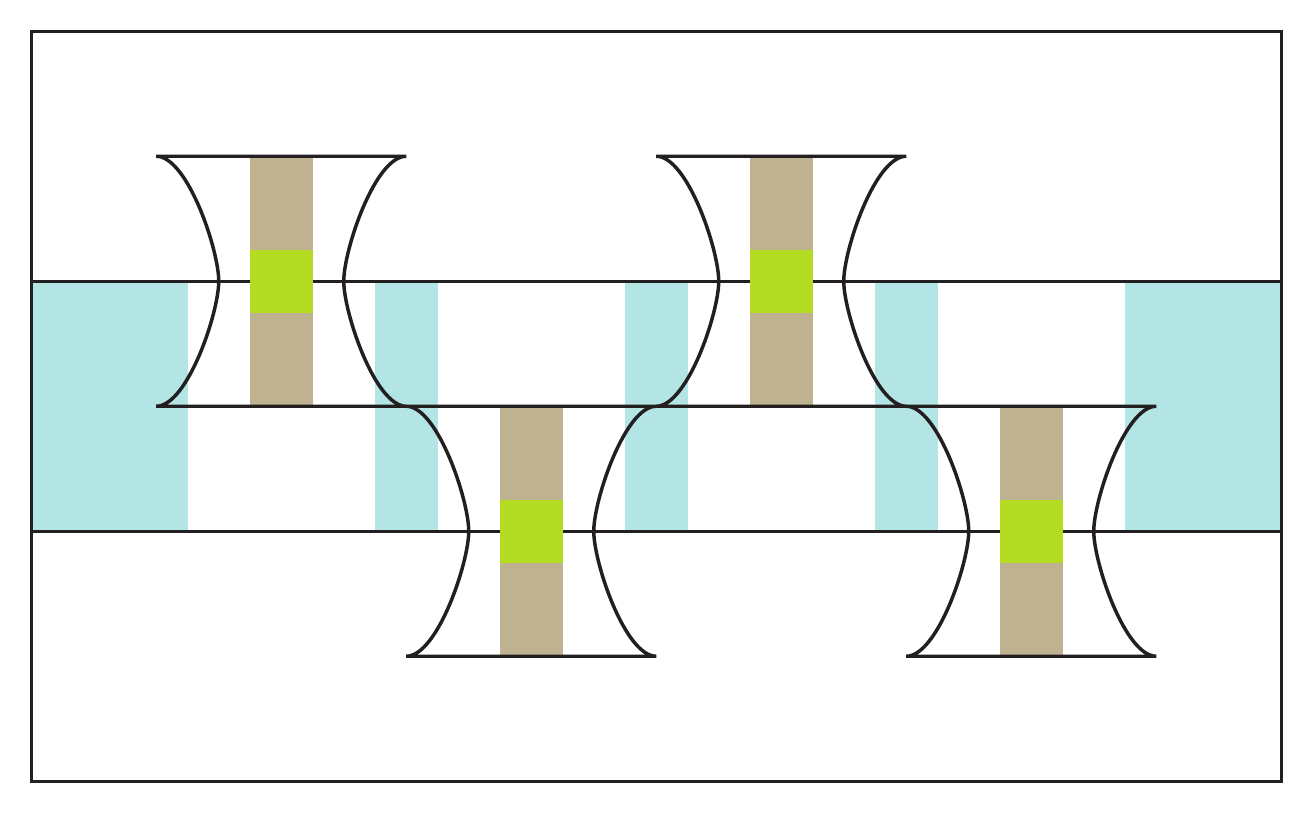}
\end{center}
\caption{\emph{Theorem~\ref{thm:kinkapprox} provides asymptotics in horizontal strips in the 
$(z,t)$-plane with the exception of the small rectangles shown in green, and with optimal accuracy of $\bo(\epsilon_N^{1/3})$ in the subregions shown in blue.  Theorem~\ref{thm:grazingapprox} provides asymptotics in hourglass-shaped regions with the exception of where they overlap the blue regions (where the error estimate of Theorem~\ref{thm:kinkapprox} is optimally small) and with optimal accuracy of $\bo(\epsilon_N^{1/3})$ in the included green regions (where the error estimate of Theorem~\ref{thm:kinkapprox} fails to be small) and brown regions.}}
\label{fig:RegionsSchematic}
\end{figure}
It is clear from this figure that Theorems~\ref{thm:kinkapprox} and \ref{thm:grazingapprox} both provide simple asymptotic formulae for the
same fields in overlapping regions.  In the overlap domains, formulae of kink type provided by
Theorem~\ref{thm:kinkapprox}, and of grazing kink collision type provided by Theorem~\ref{thm:grazingapprox} are \emph{simultaneously valid}.  This is actually a consequence of
Theorem~\ref{thm:main}, but it may also be checked directly, by writing both types of formulae in
terms of common spatiotemporal independent variables.  

One way to characterize the asymptotic formulae giving the universal form of the dispersive
breakup of the simple pendulum separatrix under the sine-Gordon equation is to say that the wave pattern consists of waves of elementary ``solitonic'' forms that are spatiotemporally arranged according to solutions of certain nonlinear differential equations of Painlev\'e type.  To our knowledge, this type of phenomenon was first observed quite recently in a paper by Bertola and Tovbis \cite{BertolaT10cusp} in which all of the local maxima of the modulus of the solution to the focusing NLS equation in the semiclassical limit near the onset of oscillatory behavior (the point of so-called elliptic umbilic catastrophe of the approximating elliptic quasilinear Whitham modulational system) are individually modeled by the same exact solution of the focusing NLS equation (the rational or Peregrine breather solution) but the locations of the maxima are far more transcendental, being determined by the location of poles in the complex plane of the so-called \emph{tritronqu\'ee} solution of the Painlev\'e-I equation.  In the current context, the role of the tritronqu\'ee solution is played by the rational functions $\{\pu_\ind\}_{\ind\in\mathbb{Z}}$ solving the Painlev\'e-II system, and interestingly, \emph{we require an infinite number of different Painlev\'e functions to describe the full wave pattern}.

\section{Choice of $g(w)$}
\label{sec:frakg}
The use of a so-called ``$g$-function'' to precondition a matrix Riemann-Hilbert problem for subsequent asymptotic analysis by the Deift-Zhou steepest-descent method is by now a standard tool, having first been introduced by Deift, Venakides, and Zhou \cite{DeiftVZ97} in their analysis of the Korteweg-de Vries (KdV) equation in the small-dispersion limit.
Recall that the scalar function $g$ appearing in the formulation of Riemann-Hilbert Problem~\ref{rhp:basic} is subject to the basic requirements of
\begin{itemize}
\item Analyticity:  $g$ is analytic in the domain $\mathbb{C}\setminus (P_\infty\cup\mathbb{R}_+)$,
\item Boundary behavior:  $g$ takes well-defined boundary values on $P_\infty\cup\mathbb{R}_+$
in the classical sense, and the sum of the boundary values vanishes on $\mathbb{R}_+$, and,
\item Normalization:  $g(0)=g(\infty)=0$.
\end{itemize}
We can use any such function to formulate the basic Riemann-Hilbert problem for $\mathbf{N}(w)=\mathbf{N}_{N}^g(w;x,t)$, but it is well-known that to study the limit $N\to\infty$ it is important that  $g$ be chosen appropriately.  In \cite{BuckinghamMelliptic} it is explained how $g$ should be chosen
given $(x,t)\in\mathbb{R}^2$ to calculate the limit $N\to\infty$ with $(x,t)$ fixed.  The chosen function $g(w)=g(w;x,t)$ has the property that its boundary values satisfy a kind of equilibrium condition (an analogue of \eqref{eq:frakgcircle} below) on a contour (called $\beta$ in \cite{BuckinghamMelliptic}) whose topology changes
at criticality:  in the simplest (genus one) case, near $(x,t)=(x_\mathrm{crit},0)$ the equilibrium contour is either a nearly circular
arc through $w=1$ with complex-conjugate endpoints near $w=-1$ (this corresponds to modulated superluminal librational waves) or the union of a nearly circular closed contour with a transecting small interval of the real axis with both endpoints near $w=-1$ (this corresponds to modulated superluminal rotational waves).  Exactly at criticality the two endpoints coalesce at $w=-1$ and the equilibrium contour becomes the unit circle.  It turns out that proximity of endpoints of this contour is an obstruction to proving
uniform asymptotics for $(x,t)$ near, but not exactly at, criticality.

Here, since we want to allow $(x,t)$ to approach the critical point $(x_\mathrm{crit},0)$ at some
rate depending on $N$, we need to use a different function than $g(w)=g(w;x,t)$ defined in \cite{BuckinghamMelliptic}.  We will replace $g(w)$ with $\mathfrak{g}(w)=\mathfrak{g}(w;x,t)$
whose construction and properties are described below.  It is important that while
for general $x$ and $t$, the functions $g(w)$ and $\mathfrak{g}(w)$
are different,  at criticality they
coincide (they satisfy the same conditions, which can be shown to
uniquely determine the solution).

The function $\mathfrak{g}(w)$ is required to satisfy
the following conditions:
\begin{itemize}
\item $\mathfrak{g}(w)$ is analytic if $|w|\neq 1$ and $w\not\in\mathbb{R}_+$,
and it takes continuous and bounded boundary values on the unit circle $|w|=1$
and on the positive half-line $\mathbb{R}_+$.
\item $\mathfrak{g}(w)$ obeys the following jump conditions:
\begin{equation}
\mathfrak{g}_+(\xi)+\mathfrak{g}_-(\xi)=0,\quad \xi\in\mathbb{R}_+.
\end{equation}
\begin{equation}
2iQ(\xi;x,t)+\overline{L}(\xi)-\mathfrak{g}_+(\xi)-\mathfrak{g}_-(\xi)=0,\quad
|\xi|=1.
\label{eq:frakgcircle}
\end{equation}
\item $\mathfrak{g}(w)$ has the following values:
\begin{equation}
\mathfrak{g}(0)=0.
\end{equation}
\begin{equation}
\lim_{w\to\infty}\mathfrak{g}(w)=0.
\end{equation}
\end{itemize}
To obtain a formula for $\mathfrak{g}(w)$, we first write $\mathfrak{g}$ in the form
\begin{equation}
\mathfrak{g}(w)=\frac{1}{2}L(w)+\mathfrak{h}(w),
\label{eq:gh}
\end{equation}
where the function $L$ is defined by \eqref{eq:Ldefine}.  Therefore the above
conditions on $\mathfrak{g}$ are translated into equivalent
conditions on $\mathfrak{h}$:
\begin{itemize}
\item $\mathfrak{h}(w)$ is analytic if $|w|\neq 1$, $w\not\in\mathbb{R}_+$,
and $w\not\in [\mathfrak{a},\mathfrak{b}]$, and it takes continuous boundary values on the
unit circle $|w|=1$, on the positive half-line $\mathbb{R}_+$, and 
on the interval $[\mathfrak{a},\mathfrak{b}]$.
\item $\mathfrak{h}(w)$ obeys the following jump conditions:
\begin{equation}
\mathfrak{h}_+(\xi)+\mathfrak{h}_-(\xi)=0,\quad \xi\in\mathbb{R}_+.
\end{equation}
\begin{equation}
2iQ(\xi;x,t)-\mathfrak{h}_+(\xi)-\mathfrak{h}_-(\xi)=0,\quad |\xi|=1.
\end{equation}
\begin{equation}
\mathfrak{h}_+(\xi)-\mathfrak{h}_-(\xi)=-\tfrac{1}{2}(L_+(\xi)-L_-(\xi))=-i\theta_0(\xi),\quad w\in [\mathfrak{a},\mathfrak{b}].
\end{equation}
In the latter condition, the two complementary segments $(\mathfrak{a},-1)$ and $(-1,\mathfrak{b})$ are taken to be oppositely oriented toward $\xi=-1$.
\item $\mathfrak{h}(w)$ has the following values:
\begin{equation}
\mathfrak{h}(0)=0.
\end{equation}
\begin{equation}
\lim_{w\to\infty}\mathfrak{h}(w)=0.
\end{equation}
\end{itemize}
Finally, we write $\mathfrak{h}(w)$ in terms of a new equivalent unknown
$\mathfrak{m}(w)$ by the piecewise substitution:
\begin{equation}
\mathfrak{h}(w):=\begin{cases}(-w)^{-1/2}\mathfrak{m}(w),
&\quad |w|>1\\
 -(-w)^{-1/2}\mathfrak{m}(w),&\quad |w|<1.
\end{cases}
\label{eq:hm}
\end{equation}
The conditions satisfied by $\mathfrak{m}(w)$ are then the following:
\begin{itemize}
\item $\mathfrak{m}(w)$ is analytic if $|w|\neq 1$ and $w\not\in [\mathfrak{a},\mathfrak{b}]$,
and it takes continuous boundary values on the unit circle $|w|=1$ and
the interval $[\mathfrak{a},\mathfrak{b}]$.
\item $\mathfrak{m}(w)$ obeys the following jump conditions:
\begin{equation}
\mathfrak{m}_+(\xi)-\mathfrak{m}_-(\xi)=-2i(-\xi)^{1/2}Q(\xi;x,t),\quad |\xi|=1,
\label{eq:mjumpcircle}
\end{equation}
where the orientation of the circle is in the positive (counterclockwise)
sense, and
\begin{equation}
\mathfrak{m}_+(\xi)-\mathfrak{m}_-(\xi)=-i(-\xi)^{1/2}\theta_0(\xi),\quad
\xi\in [\mathfrak{a},\mathfrak{b}],
\end{equation}
where now the interval $[\mathfrak{a},\mathfrak{b}]$ is taken to be oriented left-to-right.
Note that $(-\xi)^{1/2}>0$ for $\mathfrak{a}<\xi<\mathfrak{b}<0$, and that the right-hand side of \eqref{eq:mjumpcircle} is a single-valued function on 
the circle:
\begin{equation}
-2i(-\xi)^{1/2}Q(\xi;x,t)=\frac{x}{2}(1-\xi)-\frac{t}{2}(1+\xi),\quad |\xi|=1.
\end{equation}
\item $\mathfrak{m}(w)$ is uniformly bounded, and has the following value:
\begin{equation}
\mathfrak{m}(0)=0.
\label{eq:mathfrakmzero}
\end{equation}
\end{itemize}
We may now express $\mathfrak{m}(w)$ in terms of a Cauchy integral
using the Plemelj formula:
\begin{equation}
\mathfrak{m}(w)=\mathfrak{m}_0 + \frac{1}{2\pi i}\oint_{|\xi|=1}
\frac{-2i(-\xi)^{1/2}Q(\xi;x,t)}{\xi-w}\,d\xi +\frac{1}{2\pi i}
\int_{\mathfrak{a}}^{\mathfrak{b}}\frac{-i(-\xi)^{1/2}\theta_0(\xi)}{\xi-w}\,d\xi,
\end{equation}
where the constant $\mathfrak{m}_0$ is the most general entire function
we can add to the Cauchy integrals consistent with the uniform boundedness
of $\mathfrak{m}(w)$.  The integral over the positively-oriented unit
circle may be evaluated in closed form:
\begin{equation}
\mathfrak{m}(w)=\mathfrak{m}_0 +\frac{1}{2\pi i}
\int_{\mathfrak{a}}^{\mathfrak{b}}\frac{-i(-\xi)^{1/2}\theta_0(\xi)}{\xi-w}\,d\xi
+\begin{cases}\displaystyle \frac{x}{2}(1-w)-\frac{t}{2}(1+w),&\quad |w|<1\\\\
0,&\quad |w|>1.
\end{cases}
\end{equation}
We choose the value of the constant $\mathfrak{m}_0$ to enforce the condition 
\eqref{eq:mathfrakmzero}:
\begin{equation}
\mathfrak{m}_0=\frac{t-x}{2} -\frac{1}{2\pi}\int_{\mathfrak{a}}^{\mathfrak{b}}
\frac{\theta_0(\xi)}{(-\xi)^{1/2}}\,d\xi.
\end{equation}
It follows that 
\begin{equation}
\mathfrak{m}(w)=\frac{w}{2\pi}\int_{\mathfrak{a}}^{\mathfrak{b}}\frac{\theta_0(\xi)\,d\xi}
{(-\xi)^{1/2}(\xi-w)} +
\begin{cases}\displaystyle -\frac{x+t}{2}w,&\quad |w|<1\\\\
\displaystyle\frac{t-x}{2},&\quad |w|>1.
\end{cases}
\end{equation}
By writing $\mathfrak{g}$ in terms of $\mathfrak{h}$ (by \eqref{eq:gh}) and then writing $\mathfrak{h}$ in terms of $\mathfrak{m}$ (by \eqref{eq:hm}), this completes the construction of $\mathfrak{g}(w)$.

\section{
Steepest Descent, the Outer 
Model Problem, and its Solution}  

According to \eqref{eq:frakgcircle}, with $\mathbf{N}(w)=\mathbf{N}^\mathfrak{g}_N(w;x,t)$
and with $\mathfrak{g}(w)=\mathfrak{g}(w;x,t)$ defined as in \S\ref{sec:frakg}, a condition of ``equilibrium type'' holds on
the unit circle $|w|=1$,
regardless of the values of $x$ and $t$.  The next step is to ``open a
lens'' about the whole circle.  The lens consists of four disjoint open sets separated by the unit circle and the interval $I\subset\mathbb{R}_+$ as shown with shading in Figure~\ref{fig:critOcontour}.  The union of the two components lying on the left (respectively right) of the unit circle with its orientation will be denoted $\Lambda^+$ (respectively $\Lambda^-$).   We
now define a new unknown $\mathbf{O}(w)$ equivalent to $\mathbf{N}(w)$ by the substitution
\begin{equation}
\mathbf{O}(w):=\begin{cases}
\displaystyle \mathbf{N}(w)T_N(w)^{-\sigma_3/2}\begin{bmatrix}1 & \mp ie^{-[2iQ(w;x,t)+L(w)\mp i\theta_0(w)-2\mathfrak{g}(w)]/\epsilon_N}\\ 0 &1\end{bmatrix},&\quad
w\in\Lambda^\pm \\\\
\mathbf{N}(w),&\quad w\in\mathbb{C}\setminus(\Sigma_{\mathbf{N}}\cup\Lambda^+\cup\Lambda^-).
\end{cases}
\label{eq:OintermsofN}
\end{equation}
Here, the square root $T_N(w)^{1/2}$ is well-defined as the principal branch, since
the uniform approximation $T_N(w)\approx 1$ holds when $N$ is large.
(The definition \eqref{eq:OintermsofN} coincides with that of $\mathbf{O}(w)$ in terms of $\mathbf{N}(w)$ given
in \cite{BuckinghamMelliptic} in the case that $\Delta=\emptyset$ and with $g$ replaced
by $\mathfrak{g}$.)
The jump contour for $\mathbf{O}(w)$
is the same as that for $\mathbf{N}(w)$ but augmented by four arcs emanating
from $w=-1$ representing the ``outer'' boundaries of the lens halves $\Lambda^\pm$ as shown
in Figure~\ref{fig:critOcontour}.
\begin{figure}[h]
\begin{center}
\includegraphics{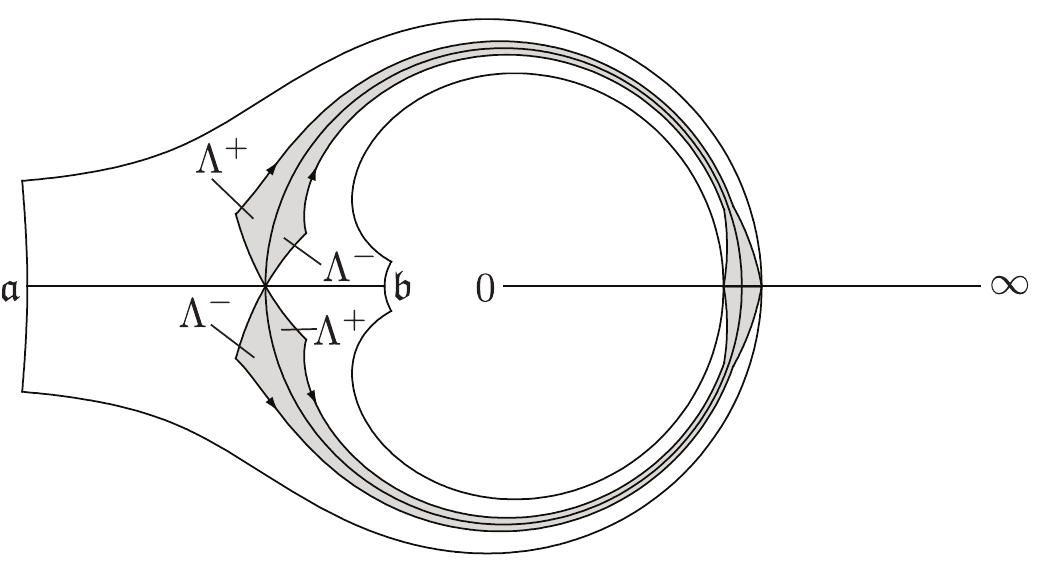}
\end{center}
\caption{\emph{The contour of discontinuity for the sectionally analytic
function $\mathbf{O}(w)$, with the lens halves $\Lambda^\pm$ indicated with shading.
Arcs of the jump contour also present in the jump contour
for $\mathbf{N}(w)$ (see Figure~\ref{fig:critNcontour}) retain that
orientation, and the lens boundaries are oriented as indicated.}}
\label{fig:critOcontour}
\end{figure}

It is a consequence of the equilibrium condition \eqref{eq:frakgcircle} satisfied by the
boundary values taken by $\mathfrak{g}(w)$ on the 
unit circle that
$\mathbf{O}(w)$ satisfies the piecewise-constant jump conditions
\begin{equation}
\mathbf{O}_+(\xi)=\mathbf{O}_-(\xi)(\mp i\sigma_1),\quad |\xi|=1,\quad
\pm\Im\{\xi\}>0.
\end{equation}
Since the equilibrium contour is a closed curve, we may 
remove these discontinuities from the problem by making another explicit 
transformation:
\begin{equation}
\mathbf{P}(w):=\begin{cases}
\mathbf{O}(w),&\quad |w|>1\\
\mathbf{O}(w)(-i\sigma_1),&
\quad \text{$|w|<1$ and $\Im\{w\}>0$}\\
\mathbf{O}(w)i\sigma_1,&
\quad\text{$|w|<1$ and $\Im\{w\}<0$}.
\end{cases}
\end{equation}
Unlike $\mathbf{O}(w)$, the matrix function $\mathbf{P}(w)$  extends continuously to the unit circle, but 
the jump conditions it satisfies within the unit disk where it differs from $\mathbf{O}(w)$ are altered somewhat from 
those of $\mathbf{O}(w)$, including a new jump of the simple form
$\mathbf{P}_+(\xi)=-\mathbf{P}_-(\xi)$ on the segment $\mathfrak{b}<\xi<0$.
The contour of discontinuity for $\mathbf{P}(w)$ is illustrated in 
Figure~\ref{fig:critPcontour}.
\begin{figure}[h]
\begin{center}
\includegraphics{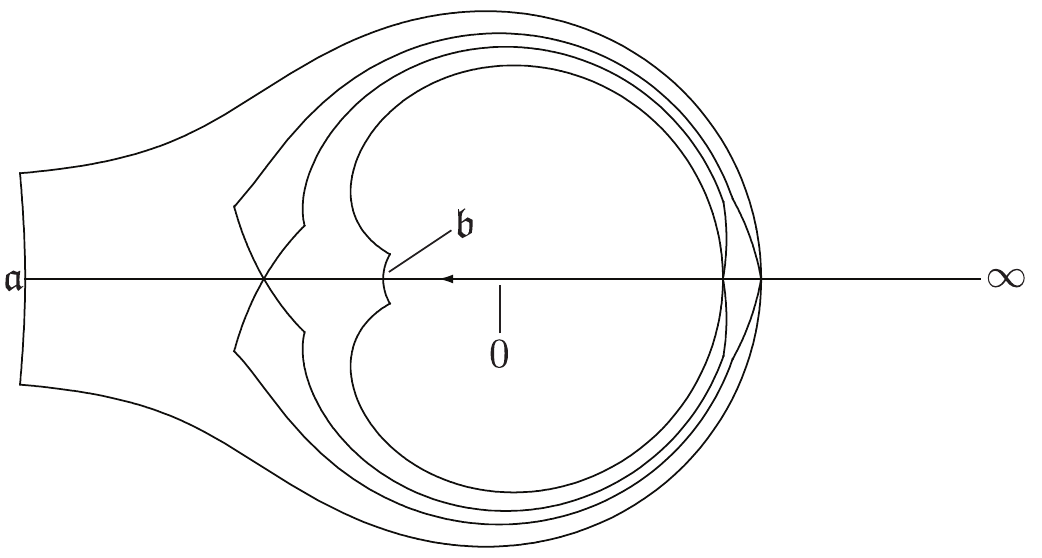}
\end{center}
\caption{\emph{The contour of discontinuity for the sectionally
analytic function 
$\mathbf{P}(w)$.  Contour arcs also contained in the jump contour
for $\mathbf{O}(w)$ (see Figures~\ref{fig:critNcontour} and 
\ref{fig:critOcontour}) retain that
orientation, and otherwise (that is, in the interval $(\mathfrak{b},0)$) the orientation is as indicated.}}
\label{fig:critPcontour}
\end{figure}

The jump discontinuities of $\mathbf{P}(w)$ will turn out to be
negligible except in a small neighborhood of $w=-1$ and along the ray
$w>-1$, and $\mathbf{P}(w)$ is a matrix tending to the identity
as $w\to\infty$.  These considerations lead us to propose a model 
Riemann-Hilbert problem
whose solution we expect to approximate $\mathbf{P}(w)$ away
from $w=-1$:
\begin{rhp}[Outer model problem near criticality]
Let $w_*\approx -1$ be a real parameter.  Find a $2\times 2$ matrix
$\dot{\mathbf{P}}^\mathrm{out}(w)$ with the following properties.
\begin{itemize}
\item[]\textbf{Analyticity:} $\dot{\mathbf{P}}^\mathrm{out}(w)$ is an analytic
  function of $w$ for $w\in\mathbb{C}\setminus 
[w_*,+\infty)$, H\"older
continuous up to the jump interval $[w_*,+\infty)$
with the exception of an arbitrarily
small neighborhood of the point $w=w_*$.  In a neighborhood of
$w=w_*$, the elements of $\dot{\mathbf{P}}^\mathrm{out}(w)$ are
bounded by an unspecified power of $|w-w_*|$.
\item[]\textbf{Jump condition:} The boundary values taken by
  $\dot{\mathbf{P}}^\mathrm{out}(w)$ along $(w_*,0)$ and 
  $\mathbb{R}_+$ 
  satisfy the following jump conditions:
\begin{equation}
\dot{\mathbf{P}}^\mathrm{out}_+(\xi)=-\dot{\mathbf{P}}^\mathrm{out}_-(\xi),
\quad w_*<\xi<0,
\end{equation}
and
\begin{equation}
\dot{\mathbf{P}}^\mathrm{out}_+(\xi)=
\sigma_2\dot{\mathbf{P}}^\mathrm{out}_-(\xi)
\sigma_2,\quad\xi\in\mathbb{R}_+.
\end{equation}
As these jump conditions are both involutive,  orientation of
the jump contours is irrelevant.
\item[]\textbf{Normalization:}  The following normalization conditions hold:
\begin{equation}
\lim_{w\to\infty}\dot{\mathbf{P}}^\mathrm{out}(w)=\mathbb{I}\quad\text{and}\quad
\det(\dot{\mathbf{P}}^\mathrm{out}(w))\equiv 1.
\end{equation}
\end{itemize}
\label{rhp:wOdotcrit}
\end{rhp}

There are many solutions of this Riemann-Hilbert problem, a fact that can be traced partly to the unspecified power-law rate of growth admitted as $w\to w_*$.  The complete variety
of solutions is not directly relevant for us; we will simply select a family of particular solutions
in the class of diagonal matrices.  Indeed, for each $\ind\in\mathbb{Z}$, we
have the solution
\begin{equation}
\dot{\mathbf{P}}^\mathrm{out}(w)=\dot{\mathbf{P}}_\ind^\mathrm{out}(w):=
\left[(-w)^{1/2}+(-w_*)^{1/2}\right]^{(1-2\ind)\sigma_3}(w_*-w)^{(\ind-1/2)\sigma_3}.
\label{eq:PoutDS}
\end{equation}
Note that if $\ind$ is held fixed, then
$\dot{\mathbf{P}}^\mathrm{out}_\ind(w)$ and its inverse are bounded when
$w$ is bounded away from $w_*$.  Later (see \eqref{eq:wstardefine}), 
we will let $w_*$ depend
weakly on $x$ and $t$ in such a way that $w_*\to -1$ as $t\to 0$ and
$x\to x_\mathrm{crit}$.  Clearly, any bound for
$\dot{\mathbf{P}}_\ind^\mathrm{out}(w)$ or its inverse valid for $w$
bounded away from $-1$ will hold uniformly with respect to $x$ and $t$ 
near criticality,
and of course
$\dot{\mathbf{P}}_\ind^\mathrm{out}(w)$ is independent of $\epsilon_N$.

\section{Inner Model Problem Valid near $w=-1$}
\subsection{Exact jump matrices for $\mathbf{P}(w)$ near $w=-1$.}
Let $U$ be a fixed neighborhood of $w=-1$.  By straightforward substitutions,
we may assume without loss of generality that within $U$, the jump contour for 
$\mathbf{P}(w)$ consists of the real axis together with four
arcs (the lens boundaries) meeting at some real point $w=w_*\in U\cap\mathbb{R}$
tending to $-1$ as criticality is approached; the way $w_*$ is determined
as a function of $x$ and $t$ will be explained
later (the resulting formula being \eqref{eq:wstardefine}).
We suppose further that the lens boundaries
lie along certain curves (also to be specified later) tangent at $w=w_*$ to the 
straight lines $\arg(w-w_*)=\pm\pi/3$ and $\arg(w-w_*)=\pm 2\pi/3$.
Some calculations show that the exact jump conditions for 
$\mathbf{P}(w)$ within $U$ can be expressed only in terms of the
analytic function $T_N(w)\approx 1$ and another analytic function 
$k(w)$, which is the analytic continuation from $w<\min\{w_*,-1\}$ of the function
\begin{equation}
k(w)=k(w;x,t):= 2iQ(w;x,t)+\overline{L}(w)-2\mathfrak{g}(w;x,t),\quad w<-1.
\label{eq:kdefinecrit}
\end{equation}
The jump matrix $\mathbf{V}_{\mathbf{P}}(\xi)$ for which along each
of six contour arcs meeting at $w=w_*$ 
we have 
$\mathbf{P}_+(\xi)=\mathbf{P}_-(\xi)\mathbf{V}_{\mathbf{P}}(\xi)$
is illustrated in Figure~\ref{fig:PtildeLocalJumps}.
\begin{figure}[h]
\begin{center}
\includegraphics{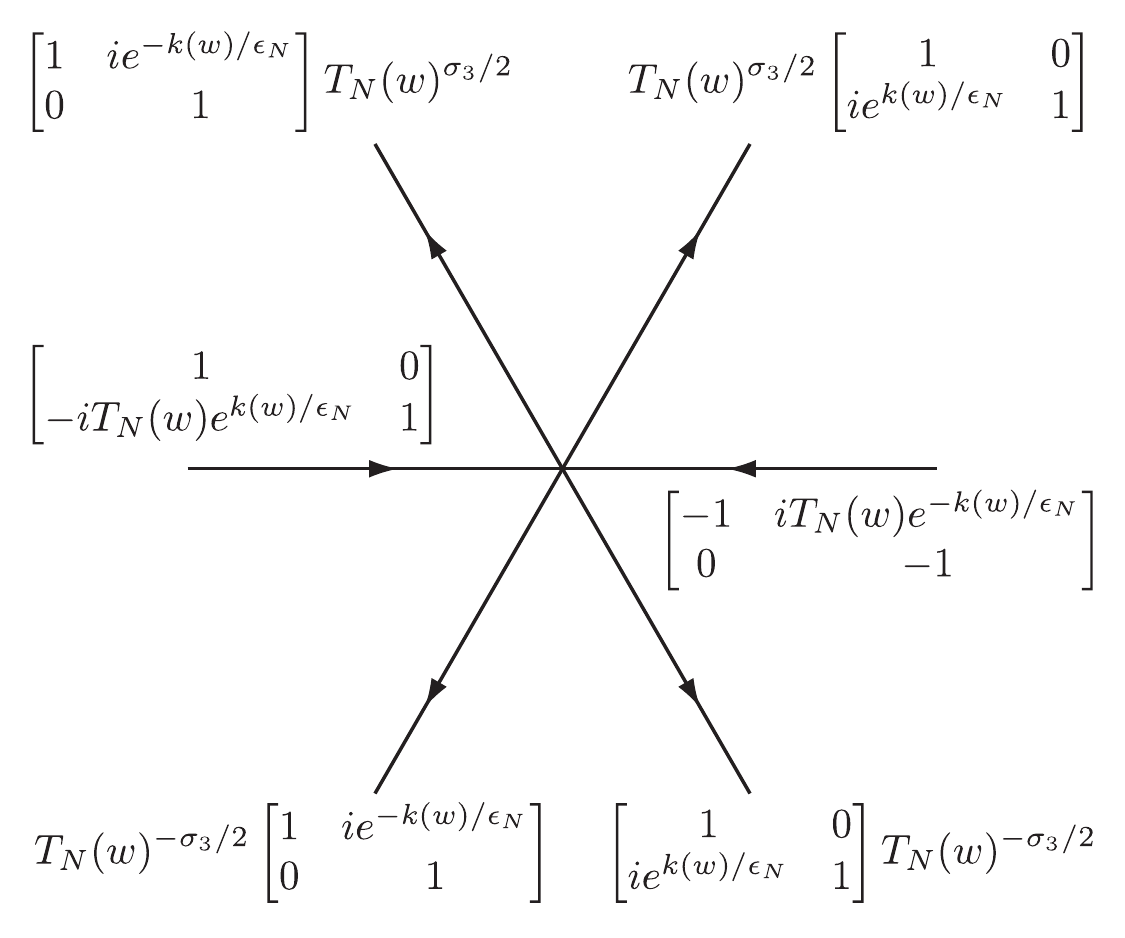}
\end{center}
\caption{\emph{The jump matrix for $\mathbf{P}(w)$ near
    $w=w_*\approx -1$.  The four non-real arcs of the jump contour are in 
general only
    approximately linear near the intersection point $w=w_*$.}}
\label{fig:PtildeLocalJumps}
\end{figure}
The function $T_N(w)$ can easily be removed from the jump conditions by
making the following near-identity transformation (locally, in $U$):
\begin{equation}
\tilde{\mathbf{P}}(w):=\mathbf{P}(w)d_N(w)^{\sigma_3}
\label{eq:tildeQdefine}
\end{equation}
where $d_N(w)$ is the piecewise-analytic function given by
\begin{equation}
d_N(w):=
\begin{cases}
T_N(w)^{-1/2},&\quad \text{in the region tangent to the sector $|\arg(w_*-w)|<\frac{\pi}{3}$}\\
T_N(w)^{1/2},&\quad \text{in the region tangent to the sector $|\arg(w-w_*)|<\frac{\pi}{3}$}\\
1,&\quad\text{otherwise}.
\end{cases}
\label{eq:Ddefine}
\end{equation}
Note that $d_N(w)=1+\bo(\epsilon_N)$ holds uniformly for $w\in U$, so indeed
$\mathbf{P}(w)^{-1}\tilde{\mathbf{P}}(w)= \mathbb{I}+\bo(\epsilon_N)$.  The
jump conditions satisfied by $\tilde{\mathbf{P}}(w)$ near $w=w_*\approx -1$
are as shown in Figure~\ref{fig:PtildeLocalJumps} but with $T_N(w)$
replaced by $1$ in all cases.

\subsection{Expansion of $k(w)$ about $w=-1$.}
\begin{proposition}
The function $k(w)=k(w;x,t)$ is analytic at $w=-1$, having the Taylor expansion 
\begin{equation}
\begin{split}
k(w;x,t)&=-t +\frac{1}{2}\Delta x(w+1) +
\left[\frac{1}{4}\Delta x -\frac{1}{8}t\right](w+1)^2 \\
&\quad\quad\quad\quad{}+\left[\frac{3}{16}\Delta x-
\frac{1}{8}t + \nu\right](w+1)^3 \\
&\quad\quad\quad\quad{}+\left[\frac{5}{32}\Delta x-\frac{15}{128}t
+\frac{3}{2}\nu\right](w+1)^4
+ \bo((w+1)^5),\quad w\to -1,
\end{split}
\label{eq:bexpansion}
\end{equation}
where $\nu>0$ is independent of $w$, $x$, and $t$, being defined by \eqref{eq:nudefinecrit}.  All Taylor coefficients are linear in $\Delta x$ and $t$, and hence in particular the error term is uniform for bounded $\Delta x$ and $t$.  
\label{prop:kexpand}
\end{proposition}

\begin{proof}
For $w<\min\{w_*,-1\}$ and $w\in U$, we have the formula
\begin{equation}
k(w;x,t)=2iQ(w;x,t)-\frac{1}{(-w)^{1/2}}[\mathfrak{m}_+(w)+\mathfrak{m}_-(w)],
\end{equation}
where $\mathfrak{m}(w)$ is given by 
\begin{equation}
\mathfrak{m}(w) = \frac{t-x}{2}+
\frac{w}{2\pi}\int_{\mathfrak{a}}^{\mathfrak{b}}\frac{\theta_0(s)\,ds}{(-s)^{1/2}(s-w)}.
\end{equation}
To express the boundary values $\mathfrak{m}_\pm(w)$ for $w<\min\{w_*,-1\}$ in a form admitting analytic continuation to a full neighborhood of $w=-1$ (a point that lies on the discontinuity contour $[\mathfrak{a},\mathfrak{b}]$ of $\mathfrak{m}(w)$),
let $C_\pm$ denote two
contours from $s=\mathfrak{a}$ to $s=\mathfrak{b}$, with $C_+$ in the upper half $s$-plane and
$C_-$ in the lower half $s$-plane, and such that the two contours are mapped onto
each other with reversal of orientation under the mapping $s\mapsto s^{-1}$ (in particular, this 
mapping permutes the points $\mathfrak{a}$ and $\mathfrak{b}$).
Then, also recalling the definition of $Q(w;x,t)$ in terms of $E(w)$ and $D(w)$,
\begin{equation}
k(w)=-\frac{x}{2}\left((-w)^{1/2}-(-w)^{-1/2}\right) 
-\frac{t}{2}\left((-w)^{1/2}+(-w)^{-1/2}\right)
+\frac{(-w)^{1/2}}{2\pi}\int_{C_+\cup C_-}
\frac{\theta_0(s)\,ds}{(-s)^{1/2}(s-w)}.
\label{eq:kanalyticcrit}
\end{equation}
This formula represents the analytic continuation to a full neighborhood of $w=-1$ of the function $k$ originally defined by \eqref{eq:kdefinecrit}
only for $w<\min\{w_*,-1\}$. In fact, the domain of analyticity for $k$ has now been extended to the domain enclosed by the
contour $C_+\cup C_-$.  It is also obvious that $k$ depends linearly on $x$ and $t$, and this property will clearly be inherited by all of its Taylor coefficients.

To analyze $k(w)$ near the point $w=-1$ to which the self-intersection point
$w_*$ will converge at criticality, we begin with 
the following elementary Taylor expansions (convergent for $|w+1|<1$):
\begin{equation}
\begin{split}
-\frac{1}{2}\left((-w)^{1/2}-(-w)^{-1/2}\right)& =
\frac{1}{2}(w+1) +\frac{1}{4}(w+1)^2+\frac{3}{16}(w+1)^3 
+\frac{5}{32}(w+1)^4+\bo((w+1)^5)\\
-\frac{1}{2}\left((-w)^{1/2}+(-w)^{-1/2}\right)&=
-1-\frac{1}{8}(w+1)^2-\frac{1}{8}(w+1)^3 -\frac{15}{128}(w+1)^4+\bo((w+1)^5).
\end{split}
\label{eq:elementaryTaylorcrit}
\end{equation}
We also have the expansion
\begin{equation}
\begin{split}
\frac{(-w)^{1/2}}{s-w}&=\frac{1}{s+1} +\left[\frac{1}{(s+1)^2}-\frac{1}{2}
\frac{1}{s+1}\right](w+1)\\
&\quad\quad\quad\quad{}+\left[\frac{1}{(s+1)^3}-\frac{1}{2}\frac{1}{(s+1)^2}
-\frac{1}{8}\frac{1}{s+1}\right](w+1)^2\\
&\quad\quad\quad\quad{}+\left[\frac{1}{(s+1)^4}-\frac{1}{2}\frac{1}{(s+1)^3}
-\frac{1}{8}\frac{1}{(s+1)^2}-\frac{1}{16}\frac{1}{s+1}\right](w+1)^3\\
&\quad\quad\quad\quad{}+\left[
\frac{1}{(s+1)^5}-\frac{1}{2}\frac{1}{(s+1)^4}-\frac{1}{8}
\frac{1}{(s+1)^3}-\frac{1}{16}\frac{1}{(s+1)^2}-\frac{5}{128}
\frac{1}{s+1}\right](w+1)^4\\
&\quad\quad\quad\quad{}+\bo((w+1)^5),
\end{split}
\end{equation}
which is a convergent power series for $|w+1|<\min\{1,|s+1|\}$.  Therefore, given contours $C_\pm$ as above, if $|w+1|$ is sufficiently small 
we may integrate term-by-term to obtain
\begin{equation}
\begin{split}
\frac{(-w)^{1/2}}{2\pi}\int_{C_+\cup C_-}\frac{\theta_0(s)\,ds}{(-s)^{1/2}(s-w)}&
= I_1 + \left[I_2-\frac{1}{2}I_1\right](w+1) +\left[I_3-\frac{1}{2}I_2-
\frac{1}{8}I_1\right](w+1)^2 \\
&\quad\quad\quad\quad{}+\left[I_4-\frac{1}{2}I_3-\frac{1}{8}I_2-
\frac{1}{16}I_1\right](w+1)^3 \\
&\quad\quad\quad\quad{}+\left[I_5-\frac{1}{2}I_4-\frac{1}{8}I_3-\frac{1}{16}
I_2-\frac{5}{128}I_1\right](w+1)^4
+ \bo((w+1)^5),
\end{split}
\label{eq:termbytermcrit}
\end{equation}
where
\begin{equation}
I_k:=\frac{1}{2\pi}\int_{C_+\cup C_-}\frac{\theta_0(s)\,ds}{(-s)^{1/2}(s+1)^k},\quad k=1,\dots,5.
\label{eq:Ikdefine}
\end{equation}
By making the substitution $\sigma=s^{-1}$, we may rewrite these integrals
in the form
\begin{equation}
I_k=-\frac{1}{2\pi}\int_{C_+\cup C_-}\frac{\theta_0(\sigma)\,d\sigma}{(-\sigma)^{1/2}(\sigma+1)(\sigma^{-1}+1)^{k-1}}.
\label{eq:Ikrewrite}
\end{equation}
Therefore, we see immediately that 
\begin{equation}
I_1=-I_1=0, 
\label{eq:I1zerocrit}
\end{equation}
while by averaging
\eqref{eq:Ikdefine} and \eqref{eq:Ikrewrite} we obtain for $I_2$,
$I_3$, and $I_4$ the following:
\begin{equation}
I_2 = I_3=\frac{1}{4\pi}\int_{C_+\cup C_-}\frac{\theta_0(s)}{(-s)^{1/2}(s+1)}
\frac{1-s}{1+s}\,ds
\label{eq:I2I3basic}
\end{equation}
and
\begin{equation}
I_4 = \frac{1}{4\pi}\int_{C_+\cup C_-}\frac{\theta_0(s)}{(-s)^{1/2}(s+1)}
\frac{1-s^3}{(s+1)^3}\,ds.
\label{eq:I4basic}
\end{equation}
Also, averaging $I_5$ as given by \eqref{eq:Ikdefine} and \eqref{eq:Ikrewrite}
and comparing with \eqref{eq:I2I3basic} and \eqref{eq:I4basic} we obtain
the identity
\begin{equation}
I_5 = 2I_4-I_2=2I_4-I_3.
\label{eq:I5identity}
\end{equation}

Now we calculate $I_2=I_3$ exactly.  Recalling
the definitions \eqref{eq:ED} of the functions $E(\cdot)$ and $D(\cdot)$, as well as the definition
\eqref{eq:theta0defcrit} of $\theta_0(s)$, we see that
\begin{equation}
I_2=I_3=-\frac{i}{8\pi}
\int_{C_+\cup C_-}\frac{\Psi(E(s))}{D(s)^3}E(s)E'(s)\,ds.
\end{equation}
Since $D(s)$ can be eliminated in favor of $E(s)$ by the identity $D(s)^2=E(s)^2+\tfrac{1}{4}$,  we would like to introduce 
$v=-4iE(s)$ as a new integration variable.  The two contours $C_\pm$
are mapped under $-4iE(\cdot)$ to two oppositely-oriented copies of the
same contour.  In particular, $-4iE(C_-)$ is a teardrop-shaped contour beginning
and ending at the point $v=-4iE(\mathfrak{a})=-4iE(\mathfrak{b})$ lying on the real
axis to the right of $v=2$ and encircling the point $v=2=-4iE(-1)$ once in the 
positive sense.  Along the contours $C_\pm$ we have the identities 
$D(s)=\pm (E(s)^2+\tfrac{1}{4})^{1/2}$ (principal branch).  These considerations
show that both contributions from $C_+$ and from $C_-$ are equal, and so
\begin{equation}
I_2=I_3=\frac{i}{2\pi}
\int_{-4iE(C_-)}\frac{\Psi(iv/4)}{(4-v^2)^{3/2}}(-2v)\,dv.
\end{equation}
Next, we integrate by parts, using the fact that $\Psi(E(\mathfrak{a}))=0$:
\begin{equation}
I_2=I_3=\frac{i}{\pi}\int_{-4iE(C_-)}\frac{\varphi(v)\,dv}{(4-v^2)^{1/2}},\quad
\varphi(v):=\frac{d}{dv}\Psi(iv/4).
\end{equation}
Collapsing the contour to the top and bottom of the interval $(2,-4iE(\mathfrak{a}))$
and noting that $-4iE(\mathfrak{a})=-G(0)$ we obtain
\begin{equation}
I_2=I_3=\frac{2}{\pi}\int_2^{-G(0)}\frac{\varphi(v)\,dv}{\sqrt{v^2-4}}.
\end{equation}
Now substituting from  the definition \eqref{eq:Psidefinecrit} and exchanging the order of integration (see \cite[Proposition~1.1]{BuckinghamMelliptic}) leads to the identity 
\begin{equation}
I_2=I_3 = -\frac{1}{2}G^{-1}(-2)=
-\frac{1}{2}x_\mathrm{crit}.
\label{eq:I2xcrit}
\end{equation}

Next we will show that
\begin{equation}
I_4 = \nu-\frac{1}{2}x_\mathrm{crit}
\label{eq:I4formulacrit}
\end{equation}
where $\nu>0$ is defined by \eqref{eq:nudefinecrit}.
Using \eqref{eq:I2xcrit} together with \eqref{eq:I2I3basic}--\eqref{eq:I4basic}
and the definitions \eqref{eq:ED} of $E(\cdot)$ and $D(\cdot)$ gives
\begin{equation}
I_4+\frac{1}{2}x_\mathrm{crit} = I_4-I_2 = \frac{i}{128\pi}\int_{C_+\cup C_-}\frac{\Psi(E(s))}{D(s)^5}
E(s)E'(s)\,ds.
\end{equation}
Introducing $v=-4iE(s)$ as a new integration variable, as above, yields
\begin{equation}
I_4+\frac{1}{2}x_\mathrm{crit} = \frac{1}{2\pi i}\int_{-4iE(C_-)}\frac{\Psi(iv/4)}{(4-v^2)^{5/2}}(-2v)\,dv,
\end{equation}
and then integrating by parts,
\begin{equation}
I_4+\frac{1}{2}x_\mathrm{crit} = \frac{1}{3\pi i}\int_{-4iE(C_-)}\frac{\varphi(v)\,dv}{(4-v^2)^{3/2}}.
\label{eq:nuint}
\end{equation}
Note that by the subsitution $m=G(s)^2$, $\varphi(v)$ may be written in the form
\begin{equation}
\varphi(v)=-\frac{v}{2}\int_{v^2}^{G(0)^2}\frac{\mathscr{G}(m)}{\sqrt{m-v^2}\sqrt{G(0)^2-m}}
\frac{dm}{m},\quad 0<v<-G(0),
\end{equation}
where the function $\mathscr{G}$ is defined in terms of the initial condition
$G(\cdot)$ by \eqref{eq:hmdefcrit}.  We need to write this formula in a
way that admits analytic continuation to complex $v$ such as those
$v\in -4iE(C_-)$.  To this end, let
\begin{equation}
T(m,z):=i(G(0)^2-m)^{-1}\left(1-\frac{z-G(0)^2}{m-G(0)^2}\right)^{-1/2}
\end{equation}
where the principal branch is meant, and then note that
\begin{equation}
\varphi(v) = \frac{v}{4}\oint_{O} T(m,v^2)\frac{\mathscr{G}(m)}{m}\,dm,
\label{eq:varphivcomplex}
\end{equation}
where $O$ is a loop contour surrounding, in the negative (clockwise) sense,
the straight-line branch cut of $T(m,v^2)$ (viewed as a function of $m$) 
connecting $v^2$ with $G(0)^2$.  In writing the formula \eqref{eq:varphivcomplex}
we therefore are using Assumption~\ref{ass:scriptGnice} to guarantee analyticity of $\mathscr{G}$ in a neighborhood of $m=G(0)^2$.
With the loop contour $O$ fixed, it is clear that $\varphi(v)$ as given by \eqref{eq:varphivcomplex} is 
analytic if $z=v^2$ varies in the region enclosed by $O$.  In particular, 
we will assume that $C_-$ has been chosen so that $-4iE(C_-)$ is completely 
contained in this region.  Since $-4iE(C_-)$ and $L$ are both fixed contours,
we may exchange the order of integration upon substituting 
\eqref{eq:varphivcomplex} into \eqref{eq:nuint}:
\begin{equation}
I_4+\frac{1}{2}x_\mathrm{crit} = \frac{1}{24\pi i}\oint_{O}\frac{\mathscr{G}(m)}{m}j(m)\,dm,
\label{eq:nujint}
\end{equation}
where the inner integral is now:
\begin{equation}
j(m):=\int_{-4iE(C_-)}
\frac{T(m,v^2)}{(4-v^2)^{3/2}}(2v)\,dv = 
\int_{[-4iE(C_-)]^2}\frac{T(m,z)\,dz}{(4-z)^{3/2}}.
\end{equation}
Note that the contour $[-4iE(C_-)]^2$ is a teardrop-shaped contour beginning
and ending at $z=G(0)^2$ and encircling the part of the (principal) 
branch cut of $(4-z)^{3/2}$ between $z=4$ and $z=G(0)^2>4$ once in the
positive sense, and each $m\in O$ lies outside of this closed contour.
To evaluate $j(m)$ we note that the branch cut of $T(m,z)$ viewed as a function
of $z$ is the ray from $z=m$ to $z=\infty$ in the direction away from
$z=G(0)^2$.  By simple contour deformations of $[-4iE(C_-)]^2$ taking into
account that the integrand $T(m,z)(4-z)^{-3/2}$ is integrable at $z=\infty$ and 
changes sign across all branch cuts, we obtain, for $\Im\{m\}\neq 0$,
\begin{equation}
j(m)=2\,\mathrm{sgn}(\Im\{m\})\int_{G(0)^2}^m\frac{T(m,z)\,dz}{(4-z)^{3/2}},
\end{equation}
where the path of integration is a straight line in the region of analyticity 
of the integrand.  Since $\mathrm{sgn}(\Im\{m\})$ and $\mathrm{sgn}(\Im\{z\})$
coincide, we may absorb the sign by changing the branch of $(4-z)^{3/2}$
in the lower-half $z$-plane:
\begin{equation}
j(m)=2\int_{G(0)^2}^m\frac{T(m,z)\,dz}{i(z-4)^{3/2}}.
\end{equation}
With the substitution
\begin{equation}
\alpha:=\left(1-\frac{z-G(0)^2}{m-G(0)^2}\right)^{1/2},\quad 0<\alpha<1,
\end{equation}
the formula for $j(m)$ becomes
\begin{equation}
\begin{split}
j(m)&=-4\int_0^1(m-4-(m-G(0)^2)\alpha^2)^{-3/2}\,d\alpha\\
&=\frac{4}{4-m}\int_0^1\frac{d}{d\alpha}\left[\alpha(m-4-(m-G(0)^2)\alpha^2)^{-1/2}\right]\,d\alpha\\
&=\frac{4}{(4-m)\sqrt{G(0)^2-4}}.
\end{split}
\end{equation}
Inserting this formula into \eqref{eq:nujint} and taking into account that
the only singularity of the integrand enclosed by the contour $O$ is
the simple pole at $m=4$ yields
\begin{equation}
I_4+\frac{1}{2}x_\mathrm{crit} = \frac{\mathscr{G}(4)}{12\sqrt{G(0)^2-4}}.
\end{equation}
Evaluating $\mathscr{G}(4)$ directly using the definition \eqref{eq:hmdefcrit} then proves
that indeed $I_4+\tfrac{1}{2}x_\mathrm{crit}=\nu$ where $\nu>0$ is given by \eqref{eq:nudefinecrit}.

%
Using \eqref{eq:I1zerocrit}, \eqref{eq:I5identity}, \eqref{eq:I2xcrit}, and \eqref{eq:I4formulacrit}
in \eqref{eq:termbytermcrit}, and combining the resulting expansion with the expansions \eqref{eq:elementaryTaylorcrit} and the analytic formula \eqref{eq:kanalyticcrit} for $k(w;x,t)$
completes the proof of the Proposition.
\end{proof}

\subsection{Conformal mapping near $w=-1$ and new spacetime coordinates.}
From \eqref{eq:bexpansion} it is clear that
exactly at criticality, $k(w) = \nu(w+1)^3
+\bo((w+1)^4)$.  For small $|\Delta x|$ and $|t|$, the cubic
degeneration unfolds, with one real and two complex conjugate roots,
or three real roots,
near $w=-1$.  The double critical point (double root of $k'(w)$)
unfolds generically to a pair $w=w_\pm(x,t)$ 
of simple critical points near $w=-1$, and either these are both
real or they form a complex-conjugate pair (in both cases 
since $k(w)=k(w^*)^*$, $(k(w_+)-k(w_-))^2\in\mathbb{R}$).

The main idea here is that to unfold the cubic degeneracy it is really only necessary to take into
account the lower-order terms in the Taylor expansion, and therefore it seems desirable to 
somehow replace $k(w;x,t)$ by an appropriate cubic polynomial with coefficients depending on
$(x,t)$.  This issue has arisen frequently in the construction of local parametrices for matrix
Riemann-Hilbert problems corresponding to certain \emph{double-scaling limits}.  Perhaps the first time such a replacement was made rigorous was 
in the paper of Baik, Deift, and Johansson \cite{BaikDJ99}, in which a certain double-scaling limit of
orthogonal polynomials on the unit circle is analyzed, and the authors construct
a local parametrix by essentially truncating an analogue of the Taylor expansion \eqref{eq:bexpansion} after the cubic term.  Of course such a truncation is not exact, so
there are errors incurred in modeling the jump matrices by others having cubic exponents, and the  effect of these errors must be carefully controlled.  Also, any artificial truncation of a Taylor series can only be accurate if the local parametrix is constructed in a disk centered at the expansion point whose radius tends to zero at some rate tied to the large parameter in the problem.  This fact further implies that estimates must be supplied for norms of singular integral operators that are independent of the moving contour and it also 
typically means that divergence of the outer parametrix near the expansion point can lead to a
sub-optimal mismatch with the local parametrix, leading in turn to sub-optimal estimates of errors.
A significant advance was made by Claeys and Kuijlaars in \cite{ClaeysK06} (see in particular section 5.6 of that paper); here the authors consider a similar double-scaling limit and take the 
approach of constructing a certain nontrivial conformal mapping $W=W(w)$ of a fixed neighborhood of the expansion point to a neighborhood of the origin.  The conformal mapping is
 more-or-less explicit, and it does not depend on the parameters (analogues of $x$ and $t$) driving the system to criticality.  However, the point is that in a full fixed neighborhood of the expansion point, the analogue of the function $k(w;x,t)$ is represented exactly (no truncation required) as a cubic polynomial in $W$.  The difficulty that remains with this approach is that
 the coefficients of the cubic in $W$ actually depend on $w$ (so in fact it is not really a cubic after all).  The remarkable approach of \cite{ClaeysK06} is to simply substitute the $w$-dependent 
 coefficients into a known solution formula for exactly cubic exponents (involving solutions of
 the linear differential equation whose isomonodromy deformations are governed by solutions
 of the Painlev\'e-II equation), resulting in a local parametrix in terms of Painlev\'e transcendents depending (through the coefficients of the cubic) on $w$.  Confirming that such an approach
 actually provides a usable local parametrix requires exploiting \emph{a priori} knowledge of the
 behavior of solutions of the nonlinear Painlev\'e-II equation and its auxiliary linear differential equation.  If this information is available in a convenient form, the approach of Claeys and Kuijlaars delivers a vast improvement over earlier methods because it really works in a neighborhood of fixed size centered at the expansion point where critical points coalesce at criticality.    The technique advanced in \cite{ClaeysK06} has more recently been applied to
 problems of nonlinear wave theory as well 
 (see, for example, 
 \cite{ClaeysG10osc}).

We choose instead to replace $k(w;x,t)$ by a cubic in an exact way, an approach that provides
all of the accuracy of the Claeys-Kuijlaars method but seems simpler and requires no \emph{a priori} knowledge of the behavior of solutions of the local parametrix Riemann-Hilbert problem.
The approach we are going to explain now has also been used recently to study a different kind of double-scaling limit for a
matrix Riemann-Hilbert problem in \cite{BertolaT10cusp}.
As part of a careful study of the asymptotic behavior of exponential integrals with exponent functions having coalescing critical points (to generalize the steepest descent or saddle point method), Chester, Friedman, and Ursell \cite{ChesterFU57}
showed how to construct a substitution that rendered the exponent function in the integrand in the exact form
of a cubic polynomial.  Their method applies in the current context to establish that, \emph{because the coefficient $\nu$ in the expansion \eqref{eq:bexpansion} of $k(w;x,t)$
is strictly positive}, there is a suitable choice of new spacetime coordinates
$r=r(x,t)$ and $s=s(x,t)$ for which the relation
\begin{equation}
k(w;x,t)=W^3 +rW -s
\label{eq:langercondition}
\end{equation}
defines an invertible conformal mapping $W=W(w)=W(w;x,t)$ between $w\in U$ and
$W\in W(U)$ that preserves the real axis, for $(x,t)$ 
near enough to criticality.
Moreover, the new coordinates $r$ and $s$
depend continuously on $(x,t)$ near criticality.  Unlike in the Claeys-Kuijlaars method \cite{ClaeysK06}, neither $r$ nor $s$ depends on $w$ with the cost that the conformal
mapping will now depend on $(x,t)$.

The new coordinates $r=r(x,t)$ and $s=s(x,t)$
are to be determined so that under 
\eqref{eq:langercondition} the critical points of the cubic on the 
right-hand side, namely $W=\pm (-r/3)^{1/2}$, correspond to the two critical
points of $k(w)$ near $w=-1$ when $\Delta x$ and $t$ are
sufficiently small. Evaluating \eqref{eq:langercondition} for 
$w=w_\pm(x,t)$ and $W=\pm (-r(x,t)/3)^{1/2}$, one obtains the formulae
\begin{equation}
\begin{split}
r(x,t)&=-3\left(\frac{1}{16}\left[
k(w_+(x,t);x,t)-k(w_-(x,t);x,t)\right]^2\right)^{1/3}\\
s(x,t)&=-\frac{1}{2}\left[k(w_+(x,t);x,t)+k(w_-(x,t);x,t)\right],
\end{split}
\label{eq:rsdefine}
\end{equation}
where in the formula for $r$ the real cube root is meant, and hence
both $r$ and $s$ are real.  Moreover, it is possible to show that $r(x,t)$
and $s(x,t)$ are analytic functions of $x$ and $t$ near criticality, and 
have two-variable Taylor expansions of the form:
\begin{equation}
\begin{split}
r(x,t)&=\frac{1}{2\nu^{1/3}}
\Delta x +
\bo(\Delta x^2,t\Delta x,t^2)\\
s(x,t)&=t +\bo(\Delta x^2,t\Delta x,t^2).
\end{split}
\label{eq:rsexpansions}
\end{equation}
To see this, first we find a unique root $w=-1+a$ of $k''(w;x,t)$ such that $a=0$ at criticality and
such that $a=a(x,t)$ is an analytic function of $(x,t)$ near criticality.  Indeed, from the Taylor expansion
\eqref{eq:bexpansion} and the fact that $\nu\neq 0$ we see that the analytic implicit function theorem 
applies and we obtain 
\begin{equation}
a(x,t)=-\frac{1}{12\nu}\Delta x + \frac{1}{24\nu}t +  \bo(\Delta x^2,t\Delta x,t^2).
\end{equation}
Now write $w=-1+a(x,t)+z$, and re-expand $k(w;x,t)$ about $z=0$.  From \eqref{eq:bexpansion}
and the definition of $a(x,t)$ we obtain the convergent power series expansion (note that there is no quadratic term) 
\begin{equation}
k(w;x,t)=k_0(x,t)+k_1(x,t)z + \sum_{n=3}^\infty k_n(x,t)z^n,
\label{eq:kseriesz}
\end{equation}
where the coefficients are all analytic functions of $(x,t)$ near criticality, and in particular,
\begin{equation}
k_0(x,t)=-t + \bo(\Delta x^2,t\Delta x,t^2),
\end{equation}
\begin{equation}
k_1(x,t)=\frac{1}{2}\Delta x+ \bo(\Delta x^2,t\Delta x,t^2),
\end{equation}
\begin{equation}
k_3(x,t)=\nu+ \bo(\Delta x,t).
\end{equation}
The radius of convergence of this series is bounded away from zero near criticality.
The equation satisfied by the critical points of $k$ is
\begin{equation}
k'(w;x,t)=k_1(x,t)+3k_3(x,t)z^2 +\sum_{n=4}^\infty nk_n(x,t)z^{n-1}=0.
\label{eq:criticalintermsofz}
\end{equation}
Let $\rho$ be any number satisfying the equation
\begin{equation}
\rho^2 = -\frac{k_1(x,t)}{3k_3(x,t)}.
\label{eq:rhosquareddef}
\end{equation}
Note that $\rho^2$ is an analytic function of $(x,t)$ near criticality because $k_3(x_\mathrm{crit},0)=\nu>0$.  If $\rho=0$ then $k_1(x,t)=0$ and $z=0$ is a double root of \eqref{eq:criticalintermsofz}.
So we suppose that $\rho\neq 0$, and we rescale $z$ by writing $z=\rho q$ for some new unknown
$q$.  Dividing through by $3k_3(x,t)\neq 0$ and canceling a factor of $\rho^2$ then converts \eqref{eq:criticalintermsofz} into the equation
\begin{equation}
q^2 + \sum_{n=4}^\infty \frac{n k_n(x,t)}{3k_3(x,t)}\rho^{n-3}q^{n-1}=1.
\end{equation}
The coefficients of $q^{n-1}$ in the sum are analytic functions of the three variables $x$, $t$, and $\rho$.  When $\rho=0$ there are two solutions, $q=\pm 1$, and by the implicit function theorem there are two corresponding solutions for $\rho\approx 0$ and $(x,t)$ near criticality.  The two solutions are related
by the symmetry $(\rho,q)\mapsto (-\rho,-q)$.  We may write them in the form
\begin{equation}
q=q_\pm(x,t,\rho)=\pm \left(1 + \sum_{n=1}^\infty q_n(x,t)(\pm\rho)^n\right).
\end{equation}
The coefficients $q_n(x,t)$ are all analytic functions of $(x,t)$ near criticality, and the radius of convergence of the series with respect to $\rho$ is bounded away from zero near criticality.
The corresponding critical points are written in terms of $z$ as
\begin{equation}
z=z_\pm(x,t,\rho)=\pm\rho + \sum_{n=1}^\infty q_n(x,t)(\pm\rho)^{n+1}.
\label{eq:criticalpointsz}
\end{equation}
To calculate $r(x,t)$ and $s(x,t)$ we need to evaluate $k(w;x,t)$ at the critical points.  This is accomplished by substitution of the series \eqref{eq:criticalpointsz} into \eqref{eq:kseriesz}:
\begin{equation}
\begin{split}
k(-1+a(x,t)+z_\pm(x,t,\rho);x,t)&=k_0(x,t)+k_1(x,t)\left(\pm\rho +\sum_{n=1}^\infty q_n(x,t)(\pm\rho)^{n+1}\right) \\
&\quad\quad{}+\sum_{m=3}^\infty k_m(x,t)\left(\pm\rho +\sum_{n=1}^\infty q_n(x,t)(\pm\rho)^{n+1}\right)^m\\&=
k_0(x,t) \pm\rho^3\sum_{n=0}^\infty g_n(x,t)(\pm\rho)^n,
\end{split}
\end{equation}
where we have used \eqref{eq:rhosquareddef}, and $g_0(x,t):=-2k_3(x,t)$.  Here again, the
coefficients $g_n(x,t)$ are certain analytic functions of $(x,t)$ near criticality, and the radius of
convergence of the power series in $\pm\rho$ is bounded below near criticality.  Then, by definition,
we have
\begin{equation}
\begin{split}
r(x,t)&=-3\left(\frac{1}{16}\left[2g_0(x,t)\rho^3\left(1+\sum_{j=1}^\infty\frac{g_{2j}(x,t)}{g_0(x,t)}(\pm\rho)^{2j}
\right)\right]^2\right)^{1/3}\\
&= \frac{k_1(x,t)}{k_3(x,t)^{1/3}}\left(1-\frac{1}{2}\sum_{j=1}^\infty\frac{g_{2j}(x,t)}{k_3(x,t)}\left[-\frac{k_1(x,t)}{3k_3(x,t)}\right]^j\right)^{2/3},
\end{split}
\end{equation}
where we have used the definition of $g_0(x,t)$ in terms of $k_3(x,t)$ and \eqref{eq:rhosquareddef}.  Therefore, $r(x,t)$ is clearly an analytic function of $(x,t)$ near criticality.  Similarly, by definition we have
\begin{equation}
\begin{split}
s(x,t)&=-\frac{1}{2}\left[2k_0(x,t)+2\sum_{j=0}^\infty g_{2j+1}(x,t)(\pm\rho)^{2j+4}\right]\\
&=-k_0(x,t) -\sum_{j=0}^\infty g_{2j+1}(x,t)\left(-\frac{k_1(x,t)}{3k_3(x,t)}\right)^{j+2},
\end{split}
\end{equation}
which again is obviously an analytic function of $(x,t)$ near criticality.  The leading terms of
$r(x,t)$ and $s(x,t)$ near criticality are easy to calculate from these formulae, with the result being
\eqref{eq:rsexpansions}.

As $W(\cdot;x,t)$ is a conformal map, its inverse is an analytic
function mapping a neighborhood of $W=0$ to a neighborhood of $w=-1$.
We now define
\begin{equation}
w_*=w_*(x,t):=W^{-1}(0;x,t),
\label{eq:wstardefine}
\end{equation}
which is a real analytic function of $x$ and $t$ near criticality.
Note also that $W'(w_*(x,t);x,t)>0$ (prime denotes differentiation with
respect to $w$) near criticality.  It is not difficult to obtain the
Taylor expansion of the conformal mapping $W(w;x,t)$ about $w=w_*$
exactly at criticality.  Indeed, from \eqref{eq:bexpansion} and
\eqref{eq:langercondition} with $\Delta x=t=r=s=0$, we obtain
the equation
\begin{equation}
W(w)^3 = \nu(w+1)^3\left[1+\frac{3}{2}(w+1) +\bo((w+1)^2)\right],
\quad\text{at criticality},
\end{equation}
and analytically blowing up the cubic degeneracy we obtain
\begin{equation}
W(w)=\nu^{1/3}(w+1)\left[1+\frac{1}{2}(w+1) +\bo((w+1)^2)\right],
\quad\text{at criticality}.
\end{equation}
In particular, this implies that
\begin{equation}
W'(w_*)=W''(w_*)=\nu^{1/3}>0,\quad\text{at criticality}.
\label{eq:WprimeWdoubleprimecrit}
\end{equation}
Also, it is clear that $W(-1;x_\mathrm{crit},0)=0$, and therefore
$w_*(x_\mathrm{crit},0)=-1$.  Since $w_*(x,t)$ is an analytic function
of $x$ and $t$, it then follows that
\begin{equation}
w_*(x,t)=-1+\bo(\Delta x,t),
\label{eq:wstarexpand}
\end{equation}
and furthermore,
\begin{equation}
W'(w_*(x,t);x,t)=\nu^{1/3}+\bo(\Delta x,t)\quad\text{and}\quad
W''(w_*(x,t);x,t)=\nu^{1/3}+\bo(\Delta x,t)
\end{equation}
near criticality.

Now let $\zeta$ and $y$ be scaled versions
of $W$ and $r$ respectively:
\begin{equation}
\zeta:=\frac{W}{\epsilon_N^{1/3}}\quad \text{and}
\quad y:=
\frac{r}{\epsilon_N^{2/3}}.
\label{eq:scalings}
\end{equation}
We will regard $y$ as being bounded.  It then follows that \emph{without approximation} the exponent
appearing in the jump matrix $\mathbf{V}_\mathbf{P}(w)$ takes the form
of a simple cubic with a formally large constant term:
\begin{equation}
\frac{k(w;x,t)}{\epsilon_N} = \zeta^3 + y\zeta -\frac{s}{\epsilon_N}.
\label{eq:bzeta}
\end{equation}

\subsection{Formulation of an inner model problem.}
To begin with, we wish to find a simpler representation for 
$\dot{\mathbf{P}}_\ind^\mathrm{out}(w)$ valid
when $w$ is close to $w=w_*\approx -1$. 
Using \eqref{eq:scalings}, we write $\dot{\mathbf{P}}_\ind^\mathrm{out}(w)$ in
the form 
\begin{equation}
\dot{\mathbf{P}}_\ind^\mathrm{out}(w)=\eta_\ind(w)^{\sigma_3}(-W(w))^{(\ind-1/2)\sigma_3} = 
\eta_\ind(w)^{\sigma_3}\epsilon_N^{(2\ind-1)\sigma_3/6}(-\zeta)^{(2\ind-1)\sigma_3/2},
\label{eq:dotPrepDS}
\end{equation}
where $\eta_\ind(w)$ is independent
of $\epsilon_N$ and is analytic and nonvanishing 
in a neighborhood $U$ of $w=w_*$:
\begin{equation}
\eta_\ind(w):=(\sqrt{-w}+\sqrt{-w_*})^{1-2\ind}
\left(\frac{W(w)}{w-w_*}\right)^{1/2-\ind}.
\label{eq:Hnformula}
\end{equation}
(Analyticity follows since from \eqref{eq:wstardefine} we have $W(w_*)=0$
and $W$ is analytic at $w_*$ with $W'(w_*)>0$.)

Given a value of $\ind\in\mathbb{Z}$ (to be determined below), we wish to construct
an inner model, valid for $w\in U$, of the matrix $\tilde{\mathbf{P}}(w)$ 
(related to 
$\mathbf{P}(w)$ for $w\in U$ via \eqref{eq:tildeQdefine}).
We temporarily denote this model as $\mathbf{W}(\zeta(w))$, and we want it
to have the following properties:
\begin{itemize}
\item
Supposing that in $U$ the lens boundaries are identified with
arcs of the curves $\arg(W(w))=\pm \pi/3$ and $\arg(W(w))=\pm 2\pi/3$, which
makes them segments (of length proportional to $\epsilon_N^{-1/3}$)
of straight rays in the $\zeta$-plane, the inner model
$\mathbf{W}(\zeta(w))$ should be analytic exactly where $\tilde{\mathbf{P}}(w)$
is within $U$ and should satisfy \emph{exactly the same jump 
conditions} that $\tilde{\mathbf{P}}(w)$ does
within $U$.
\item
The inner model $\mathbf{W}(\zeta)$ 
should match onto the latter factors 
in \eqref{eq:dotPrepDS} 
along the disc boundary $\partial U$
in the sense that $\mathbf{W}(\zeta)$ may be analytically continued from
each sector of the domain $\zeta\in\epsilon_N^{-1/3}W(U)$ 
to the corresponding infinite sector in
the $\zeta$-plane, and that
\begin{equation}
\lim_{\zeta\to\infty}\mathbf{W}(\zeta)
\epsilon_N^{(1-2\ind)\sigma_3/6}(-\zeta)^{(1-2\ind)\sigma_3/2}=\mathbb{I},
\label{eq:Qdotnorm}
\end{equation}
with the limit being uniform with respect to direction in each of the
six sectors of analyticity.  Since $\partial U$ is fixed and bounded
away from $w_*\approx -1$, upon scaling its image under $W(\cdot;x,t)$
by $\epsilon_N^{-1/3}$ to work in terms of the variable $\zeta$,
we see that $w\in\partial U$ corresponds to $\zeta\to\infty$ at a
uniform rate of $\epsilon_N^{-1/3}$.
\end{itemize}

Since $\mathbf{W}(\zeta(w))$ depends on $w$ through $\zeta(w)$, it is
convenient to use \eqref{eq:bzeta} to 
write the jump matrices for $\tilde{\mathbf{P}}(w)$ (and hence
also for $\mathbf{W}(\zeta(w))$) in terms
of $\zeta$, which shows that the jump matrices 
involve the product of exponentials
$e^{\mp s/\epsilon_N}e^{\pm (\zeta^3+y\zeta)}$.
The constant ($\zeta$-independent) 
factors $e^{\mp s/\epsilon_N}$ present in the jump matrices 
can be removed, and simultaneously the presence
of $\epsilon_N$ in the normalization condition \eqref{eq:Qdotnorm}
can be eliminated, by defining
the equivalent unknown
\begin{equation}
\mathbf{Z}_\ind(\zeta;y):=
\epsilon_N^{(1-2\ind)\sigma_3/6}e^{-\frac{1}{2}s\sigma_3/\epsilon_N}
\mathbf{W}(\zeta)e^{\frac{1}{2}s\sigma_3/\epsilon_N}.
\end{equation}
The conditions previously discussed as being desirable for
$\mathbf{W}(\zeta)$ are then easily seen to be equivalent to
those of the following problem for
$\mathbf{Z}_\ind(\zeta;y)$:
\begin{rhp}[Inner model problem near criticality]
Let a real number $y$ and an integer $\ind$ be fixed.  
Seek a matrix $\mathbf{Z}_\ind(\zeta;y)$ 
with the following properties:
\begin{itemize}
\item[]\textbf{Analyticity:}  $\mathbf{Z}_\ind(\zeta;y)$
is analytic in $\zeta$ except along the rays $\arg(\zeta)=k\pi/3$, 
$k=0,\dots,5$, from each sector of analyticity it may be continued to
a slightly larger sector, and in each sector is H\"older continuous 
up to the boundary in a neighborhood of $\zeta=0$.
\item[]\textbf{Jump condition:}  The jump conditions satisfied by 
the matrix function $\mathbf{Z}_\ind(\zeta;y)$ are of the form 
$\mathbf{Z}_{\ind+}(\zeta;y)=
\mathbf{Z}_{\ind-}(\zeta;y)
\mathbf{V}_{\mathbf{Z}}(\zeta;y)$, 
where the jump matrix $\mathbf{V}_{\mathbf{Z}}(\zeta;y)$ 
is as shown in 
Figure~\ref{fig:VlocR} 
\begin{figure}[h]
\begin{center}
\includegraphics{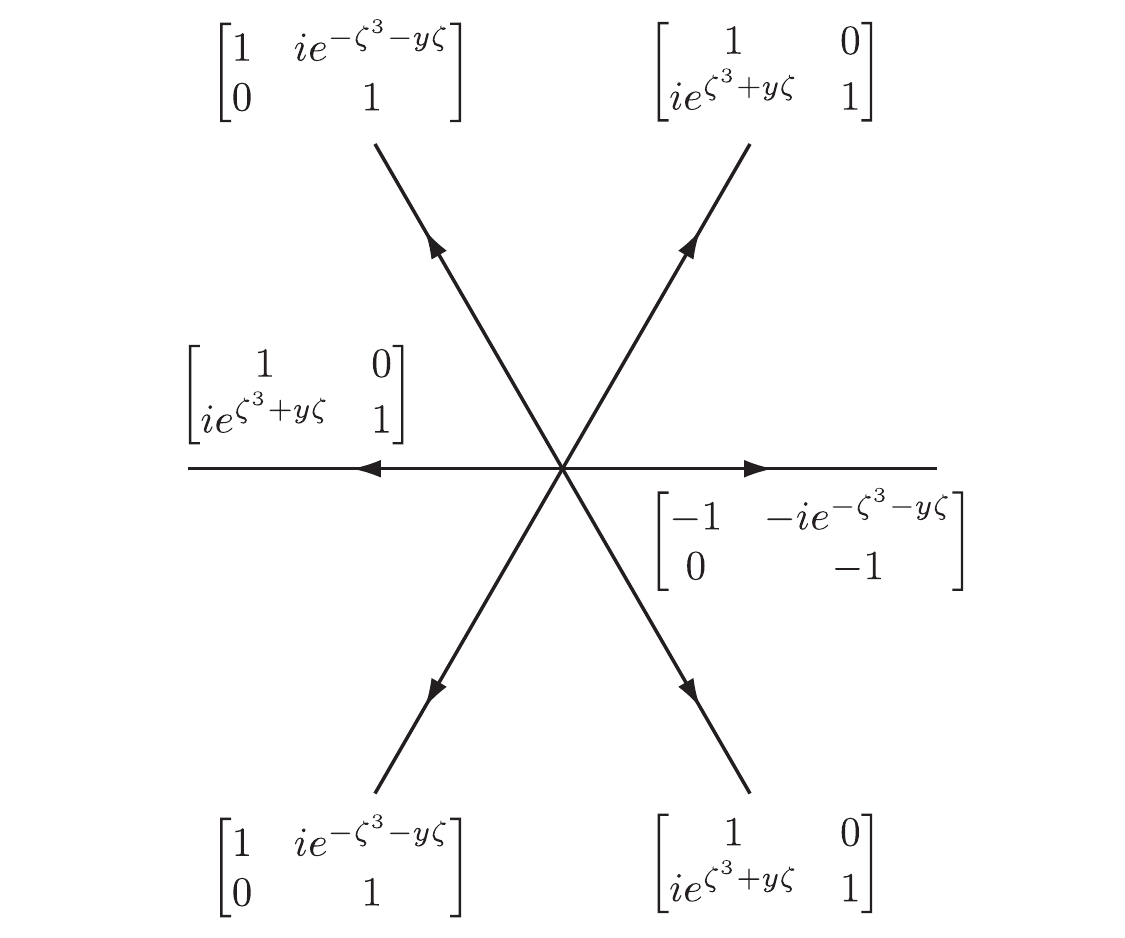}
\end{center}
\caption{\emph{The jump matrix $\mathbf{V}_{\mathbf{Z}}(\zeta;y)$.}}
\label{fig:VlocR}
\end{figure}
and all rays are oriented toward infinity.
\item[]\textbf{Normalization:}  The matrix 
$\mathbf{Z}_\ind(\zeta;y)$
satisifies the condition
\begin{equation}
\lim_{\zeta\to\infty}\mathbf{Z}_\ind(\zeta;y)
(-\zeta)^{(1-2\ind)\sigma_3/2}=\mathbb{I},
\label{eq:Zindnorm}
\end{equation}
with the limit being uniform with respect to direction in each of the
six sectors of analyticity.
\end{itemize}
\label{rhp:DSlocalII}
\end{rhp}

Suppose that Riemann-Hilbert Problem~\ref{rhp:DSlocalII} has a unique solution
for $y\in (y_-,y_+)\subset\mathbb{R}$ and for some integer $\ind$.  
We will now describe how to use it to create
a model for $\mathbf{P}(w)$ valid in the neighborhood $U$ assuming
that $y\in (y_-,y_+)$;  
we set
\begin{equation}
\begin{split}
\dot{\mathbf{P}}_\ind^\mathrm{in}(w):={}&
\eta_\ind(w)^{\sigma_3}\mathbf{W}(\zeta(w))d_N(w)^{-\sigma_3}\\
={}& \eta_\ind(w)^{\sigma_3}\epsilon_N^{(2\ind-1)\sigma_3/6}e^{\frac{1}{2}s\sigma_3/\epsilon_N}
\mathbf{Z}_\ind(\zeta(w);y)e^{-\frac{1}{2}s\sigma_3/\epsilon_N}
d_N(w)^{-\sigma_3},\quad w\in U.
\end{split}
\label{eq:dotPnlocdefine}
\end{equation}
The effect of post-multiplication by $d_N(w)^{-\sigma_3}$ is simply to restore the factors involving $T_N(w)$ to the jump conditions.  We emphasize that within the fixed neighborhood $U$ of $w=-1$, 
the inner model matrix
$\dot{\mathbf{P}}_\ind^\mathrm{in}(w)$ satisfies exactly the same jump conditions 
along the six contour arcs meeting at $w=w_*$ as does $\mathbf{P}(w)$.
It appears reasonable to propose a 
global model for $\mathbf{P}(w)$ in the following form:
\begin{equation}
\dot{\mathbf{P}}_\ind(w):=
\begin{cases}
\dot{\mathbf{P}}_\ind^\mathrm{in}(w),&\quad w\in U\\
\dot{\mathbf{P}}_\ind^\mathrm{out}(w),&\quad w\not\in \overline{U}.
\end{cases}
\label{eq:dotPdefineDS}
\end{equation}
The integer $\ind$ is evidently at our disposal.  We will later describe how it
should be chosen.

\section{Solution of the Inner Model Problem}
\subsection{Symmetry analysis of the inner model problem.}
While it may not yet be clear whether there exists a solution of
Riemann-Hilbert Problem~\ref{rhp:DSlocalII} for any $y\in\mathbb{R}$ 
and $\ind\in\mathbb{Z}$ at all, 
it is a standard argument that there exists at most one solution
for each $y\in\mathbb{R}$ and $\ind\in\mathbb{Z}$, 
and that every solution must satisfy
$\det(\mathbf{Z}_\ind(\zeta;y))\equiv 1$.  In this short
section we suppose that $y\in\mathbb{R}$ and $\ind\in\mathbb{Z}$ 
are values for which
Riemann-Hilbert Problem~\ref{rhp:DSlocalII} has a (unique and
unimodular) solution, and we examine some of the consequences of the
existence.  We will later show that $\mathbf{Z}_\ind(\zeta;y)
(-\zeta)^{(1-2\ind)\sigma_3/2}$
has a complete asymptotic expansion in descending powers of $\zeta$
as $\zeta\to\infty$:
\begin{equation}
\mathbf{Z}_\ind(\zeta;y)(-\zeta)^{(1-2\ind)\sigma_3/2}=
\mathbb{I}+\mathbf{A}_\ind(y)\zeta^{-1}+
\mathbf{B}_\ind(y)\zeta^{-2}+\mathbf{C}_\ind(y)\zeta^{-3}+\bo(\zeta^{-4}),
\quad\zeta\to\infty
\label{eq:Fexpansion}
\end{equation}
and where the matrix coefficients are the same regardless of which
of the six sectors of analyticity of $\mathbf{Z}_\ind(\zeta;y)$
is used to compute the expansion.

Given a solution $\mathbf{Z}_\ind(\zeta;y)$ of
Riemann-Hilbert Problem~\ref{rhp:DSlocalII}, consider the related matrix
\begin{equation}
\mathbf{F}(\zeta):=\mathbf{Z}_\ind(\zeta^*;y)^*.
\end{equation}
Since the jump contour for $\mathbf{Z}_\ind(\zeta;y)$ is
invariant under complex conjugation, $\mathbf{F}(\zeta)$ will be
analytic in the same domain that
$\mathbf{Z}_\ind(\zeta;y)$ is.  It is a direct calculation
that $\mathbf{F}(\zeta)$ satisfies exactly the same jump conditions as
does $\mathbf{Z}_\ind(\zeta;y)$.  Therefore, the matrix
$\mathbf{F}(\zeta)\mathbf{Z}_\ind(\zeta;y)^{-1}$ extends
continuously to the jump contour from both sides of each arc.  It
follows from the classical sense in which the boundary values are
attained (even at the self-intersection point $\zeta=0$), and from the
fact that $\mathbf{Z}_\ind(\zeta;y)$ is unimodular, that
in fact $\mathbf{F}(\zeta)\mathbf{Z}_\ind(\zeta;y)^{-1}$ 
is an entire function of $\zeta$.  Moreover, 
\begin{equation}
\begin{split}
\lim_{\zeta\to\infty}\mathbf{F}(\zeta)
\mathbf{Z}_\ind(\zeta;y)^{-1}
&=
\lim_{\zeta\to\infty}\mathbf{F}(\zeta)(-\zeta)^{(1-2\ind)\sigma_3/2}\cdot
\left[\mathbf{Z}_\ind(\zeta;y)(-\zeta)^{(1-2\ind)\sigma_3/2}
\right]^{-1}\\
&{}=\lim_{\zeta\to\infty}\mathbf{Z}_\ind(\zeta^*;y)^*
(-\zeta)^{(1-2\ind)\sigma_3/2}
=\lim_{\zeta\to\infty}\left[\mathbf{Z}_\ind(\zeta^*;y)
(-\zeta^*)^{(1-2\ind)\sigma_3/2}\right]^* = \mathbb{I},
\end{split}
\end{equation}
so by Liouville's Theorem, 
$\mathbf{Z}_\ind(\zeta^*;y)^*=
\mathbf{F}(\zeta)=\mathbf{Z}_\ind(\zeta;y)$.  Applying this
symmetry to the expansion \eqref{eq:Fexpansion} shows that the elements
of the matrices $\mathbf{A}_\ind(y)$ and $\mathbf{B}_\ind(y)$
(and in fact all of the matrix expansion coefficients) are real numbers.

We can obtain more detailed information by exploiting a further symmetry in
the special case of $\ind=0$.
Indeed, let us compare $\mathbf{Z}_0(\zeta;y)$ with the matrix
\begin{equation}
\mathbf{F}(\zeta):=\begin{cases}
\mathbf{Z}_0(-\zeta;y)i\sigma_1, &\quad \Im\{\zeta\}>0\\
-\mathbf{Z}_0(-\zeta;y)i\sigma_1,&\quad \Im\{\zeta\}<0.
\end{cases}
\label{eq:FQminus}
\end{equation}
Clearly, $\det(\mathbf{F}(\zeta))\equiv 1$, and it is a direct matter
to check that $\mathbf{F}(\zeta)$ is analytic precisely where
$\mathbf{Z}_0(\zeta;y)$ is, and $\mathbf{F}(\zeta)$ 
satisfies exactly the same jump conditions on the six rays of Riemann-Hilbert 
Problem~\ref{rhp:DSlocalII} as does $\mathbf{Z}_0(\zeta;y)$.
Therefore, $\mathbf{F}(\zeta)\mathbf{Z}_0(\zeta;y)^{-1}$ is
an entire unimodular matrix function of $\zeta$.  To identify this entire
function as a polynomial, it is enough to extract the non-decaying terms
in its asymptotic expansion as $\zeta\to\infty$;  Assuming without
loss of generality that $\Im\{\zeta\}>0$ so that 
$(-\zeta)^{\sigma_3/2}=-i\sigma_3\zeta^{\sigma_3/2}$, we have
\begin{equation}
\begin{split}
\mathbf{F}(\zeta)\mathbf{Z}_0(\zeta;y)^{-1}&=
\mathbf{F}(\zeta)(-i\sigma_3)\zeta^{\sigma_3/2}\cdot
\left[\mathbf{Z}_0(\zeta;y)(-\zeta)^{\sigma_3/2}\right]^{-1}\\
&{}=\mathbf{Z}_0(-\zeta;y)\sigma_1\sigma_3\zeta^{\sigma_3/2}
\cdot\left(\mathbb{I}-\mathbf{A}_0(y)\zeta^{-1} + \bo(\zeta^{-2})\right)\\
&{}=\mathbf{Z}_0(-\zeta;y)\zeta^{\sigma_3/2}
\begin{bmatrix}0 & -\zeta^{-1}\\\zeta & 0\end{bmatrix}
\left(\mathbb{I}-\mathbf{A}_0(y)\zeta^{-1} + \bo(\zeta^{-2})\right)\\
&{}=\left(\mathbb{I}-\mathbf{A}_0(y)\zeta^{-1}+\bo(\zeta^{-2})\right)
\begin{bmatrix}0 & -\zeta^{-1}\\\zeta & 0\end{bmatrix}
\left(\mathbb{I}-\mathbf{A}_0(y)\zeta^{-1}+\bo(\zeta^{-2})\right)\\
&=\sigma_-\zeta
-\mathbf{A}_0(y)\sigma_--
\sigma_-\mathbf{A}_0(y)
+\bo(\zeta^{-1}),\quad\zeta\to\infty,
\end{split}
\label{eq:FQminusexpansion}
\end{equation}
where
\begin{equation}
\sigma_+:=\begin{bmatrix}0 & 1\\0 & 0\end{bmatrix}\quad\text{and}\quad
\sigma_-:=\begin{bmatrix}0 & 0\\1 & 0\end{bmatrix},
\end{equation}
and so $\mathbf{F}(\zeta)\mathbf{Z}_0(\zeta;y)^{-1}$
may be identified with the polynomial terms in this expansion; therefore
we arrive at the identity
\begin{equation}
\mathbf{F}(\zeta)\mathbf{Z}_0(\zeta;y)^{-1}=\begin{bmatrix}-A_{0,12}(y) & 0\\\zeta-A_{0,11}(y)-A_{0,22}(y) &
-A_{0,12}(y)\end{bmatrix}.
\end{equation}
Since both sides represent unimodular matrices, we learn that
\begin{equation}
A_{0,12}(y)^2 = 1.
\label{eq:A12squaredeq1}
\end{equation}
We will later obtain a refinement of
    \eqref{eq:A12squaredeq1}, namely that 
$A_{0,12}(y)= 1$ holds for all real $y$.  The sign
information will follow from the explicit solution of 
Riemann-Hilbert Problem~\ref{rhp:DSlocalII} for $\ind=0$ in terms
of special functions, which will be possible precisely as a consequence of the
identity \eqref{eq:A12squaredeq1}.

\subsection{Differential equations derived from the inner model problem.  
Lax equations and Painlev\'e-II}
The only parameters in the Riemann-Hilbert problem characterizing
$\mathbf{Z}_\ind(\zeta;y)$ are $\ind\in\mathbb{Z}$ and
$y\in\mathbb{R}$.  Assuming existence
for some fixed $\ind\in\mathbb{Z}$ and for $y$ in some open set, we will now
investigate some consequences of the dependence of the solution on $y$.
We suppose that the expansion \eqref{eq:Fexpansion}
holds in the stronger sense that 
the following are also true:
\begin{equation}
\frac{\partial}{\partial y}\left[\mathbf{Z}_\ind(\zeta;y)
(-\zeta)^{(1-2\ind)\sigma_3/2}
\right]=\bo(\zeta^{-1})\quad\text{and}\quad
\frac{\partial}{\partial\zeta}\left[\mathbf{Z}_\ind(\zeta;y)
(-\zeta)^{(1-2\ind)\sigma_3/2}
\right]=\bo(\zeta^{-2}),\quad\zeta\to\infty.
\end{equation}
The matrix $\mathbf{L}_\ind(y,\zeta)$ defined by
\begin{equation}
\mathbf{L}_\ind(y,\zeta):=\mathbf{Z}_\ind(\zeta;y)
e^{-(\zeta^3+y\zeta)\sigma_3/2}
\label{eq:GFrelation}
\end{equation}
has jump discontinuities along the six rays mediated by jump matrices that are independent of both $y$ and $\zeta$.  Given
the compatibility of the jump matrices at $\zeta=0$, this implies that
the matrices
\begin{equation}
\mathbf{U}_\ind(y,\zeta):=\frac{\partial\mathbf{L}_\ind}{\partial y}\mathbf{L}_\ind^{-1}
\quad\text{and}
\quad
\mathbf{V}_\ind(y,\zeta):=\frac{\partial\mathbf{L}_\ind}{\partial\zeta}
\mathbf{L}_\ind^{-1}
\end{equation}
are both entire functions of $\zeta$.  Using the above expansions of
$\mathbf{Z}_\ind(\zeta;y)$ and its derivatives, we easily obtain
the following expansions involving $\mathbf{L}_\ind$:
\begin{equation}
\begin{split}
\frac{\partial \mathbf{L}_\ind}{\partial y}&=
\left(-\frac{\partial\Theta}{\partial y}
\left[\mathbb{I}+\mathbf{A}_\ind(y)\zeta^{-1}+
\mathbf{B}_\ind(y)\zeta^{-2}+\bo(\zeta^{-3})\right]\sigma_3 +\bo(\zeta^{-1})\right)
e^{-\Theta\sigma_3}\\
&=\left(-\frac{1}{2}\sigma_3\zeta-\frac{1}{2}\mathbf{A}_\ind(y)\sigma_3 +
\bo(\zeta^{-1})\right)e^{-\Theta\sigma_3}\\
\frac{\partial\mathbf{L}_\ind}{\partial\zeta}&=
\left(-\frac{\partial\Theta}{\partial\zeta}
\left[\mathbb{I}+\mathbf{A}_\ind(y)\zeta^{-1}+\mathbf{B}_\ind(y)\zeta^{-2}
+\bo(\zeta^{-3})\right]\sigma_3 +
\bo(\zeta^{-2})\right)e^{-\Theta\sigma_3}\\
&=\left(-\frac{3}{2}\sigma_3\zeta^2-\frac{3}{2}\mathbf{A}_\ind(y)\sigma_3\zeta
-\frac{3}{2}\mathbf{B}_\ind(y)\sigma_3-\frac{1}{2}y\sigma_3 + \bo(\zeta^{-1})
\right)e^{-\Theta\sigma_3}\\
\mathbf{L}_\ind^{-1}&=e^{\Theta\sigma_3}\left(\mathbb{I}-\mathbf{A}_\ind(y)\zeta^{-1}
+\left[\mathbf{A}_\ind(y)^2-\mathbf{B}_\ind(y)\right]\zeta^{-2} +\bo(\zeta^{-3})\right),
\end{split}
\end{equation}
where 
\begin{equation}
\Theta:=\frac{1}{2}\left[\zeta^3+y\zeta+(1-2\ind)\log(-\zeta)\right].
\end{equation}
We can then easily see that both $\mathbf{U}_\ind$
and $\mathbf{V}_\ind$ grow algebraically in $\zeta$ as $\zeta\to\infty$
and hence are necessarily polynomials:
\begin{equation}
\begin{split}
\mathbf{U}_\ind(y,\zeta)&
=-\frac{1}{2}\sigma_3\zeta +\frac{1}{2}[\sigma_3,\mathbf{A}_\ind(y)]\\
\mathbf{V}_\ind(y,\zeta)&=-\frac{3}{2}\sigma_3\zeta^2+\frac{3}{2} 
[\sigma_3,\mathbf{A}_\ind(y)]\zeta
-\frac{3}{2}[\sigma_3,\mathbf{A}_\ind(y)]\mathbf{A}_\ind(y)+\frac{3}{2}
[\sigma_3,\mathbf{B}_\ind(y)]-\frac{1}{2}y\sigma_3.
\end{split}
\end{equation}
We may rewrite these in the form
\begin{equation}
\begin{split}
\mathbf{U}_\ind(y,\zeta)&=
-\frac{1}{2}\sigma_3\zeta +\begin{bmatrix}0 &\pu_\ind(y) \\ -\pv_\ind(y) & 0
\end{bmatrix}\\
\mathbf{V}_\ind(y,\zeta)&=-\frac{3}{2}\sigma_3\zeta^2 +3
\begin{bmatrix}0 & \pu_\ind(y)\\-\pv_\ind(y) & 0\end{bmatrix}\zeta +
\frac{1}{2}\begin{bmatrix}-6\pu_\ind(y)\pv_\ind(y)-y & 2\pw_\ind(y) \\-2\pz_\ind(y) & 
6\pu_\ind(y)\pv_\ind(y)+y\end{bmatrix},
\end{split}
\label{eq:DSUVdefine}
\end{equation}
where
\begin{equation}
\begin{split}
\pu_\ind(y)&:=A_{\ind,12}(y)\\ \pv_\ind(y)&:=A_{\ind,21}(y)\\ 
\pw_\ind(y)&:=3B_{\ind,12}(y)-3A_{\ind,12}(y)A_{\ind,22}(y)\\
\pz_\ind(y)&:=3B_{\ind,21}(y)-3A_{\ind,21}(y)A_{\ind,11}(y).
\end{split}
\label{eq:uvwzdefs}
\end{equation}
The matrix $\mathbf{L}_\ind$ is therefore a simultaneous 
(and fundamental, since $\det(\mathbf{L}_\ind)\equiv 1$) solution of the \emph{Lax
equations}
\begin{equation}
\frac{\partial \mathbf{L}_\ind}{\partial y}=\mathbf{U}_\ind(y,\zeta)\mathbf{L}_\ind
\quad\text{and}\quad
\frac{\partial \mathbf{L}_\ind}{\partial\zeta}=\mathbf{V}_\ind(y,\zeta)\mathbf{L}_\ind,
\label{eq:DSLaxEqns}
\end{equation}
and therefore the compatibility condition
\begin{equation}
\frac{\partial \mathbf{U}_\ind}{\partial\zeta}-
\frac{\partial\mathbf{V}_\ind}{\partial y} + [\mathbf{U}_\ind,\mathbf{V}_\ind]=0
\end{equation}
holds identically in $\zeta$.  Separating the powers of $\zeta$
leads to the following system of differential equations (the \emph{Painlev\'e-II system}) governing
the quantities defined by \eqref{eq:uvwzdefs}:
\begin{equation}
\begin{split}
\pu_\ind'(y)&=-\frac{1}{3}\pw_\ind(y)\\
\pv_\ind'(y)&=\frac{1}{3}\pz_\ind(y)\\
\pw_\ind'(y)&=6\pu_\ind(y)^2\pv_\ind(y)+y\pu_\ind(y)\\
\pz_\ind'(y)&=-6\pu_\ind(y)\pv_\ind(y)^2-y\pv_\ind(y).
\end{split}
\label{eq:PIIsystem}
\end{equation}
Although satisfied by quantities evidently depending on $\ind\in\mathbb{Z}$, the Painlev\'e-II system does not involve $\ind$ in any explicit way.
Eliminating $\pw_\ind$ and $\pz_\ind$ yields the coupled system of second-order 
Painlev\'e-II-type equations
\begin{equation}
\pu_\ind''(y)+2\pu_\ind(y)^2\pv_\ind(y)+\frac{1}{3}y\pu_\ind(y)=0\quad\text{and}\quad
\pv_\ind''(y)+2\pu_\ind(y)\pv_\ind(y)^2+\frac{1}{3}y\pv_\ind(y)=0.
\label{eq:uvPIIsystem}
\end{equation}

In fact, functions related in an elementary way to $\pu_\ind$ and $\pv_\ind$ turn out to satisfy uncoupled second-order equations of Painlev\'e-II type, but we must use the \emph{inhomogeneous} form of the Painlev\'e-II equation, and the coefficient of inhomogeneity
will depend on $\ind\in\mathbb{Z}$.  To see that this is so, we follow \cite[pages 154--155]{FokasIKN06}  by multiplying the first equation of \eqref{eq:uvPIIsystem} by $\pv_\ind$, the second by $\pu_\ind$, and subtracting: 
\begin{equation}
0=\pv_\ind(y)\pu_\ind''(y)-\pu_\ind(y)\pv_\ind''(y)=\frac{d}{dy}\left[\pv_\ind(y)\pu_\ind'(y)-\pu_\ind(y)
\pv_\ind'(y)\right].
\end{equation}
Hence, the quantity
\begin{equation}
\lambda_\ind:=\pv_\ind(y)\pu_\ind'(y)-\pu_\ind(y)\pv_\ind'(y)=-\frac{1}{3}\left[\pv_\ind(y)\pw_\ind(y)+\pu_\ind(y)\pz_\ind(y)
\right]
\label{eq:nuconst}
\end{equation}
is a constant, independent of $y$.  

To compute the value of $\lambda_\ind$ for the particular solution of 
\eqref{eq:PIIsystem} corresponding to the specific Riemann-Hilbert problem
for $\mathbf{Z}_\ind(\zeta;y)$, we first consider the 
direct problem for
the Lax equation $\mathbf{L}_{\ind,\zeta}=\mathbf{V}_\ind\mathbf{L}_\ind$, formally expanding
$\mathbf{L}_\ind(y,\zeta)$ for large $\zeta$ by assuming the form
\begin{equation}
\mathbf{L}_\ind(y,\zeta)\sim\left(\mathbb{I}+\sum_{k=1}^\infty 
\mathbf{Y}_{\ind,k}(y)\zeta^{-k}
\right)e^{\mathbf{D}_\ind(y,\zeta)},\quad\zeta\to\infty,
\label{eq:GexpansionYD}
\end{equation}
where each of the coefficients $\mathbf{Y}_{\ind,k}(y)$ is an 
off-diagonal matrix, and
where $\mathbf{D}_\ind(y,\zeta)$ is a diagonal matrix with the asymptotic expansion
\begin{equation}
\mathbf{D}_\ind(y,\zeta)\sim \mathbf{D}_{\ind,-3}(y)\zeta^3 +
\mathbf{D}_{\ind,-2}(y)\zeta^2
+\mathbf{D}_{\ind,-1}(y)\zeta +\mathbf{D}_{\ind,0}(y)\log(-\zeta)+\sum_{k=1}^\infty
\mathbf{D}_{\ind,k}(y)\zeta^{-k},\quad \zeta\to\infty,
\end{equation}
where all of the coefficients $\mathbf{D}_{\ind,k}(y)$ are diagonal matrices.  The value of $\lambda_\ind$ will be determined from the coefficient $\mathbf{D}_{\ind,0}$.  By
differentiation of the formal series, we obtain
\begin{equation}
\frac{\partial\mathbf{L}_\ind}{\partial\zeta}-\mathbf{V}_\ind\mathbf{L}_\ind\sim
\left[
\mathbf{M}_{\ind,-2}\zeta^2 +\mathbf{M}_{\ind,-1}\zeta+\mathbf{M}_{\ind,0} +
\mathbf{M}_{\ind,1}\zeta^{-1} + \mathbf{M}_{\ind,2}\zeta^{-2}+\bo(\zeta^{-3})\right]
e^{\mathbf{D}_\ind(y,\zeta)},
\quad\zeta\to\infty,
\end{equation}
where the matrix coefficients are systematically determined as follows:
\begin{equation}
\begin{split}
\mathbf{M}_{\ind,-2}&:=
3\mathbf{D}_{\ind,-3}(y)+\frac{3}{2}\sigma_3\\
\mathbf{M}_{\ind,-1}&:=
3\mathbf{Y}_{\ind,1}(y)\mathbf{D}_{\ind,-3}(y)+2\mathbf{D}_{\ind,-2}(y)
-3\begin{bmatrix}0 & \pu_\ind(y)\\-\pv_\ind(y) & 0
\end{bmatrix}+\frac{3}{2}\sigma_3\mathbf{Y}_{\ind,1}(y)\\
\mathbf{M}_{\ind,0}&:=
3\mathbf{Y}_{\ind,2}(y)\mathbf{D}_{\ind,-3}(y)+2\mathbf{Y}_{\ind,1}(y)
\mathbf{D}_{\ind,-2}(y)+\mathbf{D}_{\ind,-1}(y)\\
&\quad{}
-\frac{1}{2}\begin{bmatrix}-6\pu_\ind(y)\pv_\ind(y)-y & 2\pw_\ind(y)\\-2\pz_\ind(y) & 
6\pu_\ind(y)\pv_\ind(y)+y\end{bmatrix}
-3\begin{bmatrix}0 & \pu_\ind(y)\\-\pv_\ind(y) & 0\end{bmatrix}
\mathbf{Y}_{\ind,1}(y) +\frac{3}{2}\sigma_3
\mathbf{Y}_{\ind,2}(y)\\
\mathbf{M}_{\ind,1}&:=
3\mathbf{Y}_{\ind,3}(y)\mathbf{D}_{\ind,-3}(y)+2\mathbf{Y}_{\ind,2}(y)
\mathbf{D}_{\ind,-2}(y)+\mathbf{Y}_{\ind,1}(y)
\mathbf{D}_{\ind,-1}(y)+\mathbf{D}_{\ind,0}(y)\\
&\quad{}-\frac{1}{2}\begin{bmatrix}-6\pu_\ind(y)\pv_\ind(y)-y & 2\pw_\ind(y)\\-2\pz_\ind(y) 
& 6\pu_\ind(y)\pv_\ind(y)+y\end{bmatrix}
\mathbf{Y}_{\ind,1}(y) \\
&\quad{}
-3\begin{bmatrix}0 & \pu_\ind(y)\\-\pv_\ind(y) & 0\end{bmatrix}
\mathbf{Y}_{\ind,2}(y) +\frac{3}{2}\sigma_3
\mathbf{Y}_{\ind,3}(y)\\
\mathbf{M}_{\ind,2}&:=
-\mathbf{D}_{\ind,1}(y)
-\frac{1}{2}\begin{bmatrix}0 & 2\pw_\ind(y)\\-2\pz_\ind(y) & 0\end{bmatrix}
\mathbf{Y}_{\ind,2}(y)\\
&\quad{}
-3\begin{bmatrix}0 & \pu_\ind(y)\\-\pv_\ind(y) & 0\end{bmatrix}\mathbf{Y}_{\ind,3}(y) +
\text{off-diagonal terms},
\end{split}
\end{equation}
and so on.  By setting each $\mathbf{M}_{\ind,k}$ to zero in sequence,
and by separating the diagonal and off-diagonal parts of each of these
equations, we obtain the following:  the 
matrix equation $\mathbf{M}_{\ind,-2}=\mathbf{0}$
implies
\begin{equation}
\mathbf{D}_{\ind,-3}(y)=-\frac{1}{2}\sigma_3,
\end{equation}
the matrix equation $\mathbf{M}_{\ind,-1}=\mathbf{0}$ implies
\begin{equation}
\mathbf{Y}_{\ind,1}(y)=\begin{bmatrix}0 & \pu_\ind(y)\\\pv_\ind(y) & 0\end{bmatrix}
\quad\text{and}
\quad
\mathbf{D}_{\ind,-2}(y)=\mathbf{0},
\end{equation}
the matrix equation $\mathbf{M}_{\ind,0}=\mathbf{0}$ implies
\begin{equation}
\mathbf{Y}_{\ind,2}(y)=\frac{1}{3}\begin{bmatrix}0 & \pw_\ind(y)\\\pz_\ind(y) & 0\end{bmatrix}
\quad\text{and}\quad
\mathbf{D}_{\ind,-1}(y)=-\frac{1}{2}y\sigma_3,
\end{equation}
the matrix equation $\mathbf{M}_{\ind,1}=\mathbf{0}$ implies
\begin{equation}
\mathbf{Y}_{\ind,3}(y)=-\frac{1}{3}\begin{bmatrix}0 & 3\pu_\ind(y)^2\pv_\ind(y)+y\pu_\ind(y)\\
3\pu_\ind(y)\pv_\ind(y)^2+y\pv_\ind(y) & 0\end{bmatrix}
\quad\text{and}\quad
\mathbf{D}_{\ind,0}(y) = -3\lambda_\ind\sigma_3,
\end{equation}
where we have recalled the definition \eqref{eq:nuconst} of $\lambda_\ind$,
and finally, the diagonal terms of the matrix equation
$\mathbf{M}_{\ind,2}=\mathbf{0}$ imply that
\begin{equation}
\mathbf{D}_{\ind,1}(y)=-2H_\ind(y)\sigma_3
\end{equation}
where the Painlev\'e-II \emph{Hamiltonian} is
\begin{equation}
H_\ind(y):=\frac{1}{6}\pw_\ind(y)\pz_\ind(y)-\frac{3}{2}\pu_\ind(y)^2\pv_\ind(y)^2-\frac{1}{2}y\pu_\ind(y)\pv_\ind(y).
\label{eq:PIIHamiltonian}
\end{equation}
It follows that the formal asymptotic expansion of $\mathbf{L}_\ind(y,\zeta)$
for large $\zeta$ takes the form
\begin{equation}
\begin{split}
\mathbf{L}_\ind(y,\zeta)e^{(\zeta^3+y\zeta)\sigma_3/2}(-\zeta)^{3\lambda_\ind\sigma_3}\sim
\mathbb{I}&+\left[\mathbf{Y}_{\ind,1}(y)+\mathbf{D}_{\ind,1}(y)\right]\zeta^{-1} \\
&+ 
\left[\mathbf{Y}_{\ind,2}(y)+\mathbf{D}_{\ind,2}(y)+\frac{1}{2}\mathbf{D}_{\ind,1}(y)^2+
\mathbf{Y}_{\ind,1}(y)\mathbf{D}_{\ind,1}(y)\right]\zeta^{-2}\\
&+\bo(\zeta^{-3})
,\quad\zeta\to\infty.
\end{split}
\label{eq:Gexpansion}
\end{equation}
Comparing with \eqref{eq:Fexpansion} and \eqref{eq:GFrelation}, we see that
\begin{equation}
\lambda_\ind=\frac{1}{6}-\frac{\ind}{3},
\label{eq:nuvalue}
\end{equation}
and also that 
\begin{equation}
\begin{split}
\mathbf{A}_\ind(y)&=\mathbf{Y}_{\ind,1}(y)+\mathbf{D}_{\ind,1}(y)\\
\mathbf{B}_\ind(y)&= 
\mathbf{Y}_{\ind,2}(y)+\mathbf{D}_{\ind,2}(y)+\frac{1}{2}\mathbf{D}_{\ind,1}(y)^2+
\mathbf{Y}_{\ind,1}(y)\mathbf{D}_{\ind,1}(y),
\end{split}
\end{equation}
so that in addition to the relations \eqref{eq:uvwzdefs} we have
\begin{equation}
A_{\ind,11}(y)=-2H_\ind(y)\quad\text{and}\quad A_{\ind,22}(y)=2H_\ind(y)
\label{eq:Adiag}
\end{equation}
and
\begin{equation}
B_{\ind,12}(y)=\frac{1}{3}\pw_\ind(y)+2H_\ind(y)\pu_\ind(y)\quad\text{and}\quad
B_{\ind,21}(y)=\frac{1}{3}\pz_\ind(y)-2H_\ind(y)\pv_\ind(y).
\label{eq:Boffdiag}
\end{equation}

Now that the value of $\lambda_\ind$ has been determined according to \eqref{eq:nuvalue}, a direct calculation
shows that the logarithmic derivatives
\begin{equation}
\mathcal{P}_\ind(y):=\frac{\pu_\ind'(y)}{\pu_\ind(y)}\quad\text{and}\quad 
\mathcal{Q}_\ind(y):=\frac{\pv_\ind'(y)}{\pv_\ind(y)}
\label{eq:logderivs}
\end{equation}
satisfy the uncoupled Painlev\'e-II equations
\begin{equation}
\begin{alignedat}{2}
\mathcal{P}_\ind''(y)&=2\mathcal{P}_\ind(y)^3+\frac{2}{3}y\mathcal{P}_\ind(y) -\left[\frac{1}{3}-2\lambda_\ind\right]
&\qquad\qquad
\mathcal{Q}_\ind''(y)&=2\mathcal{Q}_\ind(y)^3+\frac{2}{3}y\mathcal{Q}_\ind(y)-\left[\frac{1}{3}+2\lambda_\ind\right]
\\
&=2\mathcal{P}_\ind(y)^3+\frac{2}{3}y\mathcal{P}_\ind(y)-\frac{2}{3}\ind
&\qquad\qquad &=2\mathcal{Q}_\ind(y)^3+\frac{2}{3}y\mathcal{Q}_\ind(y)+\frac{2}{3}(\ind-1).
\end{alignedat}
\end{equation}
Note that for general $\ind\in\mathbb{Z}$, both equations are of 
inhomogeneous type.

\subsection{Solution of Riemann-Hilbert Problem~\ref{rhp:DSlocalII} for $\ind=0$.}
Recalling the relation \eqref{eq:A12squaredeq1} and the definitions
\eqref{eq:uvwzdefs} of the potentials $\pu_\ind$, $\pv_\ind$, $\pw_\ind$, and $\pz_\ind$,
we learn that when $\ind=0$, $\pu_0(y)\equiv\sigma$, where (assuming for
the moment continuity of $\pu_0(y)$ with respect to $y$) $\sigma=\pm 1$
is a fixed sign to be determined.  According to the system of differential
equations \eqref{eq:PIIsystem} necessarily satisfied by the
potentials, we can determine the values of the remaining potentials in
the special case of $\ind=0$:
\begin{equation}
\pu_0(y)= \sigma,\quad \pv_0(y)=-\frac{\sigma}{6}y,\quad \pw_0(y)=0,\quad 
\pz_0(y)=-\frac{\sigma}{2}.
\label{eq:potentialszero}
\end{equation}
Therefore, from \eqref{eq:DSUVdefine} and \eqref{eq:DSLaxEqns} we see that
in the special case of $\ind=0$, the linear differential equations
simultaneously satisfied by the matrix $\mathbf{L}_0(a,\zeta)=
\mathbf{Z}_0(\zeta;y)e^{-(\zeta^3+y\zeta)\sigma_3/2}$ take the form
\begin{equation}
\frac{\partial\mathbf{L}_0}{\partial y} = 
\begin{bmatrix}
-\tfrac{1}{2}\zeta & \sigma\\\tfrac{1}{6}\sigma y & \tfrac{1}{2}\zeta
\end{bmatrix}\mathbf{L}_0\quad\text{and}\quad
\frac{\partial\mathbf{L}_0}{\partial\zeta}=
\begin{bmatrix}
-\tfrac{3}{2}\zeta^2 & 3\sigma\zeta\\
\tfrac{1}{2}\sigma(y\zeta+1) & \tfrac{3}{2}\zeta^2\end{bmatrix}
\mathbf{L}_0.
\end{equation}
Let $\mathbf{l}=\trans{(l_{1},l_{2})}$ denote either of the columns of 
$\mathbf{L}_0$.
From the first of the above two Lax equations, we  observe that 
\begin{equation}
l_{2}=\sigma\left(\frac{\partial l_{1}}{\partial y} +
\frac{1}{2}\zeta l_{1}\right)
\end{equation}
and then that 
\begin{equation}
l_{1}=F(\xi),\quad \xi:=\frac{1}{6^{1/3}}\left(y+\frac{3}{2}\zeta^2\right),
\end{equation}
where $F$ is any solution of Airy's equation $F''(\xi)=\xi F(\xi)$.  Therefore 
$\mathbf{l}$ necessarily has the form
\begin{equation}
\mathbf{l}=\begin{bmatrix}F(\xi)\\\sigma (6^{-1/3}F'(\xi)+\tfrac{1}{2}\zeta F(\xi))
\end{bmatrix}
\end{equation}
and by substitution into the second equation of the Lax pair we learn
that the Airy function $F$ may only depend on $\zeta$ through the variable $\xi$,
in other words, we must select for $F$ a linear combination of $\mathrm{Ai}(\xi)$
and $\mathrm{Bi}(\xi)$ with coefficients independent of $\zeta$.

Now we try to select the correct solutions of Airy's equation to
construct $\mathbf{L}_0(y,\zeta)$ and hence $\mathbf{Z}_0(\zeta;y)$
in each of the six sectors of analyticity
in such a way that we obtain the desired asymptotics for large $\zeta$
and satisfy the required jump conditions.  This procedure will also determine
the correct value of $\sigma=\pm 1$.

Let us consider three specific choices of solution to Airy's equation leading
to the three solution vectors
\begin{equation}
\mathbf{l}^0:=\begin{bmatrix}\mathrm{Ai}(\xi)\\
\sigma(6^{-1/3}\mathrm{Ai}'(\xi)+\tfrac{1}{2}\zeta\mathrm{Ai}(\xi))
\end{bmatrix},\quad
\mathbf{l}^\pm:=\begin{bmatrix}\mathrm{Ai}(e^{\pm 2\pi i/3}\xi)\\
\sigma(6^{-1/3}e^{\pm 2\pi i/3}\mathrm{Ai}'(e^{\pm 2\pi i/3}\xi) +
\tfrac{1}{2}\zeta
\mathrm{Ai}(e^{\pm 2\pi i/3}\xi))\end{bmatrix}.
\label{eq:gvectors}
\end{equation}
Each of these is an entire function of $\zeta$.
From the standard asymptotic formulae
\begin{equation}
\mathrm{Ai}(z)=\frac{1}{2\sqrt{\pi}}z^{-1/4}e^{-2z^{3/2}/3}(1+\bo(|z|^{-3/2}))\quad
\text{and}\quad
\mathrm{Ai}'(z)=-\frac{1}{2\sqrt{\pi}}z^{1/4}e^{-2z^{3/2}/3}(1+\bo(|z|^{-3/2}))
\end{equation}
valid as $z\to\infty$ with $|\arg(z)|\le \pi-\delta$
for any $\delta>0$, we see that (all of the following asymptotic statements 
assume $y$ is held fixed)
\begin{equation}
\mathbf{l}^0 = 
\frac{\zeta^{-1/2}e^{-(\zeta^3+y\zeta)/2}}{(48)^{1/6}\sqrt{\pi}}
\begin{bmatrix}1+\bo(\zeta^{-1})\\
\bo(\zeta^{-1})\end{bmatrix},\quad
\zeta\to\infty,\quad|\arg(\zeta)|\le\frac{\pi}{2}-\delta,
\end{equation}
\begin{equation}
\mathbf{l}^0 = 
-\sigma\frac{(-\zeta)^{1/2}e^{(\zeta^3+y\zeta)/2}}{(48)^{1/6}\sqrt{\pi}}
\begin{bmatrix}\bo(\zeta^{-1})\\
1+\bo(\zeta^{-1})
\end{bmatrix},\quad \zeta\to\infty,\quad
|\arg(-\zeta)|\le\frac{\pi}{2}-\delta.
\end{equation}
Similarly,
\begin{equation}
\mathbf{l}^\pm = -\sigma e^{\pm 2\pi i/3}
\frac{(e^{\pm i\pi/3}\zeta)^{1/2}e^{(\zeta^3+y\zeta)/2}}{(48)^{1/6}\sqrt{\pi}}
\begin{bmatrix}\bo(\zeta^{-1})\\1+\bo(\zeta^{-1})\end{bmatrix},\quad
\zeta\to\infty,\quad
|\arg(e^{\pm i\pi/3}\zeta)|\le\frac{\pi}{2}-\delta,
\end{equation}
\begin{equation}
\mathbf{l}^\pm = \frac{(e^{\mp 2\pi i/3}\zeta)^{-1/2}e^{-(\zeta^3+y\zeta)/2}}
{(48)^{1/6}\sqrt{\pi}}\begin{bmatrix}1+\bo(\zeta^{-1})\\
\bo(\zeta^{-1})\end{bmatrix},\quad\zeta\to\infty,\quad
|\arg(e^{\mp 2\pi i/3}\zeta)|\le\frac{\pi}{2}-\delta.
\end{equation}
We may now define a candidate for $\mathbf{Z}_0(\zeta;y)$
that satisfies the required analyticity and normalization properties among
the conditions of Riemann-Hilbert Problem~\ref{rhp:DSlocalII}:
\begin{equation}
\mathbf{Z}_0(\zeta;y):=
\begin{cases}
(48)^{1/6}\sqrt{\pi}(i\mathbf{l}^0,\sigma e^{-2\pi i/3}\mathbf{l}^-)
e^{(\zeta^3+y\zeta)\sigma_3/2},\quad &
\arg(\zeta)\in (0,\frac{\pi}{3})\\
(48)^{1/6}\sqrt{\pi}(e^{i\pi/6}\mathbf{l}^+,\sigma e^{-2\pi i/3}\mathbf{l}^-)
e^{(\zeta^3+y\zeta)\sigma_3/2},
\quad & \arg(\zeta)\in (\frac{\pi}{3},\frac{2\pi}{3})\\
(48)^{1/6}\sqrt{\pi}(e^{i\pi/6}\mathbf{l}^+,-\sigma\mathbf{l}^0)
e^{(\zeta^3+y\zeta)\sigma_3/2},\quad & \arg(\zeta)\in (\frac{2\pi}{3},\pi)\\
(48)^{1/6}\sqrt{\pi}(e^{-i\pi/6}\mathbf{l}^-,-\sigma\mathbf{l}^0)
e^{(\zeta^3+y\zeta)\sigma_3/2},\quad & \arg(\zeta)\in (-\pi,-\frac{2\pi}{3})\\
(48)^{1/6}\sqrt{\pi}(e^{-i\pi/6}\mathbf{l}^-,\sigma e^{2\pi i/3}\mathbf{l}^+)
e^{(\zeta^3+y\zeta)\sigma_3/2},\quad & \arg(\zeta)\in (-\frac{2\pi}{3},-\frac{\pi}{3})\\
(48)^{1/6}\sqrt{\pi}(-i\mathbf{l}^0,\sigma e^{2\pi i/3}\mathbf{l}^+)
e^{(\zeta^3+y\zeta)\sigma_3/2},\quad &
\arg(\zeta)\in (-\frac{\pi}{3},0).
\end{cases}
\label{eq:Rnloczero}
\end{equation}
It only remains to verify the jump conditions of Riemann-Hilbert 
Problem~\ref{rhp:DSlocalII}, 
with the help of the identity
\begin{equation}
\mathrm{Ai}(z)+e^{-2\pi i/3}\mathrm{Ai}(ze^{-2\pi i/3}) + e^{2\pi i/3}
\mathrm{Ai}(ze^{2\pi i/3})=0,\quad z\in\mathbb{C},
\end{equation}
which implies the vector identity $\mathbf{l}^0 + e^{-2\pi i/3}\mathbf{l}^- +
e^{2\pi i/3}\mathbf{l}^+=\mathbf{0}$.  
Using this fact, it is a direct matter to confirm that the required 
jump conditions are indeed satisfied, \emph{provided we make the choice of sign $\sigma=+1$}.
We have therefore proved the following.
\begin{proposition}
Let $\ind=0$.  For every $y\in\mathbb{R}$, Riemann-Hilbert 
Problem~\ref{rhp:DSlocalII} has a unique solution given explicitly 
by the formulae \eqref{eq:gvectors} and \eqref{eq:Rnloczero} with $\sigma=+1$.
Moreover, $\mathbf{Z}_0(\zeta;y)$ is an entire function of $y$,
$\mathbf{Z}_0(\zeta;y)(-\zeta)^{\sigma_3/2}$ has a complete 
asymptotic expansion as $\zeta\to\infty$ in descending integer powers of
$\zeta$ (as in \eqref{eq:Fexpansion}) differentiable term-by-term
with respect to both $y$ and $\zeta$,
and the corresponding potentials $\pu_0(y)$, $\pv_0(y)$, $\pw_0(y)$ and $\pz_0(y)$
are given by \eqref{eq:potentialszero} with $\sigma=+1$.
\label{prop:R0locsolution}
\end{proposition}
\subsection{Schlesinger-B\"acklund transformations.  Solution of Riemann-Hilbert Problem~\ref{rhp:DSlocalII} for general $\ind\in\mathbb{Z}$.}
Now we describe an inductive procedure for obtaining
$\mathbf{Z}_\ind(\zeta;y)$ uniquely from the conditions of
Riemann-Hilbert Problem~\ref{rhp:DSlocalII} for $\ind\neq 0$.  Suppose
$y\in\mathbb{R}$ is a value for which 
Riemann-Hilbert
Problem~\ref{rhp:DSlocalII}
has a (unique) solution $\mathbf{Z}_\ind(\zeta;y)$ for some 
$\ind\in\mathbb{Z}$.  
Following  \cite[Chapter~6]{FokasIKN06}, we claim that 
\begin{equation}
\mathbf{Z}_{\ind\pm 1}(\zeta;y)=
(\mathbf{S}_{\ind,1}^\pm(y)\zeta +\mathbf{S}_{\ind,0}^\pm(y))
\mathbf{Z}_\ind(\zeta;y)
\label{eq:Rnpm1}
\end{equation}
if the matrices $\mathbf{S}_{\ind,0}^\pm(y)$ and $\mathbf{S}_{\ind,1}^\pm(y)$ 
can be properly chosen. 

It is immediately clear that since the the prefactor is entire in $\zeta$, the right-hand side of \eqref{eq:Rnpm1} has the
same domain of analyticity, achieves its boundary values in the same classical
sense, and satisfies exactly the same jump conditions, as does 
$\mathbf{Z}_\ind(\zeta;y)$.  Of course the analyticity and jump
conditions in Riemann-Hilbert Problem~\ref{rhp:DSlocalII} are independent 
of $\ind$, so it only remains to impose the normalization condition \eqref{eq:Zindnorm} on $\mathbf{Z}_{\ind\pm 1}(\zeta;y)$ given by \eqref{eq:Rnpm1}, and this condition 
takes the form
\begin{equation}
\lim_{\zeta\to\infty}
(\mathbf{S}_{\ind,1}^\pm(y)\zeta+\mathbf{S}_{\ind,0}^\pm(y))
\mathbf{Z}_\ind(\zeta;y)
(-\zeta)^{(1-2\ind\mp 2)\sigma_3/2}  =\mathbb{I}.
\end{equation}
Assuming the expansion \eqref{eq:Fexpansion} (which at the moment is justified
only for $\ind=0$), we are requiring that
\begin{equation}
(\mathbf{S}_{\ind,1}^\pm(y)\zeta+\mathbf{S}_{\ind,0}^\pm(y))
(\mathbb{I}+\mathbf{A}_\ind(y)\zeta^{-1} +\mathbf{B}_\ind(y)\zeta^{-2} +
\mathbf{C}_\ind(y)\zeta^{-3}+
 \bo(\zeta^{-4}))(-\zeta)^{\mp \sigma_3} =
\mathbb{I} + \bo(\zeta^{-1}),\quad\zeta\to\infty.
\label{eq:Wpmcondgen}
\end{equation}
Equating the polynomial part of the left-hand side to the identity yields
the conditions
\begin{equation}
\begin{split}
\mathbf{S}_{\ind,1}^+(y)\begin{bmatrix}0 & 0\\0 & 1\end{bmatrix}&=\mathbf{0}\\
(\mathbf{S}_{\ind,1}^+(y)\mathbf{A}_\ind(y)+\mathbf{S}_{\ind,0}^+(y))
\begin{bmatrix}0 & 0\\0 & 1\end{bmatrix} &=\mathbf{0}\\
(\mathbf{S}_{\ind,1}^+(y)\mathbf{B}_\ind(y)+\mathbf{S}_{\ind,0}^+(y)\mathbf{A}_\ind(y))
\begin{bmatrix}0 & 0\\0 & 1\end{bmatrix}+\mathbf{S}_{\ind,1}^+(y)
\begin{bmatrix}1 & 0\\0 & 0\end{bmatrix}&=-\mathbb{I}
\end{split}
\end{equation}
which can be solved uniquely for the matrix elements of $\mathbf{S}_{\ind,1}^+(y)$
and $\mathbf{S}_{\ind,0}^+(y)$:
\begin{equation}
\begin{split}
\mathbf{S}_{\ind,1}^+(y)&=\begin{bmatrix}-1 & 0\\0 & 0\end{bmatrix},\\
\mathbf{S}_{\ind,0}^+(y)&=\begin{bmatrix}B_{\ind,12}(y)/A_{\ind,12}(y)-A_{\ind,22}(y) & A_{\ind,12}(y)\\
-1/A_{\ind,12}(y) & 0\end{bmatrix} = 
\begin{bmatrix}\pw_\ind(y)/(3\pu_\ind(y))& \pu_\ind(y)\\
-1/\pu_\ind(y) & 0\end{bmatrix},
\end{split}
\label{eq:SchlesingerPlus}
\end{equation}
\emph{provided that $\pu_\ind(y)\neq 0$}, where we have used 
\eqref{eq:uvwzdefs}, \eqref{eq:Adiag}, and
\eqref{eq:Boffdiag}.  Similarly, from \eqref{eq:Wpmcondgen} we obtain the 
conditions
\begin{equation}
\begin{split}
\mathbf{S}_{\ind,1}^-(y)\begin{bmatrix}1 & 0\\0 & 0\end{bmatrix}&=\mathbf{0}\\
(\mathbf{S}_{\ind,1}^-(y)\mathbf{A}_\ind(y)+\mathbf{S}_{\ind,0}^-(y))
\begin{bmatrix}1 & 0\\0 & 0\end{bmatrix} &=\mathbf{0}\\
(\mathbf{S}_{\ind,1}^-(y)\mathbf{B}_\ind(y)+\mathbf{S}_{\ind,0}^-(y)\mathbf{A}_\ind(y))
\begin{bmatrix}1 & 0\\0 & 0\end{bmatrix}+\mathbf{S}_{\ind,1}^-(y)
\begin{bmatrix}0 & 0\\0 & 1\end{bmatrix}&=-\mathbb{I}
\end{split}
\end{equation}
which can be solved uniquely for the matrix elements of $\mathbf{S}_{\ind,1}^-(y)$
and $\mathbf{S}_{\ind,0}^-(y)$:
\begin{equation}
\begin{split}
\mathbf{S}_{\ind,1}^-(y)&=\begin{bmatrix}0 & 0\\0 & -1\end{bmatrix},\\
\mathbf{S}_{\ind,0}^-(y)&=\begin{bmatrix}0 & -1/A_{\ind,21}(y)\\A_{\ind,21}(y) & B_{\ind,21}(y)/A_{\ind,21}(y)-A_{\ind,11}(y)
\end{bmatrix} = 
\begin{bmatrix}0 & -1/\pv_\ind(y)\\
\pv_\ind(y) & \pz_\ind(y)/(3\pv_\ind(y))\end{bmatrix},
\end{split}
\label{eq:SchlesingerMinus}
\end{equation}
\emph{provided that $\pv_\ind(y)\neq 0$}.
When it makes sense, 
the formula \eqref{eq:Rnpm1} subject to \eqref{eq:SchlesingerPlus} or
\eqref{eq:SchlesingerMinus} constitutes a so-called \emph{discrete 
isomonodromic Schlesinger transformation}.

The recurrence relation \eqref{eq:Rnpm1} also implies a corresponding
recurrence for the potentials $\pu_\ind$, $\pv_\ind$, $\pw_\ind$, and $\pz_\ind$.  Indeed,
expanding the right-hand side of \eqref{eq:Wpmcondgen} as $\mathbb{I}
+\mathbf{A}_{\ind\pm 1}(y)\zeta^{-1}+\bo(\zeta^{-2})$ the terms proportional
to $\zeta^{-1}$ yield the conditions
\begin{equation}
(\mathbf{S}_{\ind,1}^+(y)\mathbf{C}_\ind(y)+\mathbf{S}_{\ind,0}^+(y)\mathbf{B}_\ind(y))
\begin{bmatrix}0 & 0\\0 & 1\end{bmatrix} + (\mathbf{S}_{\ind,1}^+(y)\mathbf{A}_\ind(y)
+\mathbf{S}_{\ind,0}^+(y))
\begin{bmatrix}1 & 0 \\ 0 & 0\end{bmatrix}=-\mathbf{A}_{\ind+1}(y)
\label{eq:Baecklundplussystem}
\end{equation}
and
\begin{equation}
(\mathbf{S}_{\ind,1}^-(y)\mathbf{C}_\ind(y)+\mathbf{S}_{\ind,0}^-(y)\mathbf{B}_\ind(y))
\begin{bmatrix}1 & 0\\0 & 0\end{bmatrix} + (\mathbf{S}_{\ind,1}^-(y)\mathbf{A}_\ind(y)
+\mathbf{S}_{\ind,0}^-(y))
\begin{bmatrix}0 & 0 \\ 0 & 1\end{bmatrix}=-\mathbf{A}_{\ind-1}(y).
\label{eq:Baecklundminussystem}
\end{equation}
Recalling the definition \eqref{eq:uvwzdefs} of the potentials $(\pu_\ind,\pv_\ind)$
in terms of the matrix entries of $\mathbf{A}_\ind(y)$, we examine the 
$(2,1)$-entry of the matrix equation 
\eqref{eq:Baecklundplussystem} and easily obtain the relation
$\pv_{\ind+1}(y)=1/\pu_\ind(y)$.
Then, from the differential equations \eqref{eq:PIIsystem} satisfied by
the potentials for all $n$, we obtain
\begin{equation}
\pu_{\ind+1}(y)=-\frac{1}{6}y\pu_\ind(y)-\frac{\pu_\ind'(y)^2}{\pu_\ind(y)}
+\frac{1}{2}\pu_\ind''(y)\quad\text{and}\quad
\pv_{\ind+1}(y)=\frac{1}{\pu_\ind(y)}
\label{eq:Baecklundplus}
\end{equation}
which together with the implied relations for $\pw_{\ind+1}$ and $\pz_{\ind+1}$
constitutes a \emph{B\"acklund transformation} for the Painlev\'e-II system
\eqref{eq:PIIsystem}.  Similarly, the $(1,2)$-element of
\eqref{eq:Baecklundminussystem} together with \eqref{eq:PIIsystem}
yields the B\"acklund transformation
\begin{equation}
\pu_{\ind-1}(y)=\frac{1}{\pv_\ind(y)}\quad\text{and}\quad
\pv_{\ind-1}(y)=\frac{1}{2}\pv_\ind''(y)-\frac{\pv_\ind'(y)^2}{\pv_\ind(y)}
-\frac{1}{6}y\pv_\ind(y),
\label{eq:Baecklundminus}
\end{equation}
which is easily seen to be inverse to \eqref{eq:Baecklundplus}.  
The iterative scheme for generating the solution to Riemann-Hilbert
Problem~\ref{rhp:DSlocalII} for general $\ind\in\mathbb{Z}$ is then
simply to start with the solution valid for all $y\in\mathbb{R}$ for
$\ind=0$ as described by Proposition~\ref{prop:R0locsolution} and attempt
to apply the discrete isomonodromic Schlesinger transformation
\eqref{eq:Rnpm1} in order to repeatedly increase or decrease the value
of $\ind$ in integer steps.  At the level of the implied B\"acklund
transformations \eqref{eq:Baecklundplus} and
\eqref{eq:Baecklundminus}, this procedure generates a family
$\{(\pu_\ind(y),\pv_\ind(y),\pw_\ind(y),\pz_\ind(y))\}_{\ind\in\mathbb{Z}}$ of \emph{rational
solutions} of the Painlev\'e-II system \eqref{eq:PIIsystem}.  
At the level of the logarithmic derivatives given by \eqref{eq:logderivs} we have a
family of rational solutions of inhomogeneous Painlev\'e-II equations with certain quantized
values for the inhomogeneity parameter $\alpha$.  In fact, it is known \cite{Murata85} that the equation
\begin{equation}
\mathcal{P}''(y)=2\mathcal{P}(y)^3 + \frac{2}{3}y\mathcal{P}(y)-\frac{2}{3}\alpha
\end{equation}
has a unique rational solution $\mathcal{P}=\mathcal{P}_\alpha(y)$ if and only if $\alpha\in\mathbb{Z}$.  These may be represented as logarithmic derivatives of other rational functions $\pu_\alpha$ that are ratios of consecutive
\emph{Yablonskii-Vorob'ev polynomials} \cite{Clarkson03,ClarksonM03}, which are known to have a number of remarkable properties.

It is
important for us to recognize that the poles of these solutions can
occur for finite real  values of $y$.  For example, applying \eqref{eq:Baecklundplus} twice with $\pu_0(y)=1$ we obtain $\pu_1(y)=-y/6$ and $\pu_2(y)=(y^3+6)/(36y)$,
and the latter has a pole at $y=0$.  These poles correspond to values of $y$
for which the discrete Schlesinger maps from the neighboring values of
$\ind$ fail to exist, and consequently for these values of $y$ and $\ind$ 
\emph{there exists no solution to Riemann-Hilbert Problem~\ref{rhp:DSlocalII}}.

While there exist values of $y$ for which the B\"acklund
transformations \eqref{eq:Baecklundplus} or \eqref{eq:Baecklundminus}
may not make sense, we may always interpret these transformations as
maps on the ring of rational functions, and as such they can be shown
to have a kind of \emph{singularity confinement property}: if
$y_0\in\mathbb{C}$ is a pole of the solution for some $\ind$, then $y_0$
is a regular point of the solution for $\ind\pm 1$.  Indeed, a local
analysis of the system \eqref{eq:uvPIIsystem} in the spirit of
Painlev\'e's original method \cite{Ince} shows that all poles of $(\pu_\ind(y),\pv_\ind(y))$
are necessarily simple and simultaneously occur in both functions; if
$y_0$ is a pole, then for some  constants
$k\in\mathbb{C}\setminus\{0\}$ and $\omega\in\mathbb{C}$ we necessarily have
\begin{equation}
\begin{split}
\pu_\ind(y)&=k\left[\frac{1}{y-y_0} +\frac{y_0}{18}(y-y_0) +\left(\frac{1}{12}+\omega\right)(y-y_0)^2 +U_3(y)\right]\\
\pv_\ind(y)&=-k^{-1}\left[\frac{1}{y-y_0} +\frac{y_0}{18}(y-y_0) +\left(\frac{1}{12}-\omega\right)(y-y_0)^2+V_3(y)\right],
\end{split}
\label{eq:uvPainleveexpansions}
\end{equation}
where $U_3(y)$ and $V_3(y)$ are analytic functions each vanishing to third order
at $y=y_0$.  
Then, using these formulae in the B\"acklund
transformations \eqref{eq:Baecklundplus}--\eqref{eq:Baecklundminus}
shows that the singularities at $y=y_0$ in the pairs
$(\pu_{\ind+1}(y),\pv_{\ind+1}(y))$ and $(\pu_{\ind-1}(y),\pv_{\ind-1}(y))$ are all
removable.  Moreover, both $\pv_{\ind+1}(y)$ and $\pu_{\ind-1}(y)$ have simple
zeros at $y=y_0$.  It is easy to see that the singularity confinenment
property extends to the Schlesinger transformations that generated
\eqref{eq:Baecklundplus}--\eqref{eq:Baecklundminus}.  That is, if
$\mathbf{Z}_\ind(\zeta;y_0)$ does not exist (because $y_0$ is
a pole of $\pu_\ind(y)$ and $\pv_\ind(y)$), then
both $\mathbf{Z}_{\ind+1}(\zeta;y_0)$ and 
$\mathbf{Z}_{\ind-1}(\zeta;y_0)$ do exist and they may be calculated
from $\mathbf{Z}_\ind(\zeta;y)$ by applying Schlesinger transformations
and then taking the limit $y\to y_0$.  We summarize these results in the
following proposition.
\begin{proposition}
  Let $\ind\in\mathbb{Z}$.  Then by iterated Schlesinger transformations
  given by \eqref{eq:Rnpm1} subject to \eqref{eq:SchlesingerPlus} or
  \eqref{eq:SchlesingerMinus} applied to the base case of $\ind=0$ characterized
  by Proposition~\ref{prop:R0locsolution}, 
  a matrix function
  $\mathbf{Z}_\ind(\zeta;y)$ is well-defined as a rational
  function of $y$ with simple poles.  The poles of $\mathbf{Z}_\ind(\zeta;y)$
  are exactly the poles of the corresponding potentials $(\pu_\ind(y),\pv_\ind(y))$ 
  obtained from the base
  case of $(\pu_0(y),\pv_0(y))=(1,-y/6)$ by iterated B\"acklund
  transformations \eqref{eq:Baecklundplus}--\eqref{eq:Baecklundminus}.
\begin{itemize}
\item
  If $y_0$ is not a pole, then there is a neighborhood $A\subset\mathbb{C}$ 
  with $y_0\in A$ such that for $y\in A$, the matrix 
  $\mathbf{Z}_\ind(\zeta;y)$ 
  \begin{itemize}
  \item is analytic in $y$, 
  \item is the
  solution of Riemann-Hilbert Problem~\ref{rhp:DSlocalII}, 
  \item has an expansion
  for large $\zeta$ of the form \eqref{eq:Fexpansion} differentiable with
  respect to both $y$ and $\zeta$ and uniform for $y\in A$, and 
  \item is
  uniformly bounded for $y\in A$ and $\zeta$ in compact subsets of 
  $\mathbb{C}$.  
  \end{itemize}
\item 
  If $y_0$ is a pole, then Riemann-Hilbert 
  Problem~\ref{rhp:DSlocalII} has no solution for $y=y_0$; however both
  $\mathbf{Z}_{\ind+1}(\zeta;y)$ and 
  $\mathbf{Z}_{\ind-1}(\zeta;y)$ are regular at $y_0$ (in the
  sense of taking limits from $y\neq y_0$), and both $\pv_{\ind+1}(y)$ and
  $\pu_{\ind-1}(y)$ have simple zeros at $y=y_0$.
\end{itemize}
Since the poles are simple, the bound on the elements of 
$\mathbf{Z}_\ind(\zeta;y)$ (and its inverse, since 
$\det(\mathbf{Z}_\ind(\zeta;y))=1$ where defined) can be strengthened
to $\bo(|y-\mathscr{P}(\pu_\ind)|^{-1})$ (or equivalently $\bo(|y-\mathscr{P}(\pv_\ind)|^{-1})$) uniformly on compact
sets in the $\zeta$-plane.  This improved 
bound also applies to all of the coefficients
in the expansion \eqref{eq:Fexpansion} and to the error term as well.  In particular, this implies
that at $y=y_0\in\mathscr{P}(\pu_\ind)=\mathscr{P}(\pv_\ind)$, the functions $H_\ind(y)=\tfrac{1}{2}A_{\ind,22}(y)=-\tfrac{1}{2}A_{\ind,11}(y)$, $B_{\ind,12}(y)=2H_\ind(y)\pu_\ind(y)-\pu_\ind'(y)$,
and $B_{\ind,21}(y)=\pv_\ind'(y)-2H_\ind(y)\pv_\ind(y)$ all have simple poles.
\label{prop:Rnlocgeneral}
\end{proposition}

\section{Inadequacy of the Global Model}
On one hand, it seems quite reasonable to expect that the global model 
$\dot{\mathbf{P}}_\ind(w)$ defined by 
\eqref{eq:dotPdefineDS} should provide a good approximation to
$\mathbf{P}(w)$, at least if the integer $\ind$ is properly chosen.
Indeed, we have shown in an earlier paper \cite{BuckinghamMelliptic} that far from criticality a global model obtained
by steps completely analogous to those we have followed here (choice
of an appropriate $g$-function, opening of lenses, pointwise approximation
of jump matrices away from points of nonuniformity and the solution of a
corresponding outer model problem, and exact solution of the jump conditions
near points of nonuniformity yielding inner models that match well onto the
outer model at disk boundaries) leads directly to a Riemann-Hilbert problem
for the error of small-norm type.  On the other hand, in the present case
the formally large exponents $\pm s/\epsilon_N$ were removed from the jump
matrices in our construction of the inner model, and we may expect them to 
reappear when the mismatch between the inner and outer models is calculated
on $\partial U$, contaminating the error estimates when $s$ is small but large
compared with $\epsilon_N$.

In order to properly 
gauge our prospects for success, we introduce a new matrix unknown
$\mathbf{Q}_\ind(w)$ constructed from the old unknown $\mathbf{P}(w)$ and the 
explicit global model $\dot{\mathbf{P}}_\ind(w)$ defined by \eqref{eq:dotPdefineDS} as follows:
\begin{equation}
\mathbf{Q}_\ind(w):=\mathbf{P}(w)\dot{\mathbf{P}}_\ind(w)^{-1}.
\end{equation}
We might expect to be able to prove that $\mathbf{Q}_\ind(w)$ is a small perturbation of the identity matrix with the use of small-norm theory applied to the Riemann-Hilbert problem satisfied
by this matrix function.

To begin to determine the nature of the Riemann-Hilbert problem satisfied by $\mathbf{Q}_\ind(w)$, firstly note that $\mathbf{Q}_\ind(w)$ is analytic for $w\in U$,  since for such $w$,
$\dot{\mathbf{P}}_\ind(w)=\dot{\mathbf{P}}_\ind^\mathrm{in}(w)$ has determinant
one and satisfies exactly the same jump conditions as does
$\mathbf{P}(w)$.  For $w\in\mathbb{C}\setminus \overline{U}$,
$\mathbf{Q}_\ind(w)$ is analytic except on the arcs of the jump
contour for $\mathbf{P}(w)$ illustrated in Figure~\ref{fig:critPcontour}; it can further be checked that
since both $\mathbf{P}(w)$ and
$\dot{\mathbf{P}}_\ind(w)=\dot{\mathbf{P}}_\ind^\mathrm{out}(w)$ change sign
across the contour segment $(\mathfrak{b},0)$,
$\mathbf{Q}_\ind(w)$  also extends analytically to this segment.  The jump
contour for $\mathbf{Q}_\ind(w)$ also contains $\partial U$, where
$\mathbf{P}(w)$ is analytic but $\dot{\mathbf{P}}_\ind(w)$ has a
jump discontinuity stemming from the mismatch (because $\zeta$ is large but
finite on $\partial U$) between
$\dot{\mathbf{P}}_\ind^\mathrm{in}(w)$ and $\dot{\mathbf{P}}_\ind^\mathrm{out}(w)$.
The jump contour for $\mathbf{Q}_\ind(w)$ is illustrated in Figure~\ref{fig:critTcontour}.
\begin{figure}[h]
\begin{center}
\includegraphics{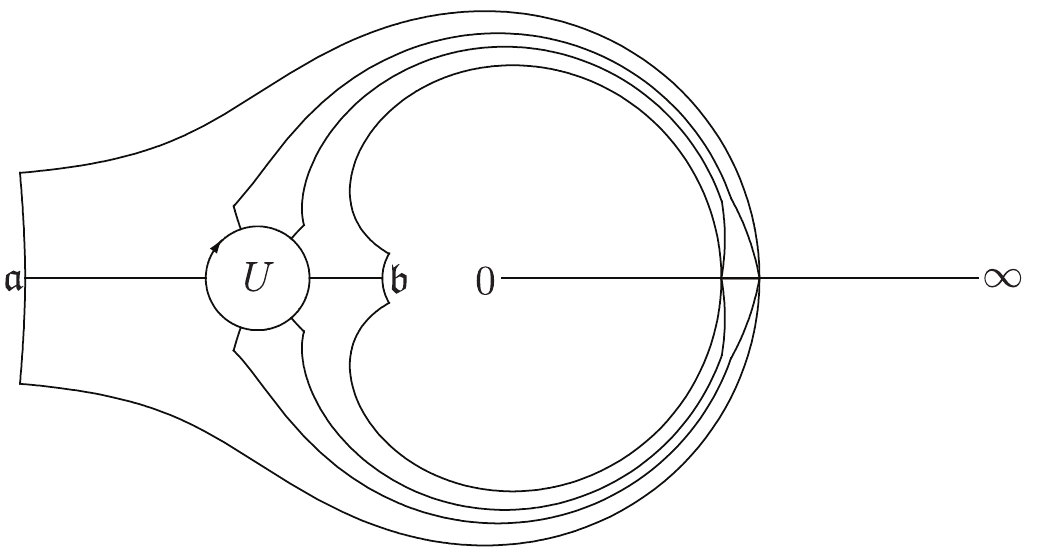}
\end{center}
\caption{\emph{The contour $\Sigma$ of discontinuity of the sectionally analytic
function $\mathbf{Q}_\ind(w)$.   With the exception of the circle $\partial U$ which is clockwise-oriented as shown, all other arcs are oriented exactly as components of the jump contour of $\mathbf{P}(w)$.}}
\label{fig:critTcontour}
\end{figure}

Let $\Sigma$ denote the jump contour for $\mathbf{Q}_\ind(w)$, that is,
$\mathbf{Q}_\ind(w)$ is analytic for $w\in\mathbb{C}\setminus\Sigma$,
and it takes its boundary values in the classical sense 
on this contour from each component of the domain of analyticity.  Also,
since both $\mathbf{P}(w)$ and $\dot{\mathbf{P}}_\ind(w)$ tend to 
the identity matrix as $w\to\infty$ (the former by hypothesis and the
latter by the explicit formula for $\dot{\mathbf{P}}_\ind^\mathrm{out}(w)$),
we have $\mathbf{Q}_\ind(w)\to \mathbb{I}$ as $w\to\infty$.

For fixed $\ind\in\mathbb{Z}$, it turns out that for $(x,t)$ sufficiently close
to criticality,
$\mathbf{Q}_{\ind+}(\xi)=\mathbf{Q}_{\ind-}(\xi)(\mathbb{I}+\bo(\epsilon_N))$ 
as $\epsilon_N\to 0$ for $\xi\in\Sigma\setminus (\mathbb{R}_+\cup\partial U)$,
with the estimate being uniform with respect to $\xi$ on the specified contour arcs and $(x,t)$
near criticality.
Indeed, exactly at criticality we have $\mathfrak{g}(w)=g(w)$ where $g(w)=g(w;x,t)$ is as defined
in \cite[Section 4]{BuckinghamMelliptic} and so by continuity
with respect to $x$ and $t$ we will have near criticality
that the jump matrix for $\mathbf{O}(w)$ is subject to the same error estimates
as recorded in \cite[Proposition~5.1]{BuckinghamMelliptic} generally valid away
from criticality, where the neighborhoods
$U_1$ and $U_2$ in the statement of the proposition are taken to coincide with $U$.  Making the transformation
from $\mathbf{O}(w)$ to $\mathbf{P}(w)$  we learn that 
the jump condition for $\mathbf{P}(w)$ is of the form $\mathbf{P}_+(\xi)=
\mathbf{P}_-(\xi)(\mathbb{I}+\bo(\epsilon_N))$  uniformly for all $\xi$ in the jump contour pictured in Figure~\ref{fig:critPcontour} with the exception of $\xi\in U$
and on the ray $\xi>-1$.  But, $\dot{\mathbf{P}}_\ind(w)$ satisfies exactly
the same jump conditions as does $\mathbf{P}(w)$ within $U$ and 
in the interval $-1<\xi<0$.  The estimate on the jump conditions for $\mathbf{Q}_\ind(w)$
asserted at the beginning of this paragraph
then follow by exact computation.

For $\xi\in\mathbb{R}_+\setminus I$ (recall that $I$ is a small open interval
containing $\xi=1$ as shown in Figure~\ref{fig:critNcontour}), both
$\mathbf{P}(w)$ and $\dot{\mathbf{P}}_\ind(w)=
\dot{\mathbf{P}}^\mathrm{out}_\ind(w)$ satisfy the same jump condition, from which
it follows that for $\xi>0$ but $\xi\not\in I$,
\begin{equation}
\mathbf{Q}_{\ind+}(\xi)=\mathbf{P}_+(\xi)
\dot{\mathbf{P}}_{\ind+}^\mathrm{out}(\xi)^{-1} = 
\sigma_2\mathbf{P}_-(\xi)\sigma_2\left[\sigma_2
\dot{\mathbf{P}}_{\ind-}^\mathrm{out}(\xi)\sigma_2\right]^{-1} = 
\sigma_2\mathbf{P}_-(\xi)\dot{\mathbf{P}}_{\ind-}^\mathrm{out}(\xi)^{-1}\sigma_2
=\sigma_2\mathbf{Q}_{\ind-}(\xi)\sigma_2.
\label{eq:Econjjump}
\end{equation}
For $\xi\in I$, it follows from the fact that $\mathfrak{g}_+(\xi)+\mathfrak{g}_-(\xi)=0$ for $\xi>0$ that \cite[Proposition~3.3]{BuckinghamMelliptic} holds,
implying that $\mathbf{P}(w)$ satisfies
$\mathbf{P}_+(\xi)=\sigma_2\mathbf{P}_-(\xi)\sigma_2(\mathbb{I}+\bo(\epsilon_N))$.
On the other hand, by construction $\dot{\mathbf{P}}_\ind(w)$ satisfies 
the corresponding jump
condition without error term.  Because $\dot{\mathbf{P}}_\ind(w)$ is
uniformly bounded with bounded inverse, by a calculation
similar to \eqref{eq:Econjjump} we see that 
$\mathbf{Q}_{\ind+}(\xi)=\sigma_2
\mathbf{Q}_{\ind-}(\xi)\sigma_2(\mathbb{I}+\bo(\epsilon_N))$ holds uniformly for $\xi\in I\subset\mathbb{R}_+$.

To compute the jump condition for $\mathbf{Q}_\ind(w)$ across the negatively-oriented boundary
$\partial U$ of the disk $U$, first use the fact that $\mathbf{P}(w)$ is continuous across
$\partial U$ to obtain
\begin{equation}
\mathbf{Q}_{\ind+}(\xi)=\mathbf{P}(\xi)
\dot{\mathbf{P}}_\ind^\mathrm{out}(\xi)^{-1} = 
\mathbf{P}(\xi)\dot{\mathbf{P}}_\ind^\mathrm{in}(\xi)^{-1}
\left[\dot{\mathbf{P}}_\ind^\mathrm{in}(\xi)
\dot{\mathbf{P}}_\ind^\mathrm{out}(\xi)^{-1}\right]=
\mathbf{Q}_{\ind-}(\xi)\left[\dot{\mathbf{P}}_\ind^\mathrm{in}(\xi)
\dot{\mathbf{P}}_\ind^\mathrm{out}(\xi)^{-1}\right],\quad \xi\in\partial U.
\end{equation}
By substitution from \eqref{eq:dotPrepDS} and \eqref{eq:dotPnlocdefine},
we see that
\begin{multline}
\dot{\mathbf{P}}_\ind^\mathrm{in}(\xi)
\dot{\mathbf{P}}_\ind^\mathrm{out}(\xi)^{-1}\\{}=
\epsilon_N^{(2\ind-1)\sigma_3/6}e^{\frac{1}{2}s\sigma_3/\epsilon_N}
\left[\eta_\ind(\xi)^{\sigma_3}\mathbf{Z}_\ind(\zeta;y)
(-\zeta)^{(1-2\ind)\sigma_3/2}\eta_\ind(\xi)^{-\sigma_3}\right]
e^{-\frac{1}{2}s\sigma_3/\epsilon_N}\epsilon_N^{(1-2\ind)\sigma_3/6}d_N(\xi)^{-\sigma_3},
\label{eq:Qjumpdiscboundary}
\end{multline}
for $\xi\in\partial U$, where $\zeta=\epsilon_N^{-1/3}W(\xi)$.  

We therefore see that for $(x,t)$ close enough to criticality,
$\mathbf{Q}_\ind(w)$ satisfies the conditions
of the following type of Riemann-Hilbert problem.
\begin{rhp}  
Seek a matrix $\mathbf{Q}_\ind(w)$ with the following properties:
\begin{itemize}
\item[]\textbf{Analyticity:}  $\mathbf{Q}_\ind(w)$ is analytic for $w\in\mathbb{C}\setminus\Sigma$,
where $\Sigma$ is 
the contour (independent of $x$, $t$, $\ind$, and $\epsilon_N$) 
pictured in Figure~\ref{fig:critTcontour}, and in each
component of the domain of analyticity is uniformly H\"older continuous up to the boundary.
\item[]\textbf{Jump conditions:}  The boundary values taken by $\mathbf{Q}_\ind(w)$
on $\Sigma$ are related as follows.  For $\xi\in\partial U\subset\Sigma$
with clockwise orientation, 
\begin{equation}
\mathbf{Q}_{\ind+}(\xi)=\mathbf{Q}_{\ind-}(\xi)\mathbf{V}_{\mathbf{Q}_\ind}(\xi)
\end{equation}
with jump matrix $\mathbf{V}_{\mathbf{Q}_\ind}(\xi)$ defined by the
right-hand side of \eqref{eq:Qjumpdiscboundary}.  For $\xi\in\mathbb{R}_+$,
\begin{equation}
\mathbf{Q}_{\ind+}(\xi)=\sigma_2\mathbf{Q}_{\ind-}(\xi)\sigma_2(\mathbb{I}+\mathbf{X}_\ind(\xi))
\end{equation}
where $\mathbf{X}_\ind(\xi)=\bo(\epsilon_N)$ holds uniformly for $\xi\in I$ and $\mathbf{X}_\ind(\xi)=0$
for $\xi\in\mathbb{R}_+\setminus I$.  For all remaining $\xi\in\Sigma$ we
have the uniform estimate 
\begin{equation}
\mathbf{Q}_{\ind+}(\xi)=\mathbf{Q}_{\ind-}(\xi)(\mathbb{I}+\bo(\epsilon_N))
\end{equation}
as $\epsilon_N\downarrow 0$.
\item[]\textbf{Normalization:}  The matrix $\mathbf{Q}_\ind(w)$ satisfies the
condition
\begin{equation}
\mathop{\lim_{w\to\infty}}_{|\arg(-w)|<\pi}\mathbf{Q}_\ind(w)=\mathbb{I}.
\end{equation}
\end{itemize}
All of the above jump matrix estimates are uniform for bounded $\ind$.
\label{rhp:critQ}
\end{rhp}
Note that since the jump matrix for $\mathbf{Q}_\ind(w)$ is not specified here precisely, the
conditions of Riemann-Hilbert Problem~\ref{rhp:critQ} are not sufficient to determine $\mathbf{Q}_\ind(w)$, but we shall see that they are enough to approximate $\mathbf{Q}_\ind(w)$ with suitable accuracy
as $\epsilon_N\to 0$.

Now $W(\xi)$ is bounded away from zero for $\xi\in\partial U$, so from
the normalization condition on $\mathbf{Z}_\ind(\zeta;y)$
given as part of Riemann-Hilbert Problem~\ref{rhp:DSlocalII} and understood
in the stronger sense of \eqref{eq:Fexpansion},
and from the fact that $\eta_\ind(\xi)$ and its reciprocal are uniformly
bounded as $\epsilon_N\to 0$ for $\xi\in\partial U$, we obtain
\begin{equation}
\dot{\mathbf{P}}_\ind^\mathrm{in}(\xi)\dot{\mathbf{P}}_\ind^\mathrm{out}(\xi)^{-1}=
\epsilon_N^{(2\ind-1)\sigma_3/6}e^{\frac{1}{2}s\sigma_3/\epsilon_N}
\left[\mathbb{I} + \bo(\epsilon_N^{1/3})\right]e^{-\frac{1}{2}s\sigma_3/\epsilon_N}
\epsilon_N^{(1-2\ind)\sigma_3/6}d_N(\xi)^{-\sigma_3},\quad \xi\in\partial U.
\end{equation}
Recall that $d_N(w)=1+\bo(\epsilon_N)$ uniformly for
$w\in\overline{U}$.  It follows that the jump matrix for $\mathbf{Q}_\ind(w)$
will be a small perturbation (of order $\bo(\epsilon_N^\mu)$ for some 
$\mu\in (0,1/3)$) of the identity for $\xi\in\partial U$ if
\begin{equation}
\left|\frac{2\ind-1}{3}-\frac{s}{\epsilon_N\log(\epsilon_N^{-1})}\right|\le
\frac{1}{3}-\mu.
\label{eq:smallnormcondition}
\end{equation}
Subject to this condition, we may formulate a Riemann-Hilbert problem for
$\mathbf{Q}_\ind(w)$ of small-norm type, that is, a problem that can
be solved by iteration.
This calculation leads to the suggestion that different choices of the integer
$\ind$ will be required to control the error in different regions of the $(y,s)$
parameter plane, an idea that has proven useful in at least two very recent works on related
problems of semiclassical analysis for nonlinear waves.  In the work of Claeys and Grava \cite{ClaeysG10sol} on the small-dispersion limit of the KdV equation near the ``trailing edge'' of the oscillatory wavepacket generated by a gradient catastrophe in the
solution of the limiting Burgers equation, the authors find that different choices of an integer parametrizing an outer model problem are required as the rescaled distance from the trailing edge varies.
Similarly, in the work of Bertola and Tovbis \cite{BertolaT10edge} the semiclassical solution
of the focusing NLS equation is examined for points near the ``primary caustic'' or ``breaking curve'' generated from an elliptic umbilic catastrophe in the solution of the formal limiting Whitham system, and it is again found that different choices of parametrices must be made as a certain rescaled distance from the caustic curve varies.  In both of these problems, as will be the case here, the rescaling of distance involves logarithms of the small parameter to get the balance right, and the resulting approximate solutions have a ``solitonic'' nature (in the KdV problem they are simple solitons while in the NLS problem they are soliton superpositions forming breathers).  

To implement this suggestion, we will 
choose a value for $\ind$ and then define corresponding regions of the $(y,s)$
plane in which we can make use of this value of $\ind$.  Also, we will find
it necessary 
to relax the small-norm condition \eqref{eq:smallnormcondition} in order
to cover a large enough part of the $(y,s)$ plane with overlapping regions, 
and this means that
the Riemann-Hilbert problem solved by the matrix $\mathbf{Q}_\ind(w)$
will no longer be treatable by iteration, and instead we shall have to
improve our inadequate model $\dot{\mathbf{P}}_\ind(w)$ 
by building parametrices for $\mathbf{Q}_\ind(w)$, another technique used in
the papers \cite{ClaeysG10sol} and \cite{BertolaT10edge} (also in the more recent
work \cite{BertolaT10cusp} of Bertola and Tovbis where the solution is analyzed in a neighborhood of the elliptic umbilic catastrophe point 
itself).

For each integer $\ind\in\mathbb{Z}$ we define a shifted and scaled coordinate $p_\ind$ equivalent
to $s$ by the relation 
\begin{equation}
s=\frac{2\ind}{3}\epsilon_N\log(\epsilon_N^{-1}) +\epsilon_Np_\ind.
\label{eq:spind}
\end{equation}
\begin{definition}[Sets $\Omega_\ind^\pm$]
Let $\ind\in\mathbb{Z}$.   For suitably small $\delta>0$ and arbitrary $\kappa\ge 0$, set
\begin{multline}
\Omega_\ind^+:=\left\{(y,s)\in\mathbb{R}^2\;\text{with}\;|p_\ind|\le\frac{1}{3}\log(\epsilon_N^{-1})+\kappa\;\text{such that}\right.\\
\begin{aligned} |y-\mathscr{P}(\pu_\ind)|\ge \delta e^{\frac{1}{2}p_\ind} &\quad\text{if} \quad
-\frac{1}{3}\log(\epsilon_N^{-1}) -\kappa\le p_\ind \le 0\quad\text{and}\\
|y-\mathscr{Z}(\pu_\ind)|\le\delta e^{-\frac{1}{2}p_\ind}&\quad\text{if}\quad
\left.0<p_\ind\le\frac{1}{3}\log(\epsilon_N^{-1})+\kappa\right\}
\end{aligned}
\end{multline}
and
\begin{multline}
\Omega_\ind^-:=\left\{(y,s)\in\mathbb{R}^2\;\text{with}\;|p_{\ind-1}|\le\frac{1}{3}\log(\epsilon_N^{-1})+\kappa\;\text{such that}\right.\\
\begin{aligned} |y-\mathscr{P}(\pv_\ind)|\ge \delta e^{-\frac{1}{2}p_{\ind-1}} &\quad\text{if} \quad
0\le p_{\ind-1}\le \frac{1}{3}\log(\epsilon_N^{-1})+\kappa \quad\text{and}\\
|y-\mathscr{Z}(\pv_\ind)|\le\delta e^{\frac{1}{2}p_{\ind-1}}&\quad\text{if}\quad
\left.-\frac{1}{3}\log(\epsilon_N^{-1})-\kappa\le p_{\ind-1}<0\right\}.
\end{aligned}
\end{multline}
\label{def:sets}
\end{definition}
The set $\Omega_\ind^+$ is basically a horizontal strip 
with excluded \emph{notches} 
located near the poles of $\pu_\ind(y)$ and upward-extending \emph{teeth} 
located near the zeros of $\pu_\ind(y)$.  Similarly, $\Omega_\ind^-$ is a horizontal
strip 
with excluded notches 
located near the poles of $\pv_\ind(y)$ and downward-extending teeth 
located near the zeros of $\pv_\ind(y)$.  Since the poles of $\pu_\ind$ and $\pv_\ind$
coincide (although their zeros do not in general), that is, $\mathscr{P}(\pu_\ind)=\mathscr{P}(\pv_\ind)$,
the notches of $\Omega_\ind^+$ align with those of $\Omega_\ind^-$.  This is very important, since
we will make use of the solution $\mathbf{Z}_\ind(\zeta;y)$ of Riemann-Hilbert Problem~\ref{rhp:DSlocalII}
for $(y,s)\in\Omega_\ind^+\cup\Omega_\ind^-$, and therefore we must exclude from this union
points with $y$-values too close to points of $\mathscr{P}(\pu_\ind)=\mathscr{P}(\pv_\ind)$
at which $\mathbf{Z}_\ind(\zeta;y)$ fails to exist.  See Figure~\ref{fig:twomorestrips}.
\begin{figure}[h]
\begin{center}
\includegraphics{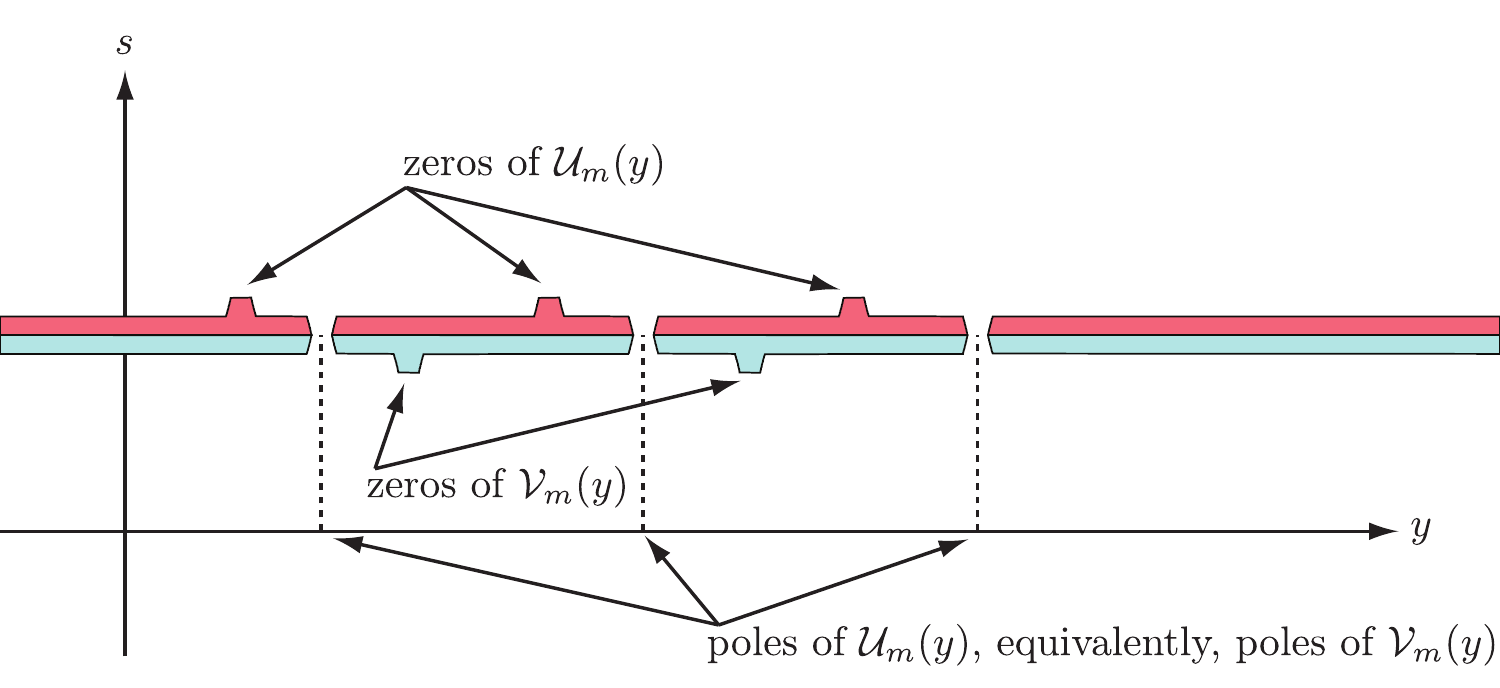}
\end{center}
\caption{\emph{The regions $\Omega_\ind^+$ (pink) and $\Omega_\ind^-$ (blue) for $\ind=6$.   For this figure, $\kappa=0$, $\delta=0.2$ and $\epsilon_N=0.1$.  Note the alignment of the notches.   It is for $(y,s)$ in the
union of these regions that we select the value $m=6$ in both the solution $\dot{\mathbf{P}}_\ind^\mathrm{out}(w)$ of the outer model problem and the solution $\mathbf{Z}_\ind(\zeta;y)$ of the inner model problem.}}
\label{fig:twomorestrips}
\end{figure}

Since $\pv_{\ind+1}(y)=1/\pu_{\ind}(y)$ according to  \eqref{eq:Baecklundplus}, it is clear that
$\mathscr{P}(\pu_{\ind+1})=\mathscr{Z}(\pv_\ind)$ and $\mathscr{Z}(\pu_{\ind+1})=\mathscr{P}(\pv_\ind)$.   It follows that while the
sets $\Omega_\ind^+$ and $\Omega_{\ind+1}^-$ have disjoint interiors,
their union is a horizontal strip:
\begin{equation}
\Omega_\ind^+\cup\Omega_{\ind+1}^- = 
\left\{(y,s)\in\mathbb{R}^2\quad\text{such that}\quad |p_\ind|\le \frac{1}{3}\log(\epsilon_N^{-1})+\kappa\right\}.
\label{eq:stripunion}
\end{equation}
The two regions interlock near the singularities
of $\log|\pu_\ind(y)|$ (equivalently the singularities of $\log|\pv_{\ind+1}(y)|$);
the teeth of $\Omega_\ind^+$ fit in a one-to-one fashion into the notches of 
$\Omega_{\ind+1}^-$, and 
the teeth of $\Omega_{\ind+1}^-$ fit in a one-to-one fashion into the notches
of $\Omega_\ind^+$.
See Figure~\ref{fig:twostrips}.  In Figure~\ref{fig:CriticalRegionsAnnotated}
we illustrate the tiling of a region of the $(y,s)$-plane with many such strips when $\kappa=0$ (the strips overlap if $\kappa>0$).
\begin{figure}[h]
\begin{center}
\includegraphics{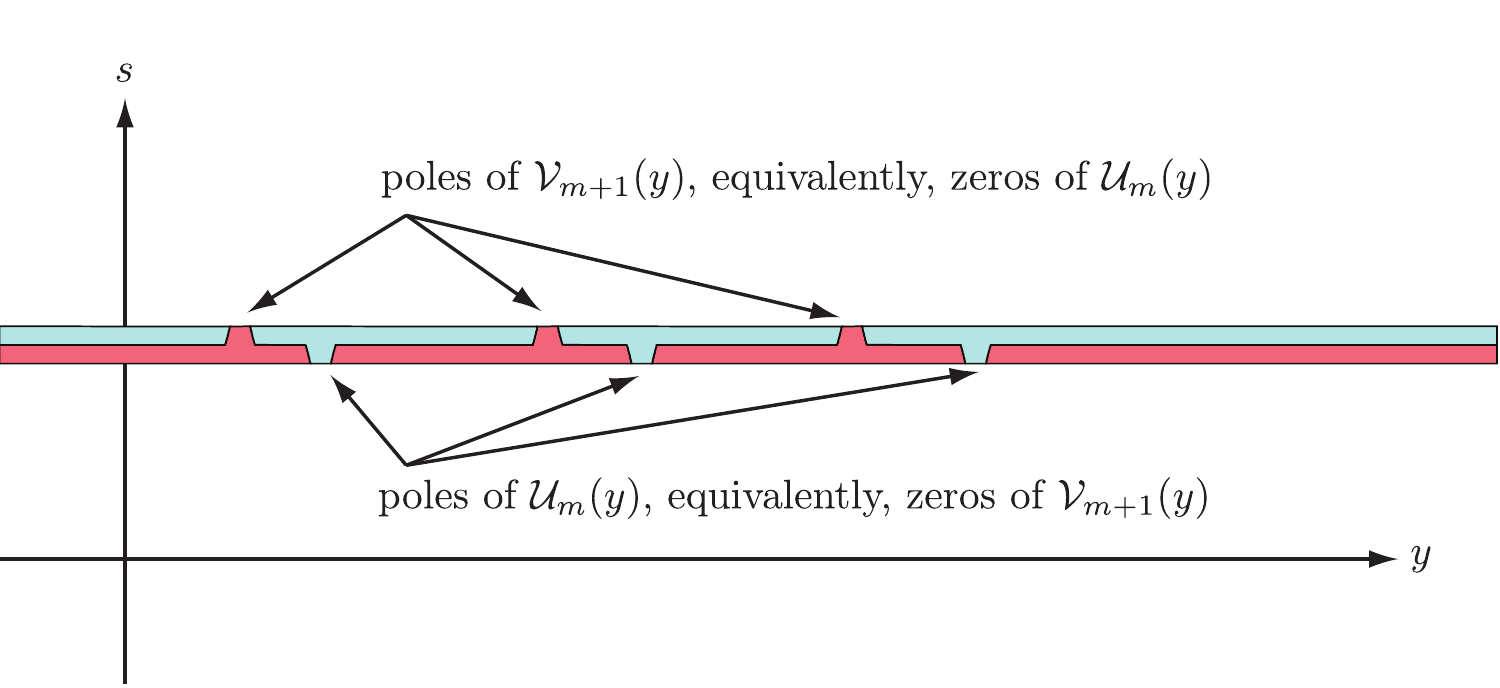}
\end{center}
\caption{\emph{The regions $\Omega_\ind^+$ (pink) and $\Omega_{\ind+1}^-$ (blue) for $\ind=6$.
Their union is a closed horizontal strip in the $(y,s)$-plane.  For this
figure, $\kappa=0$, $\delta=0.2$ and $\epsilon_N=0.1$. 
}}
\label{fig:twostrips}
\end{figure}
\begin{figure}[h]
\begin{center}
\includegraphics{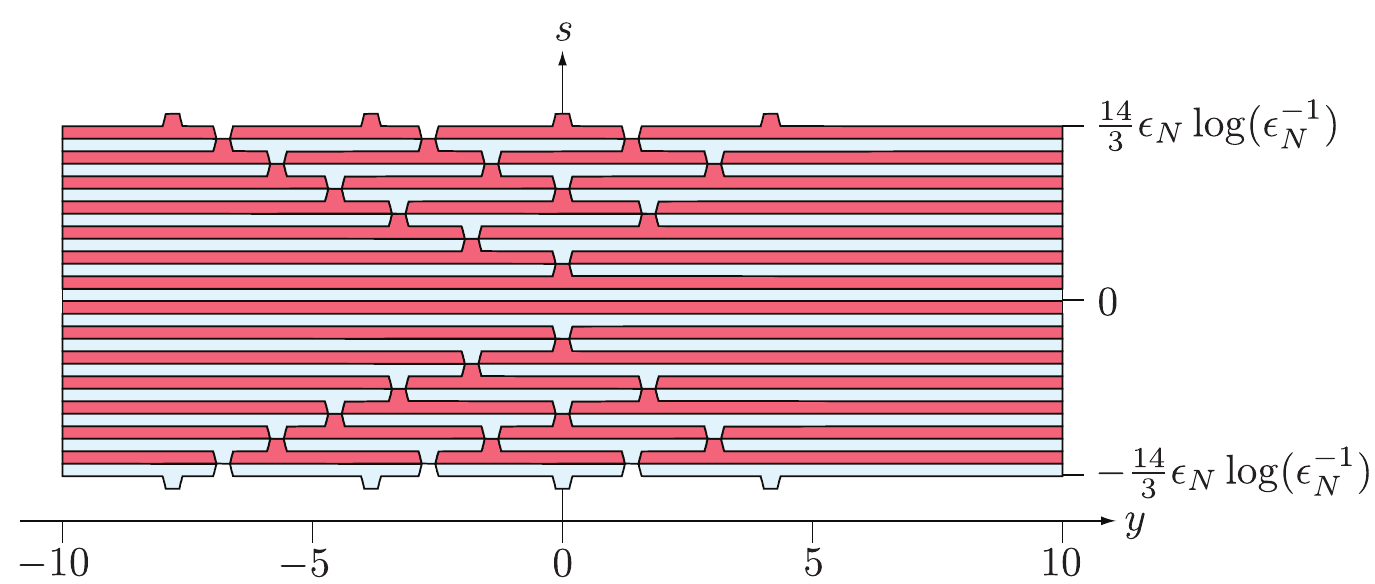}
\end{center}
\caption{\emph{The regions $\Omega_{-6}^+,\dots,\Omega_7^+$ (pink) and $\Omega_{-6}^-,\dots,\Omega_7^-$ (blue), for the ``tiling'' case $\kappa=0$.}}
\label{fig:CriticalRegionsAnnotated}
\end{figure}
This latter figure illustrates another feature of the $\kappa=0$ tiling:  since by the B\"acklund
transformations \eqref{eq:Baecklundplus}--\eqref{eq:Baecklundminus},
$y_0$ is a simple zero of $\pu_{\ind-1}(y)$ if and only if $y_0$ is a simple pole of both
$\pu_\ind(y)$ and $\pv_\ind(y)$ if and only if $y_0$ is a simple zero of $\pv_{\ind+1}(y)$, the 
tips of the teeth of $\Omega_{\ind-1}^+$ are matched in a one-to-one fashion
onto the tips of the teeth of $\Omega_{\ind+1}^-$. 

We will use Definition~\ref{def:sets} in the following context:  we fix a bound $B>0$ and consider only indices $\ind\in\mathbb{Z}$ with $|\ind|\le B$;  we then choose $\delta>0$
sufficiently small that for all $\ind\in\mathbb{Z}$ with $|\ind|\le B$,
no two real singularities of $\log|\pu_\ind(y)|$ 
 or $\log|\pv_\ind(y)|$ 
are closer\footnote{Preliminary asymptotic analysis suggests that $\delta = \bo(B^{-1/3})$ as $B\uparrow\infty$, which is one indication of failure of uniformity of the analysis in this paper when $t\sim s\gg \epsilon_N\log(\epsilon_N^{-1})$.  Indeed, for $t>0$ fixed and $x\approx x_\mathrm{crit}$ fixed one expects a different description of the sine-Gordon dynamics for sufficiently small $\epsilon_N$ in terms of higher-genus Riemann theta functions.} than $2\delta$.  This prevents
neighboring teeth/notches from overlapping at their widest point (corresponding to $p_\ind=0$).

\subsection{Parametrix for $\mathbf{Q}_\ind(w)$ in $\Omega_\ind^+$.}
\label{sec:parametrixplus}
Even with the choice of the integer $\ind$ corresponding to $(y,s)\in\Omega_\ind^\pm$
it turns out that Riemann-Hilbert Problem~\ref{rhp:critQ} cannot generally be solved with
the use of small-norm theory as the jump matrix $\mathbf{V}_{\mathbf{Q}_\ind}(\xi)$ is not 
close to the identity for $\xi\in\partial U$, at least not uniformly for $(y,s)\in\Omega_\ind^\pm$.
Therefore, even though we may think of $\mathbf{Q}_\ind(w)$ as an ``error'' matrix gauging the mismatch between the matrix function $\dot{\mathbf{P}}_\ind(w)$ and the matrix function $\mathbf{P}(w)$ for which the former was introduced as a model,  we must in general carry out further modeling of $\mathbf{Q}_\ind(w)$ by creating an explicit model (parametrix) for it, and comparing the model with $\mathbf{Q}_\ind(w)$ in hopes of --- finally --- arriving at a small-norm problem for the matrix ratio.  We now proceed to build such a \emph{parametrix for the error} by a procedure
referred to in some papers \cite{ClaeysG10osc,ClaeysG10sol} as \emph{improvement of the parametrices}.

Let $\ind\in \mathbb{Z}$ with $|\ind|\le B$ be fixed, and suppose that
$(y,s)\in \Omega_\ind^+$ with $y\in\mathbb{R}$ bounded.
In particular, this implies that the coordinate $p_\ind$ defined by \eqref{eq:spind} is subjected to the inequalities $|p_\ind|\le\tfrac{1}{3}\log(\epsilon_N^{-1})+\kappa$.
Using the coordinate $p_\ind$, and recalling Proposition~\ref{prop:Rnlocgeneral} along with the
expansion \eqref{eq:Fexpansion}, we write the matrix $\mathbf{V}_{\mathbf{Q}_\ind}(\xi)$ defined by the right-hand
side of \eqref{eq:Qjumpdiscboundary} in the form (assuming only that $y$ is bounded)
\begin{equation}
\begin{split}
\mathbf{V}_{\mathbf{Q}_\ind}(\xi)&=\exp\left(\frac{1}{2}\left(p_\ind+\frac{1}{3}\log(\epsilon_N^{-1})\right)\sigma_3\right)\\
&\quad\quad{}\cdot\eta_\ind(\xi)^{\sigma_3}\left[\mathbb{I}+\epsilon_N^{1/3}\mathbf{A}_\ind(y)W(\xi)^{-1}
+\epsilon_N^{2/3}\mathbf{B}_\ind(y)W(\xi)^{-2} +\bo\left(\frac{\epsilon_N}{|y-\mathscr{P}(\pu_\ind)|}\right)\right]
\eta_\ind(\xi)^{-\sigma_3}\\
&\quad\quad{}\cdot\exp\left(-\frac{1}{2}\left(p_\ind+\frac{1}{3}\log(\epsilon_N^{-1})\right)
\sigma_3\right)d_N(\xi)^{-\sigma_3},\quad\xi\in\partial U
\end{split}
\label{eq:VQomegaplus}
\end{equation}
because $W(\xi)$ is bounded away from zero on $\partial U$.
Since $\eta_\ind(\xi)$, $d_N(\xi)$, and their reciprocals are all bounded on $\partial U$ and $p_\ind\ge -\tfrac{1}{3}\log(\epsilon_N^{-1})-\kappa$ for $(y,s)\in\Omega_\ind^+$, 
\begin{equation}
\begin{split}
\mathbf{V}_{\mathbf{Q}_\ind}(\xi)&=\exp\left(\frac{1}{2}\left(p_\ind+\frac{1}{3}\log(\epsilon_N^{-1})\right)\sigma_3\right)\\
&\quad\quad{}\cdot\eta_\ind(\xi)^{\sigma_3}\left[\mathbb{I}+\epsilon_N^{1/3}\mathbf{A}_\ind(y)W(\xi)^{-1}
+\epsilon_N^{2/3}\mathbf{B}_\ind(y)W(\xi)^{-2}\right]
\eta_\ind(\xi)^{-\sigma_3}\\
&\quad\quad{}\cdot\exp\left(-\frac{1}{2}\left(p_\ind+\frac{1}{3}\log(\epsilon_N^{-1})\right)
\sigma_3\right)d_N(\xi)^{-\sigma_3} + \bo\left(\frac{\epsilon_N^{2/3}e^{p_\ind}}{|y-\mathscr{P}(\pu_\ind)|}\right)
\end{split}
\label{eq:VQomegaplus1}
\end{equation}
holds for $(y,s)\in\Omega_\ind^+$ with $y$ bounded and $\xi\in\partial U$.
Now, we want to neglect all but the $(1,2)$-entry of the matrix $\epsilon_N^{1/3}\mathbf{A}_\ind(y)W(\xi)^{-1} + 
\epsilon_N^{2/3}\mathbf{B}_\ind(y)W(\xi)^{-2}$.  Since the matrix elements of $\mathbf{A}_\ind(y)$ and $\mathbf{B}_\ind(y)$ are $\bo(|y-\mathscr{P}(\pu_\ind)|^{-1})$,
we may do this at the cost of adding a new error term of the
form $\bo(\epsilon_N^{1/3}|y-\mathscr{P}(\pu_\ind)|^{-1})$ (dominated by the diagonal entries).
This error term dominates the one already present in \eqref{eq:VQomegaplus1} for $(y,s)\in\Omega_\ind^+$ as a consequence of the inequality $p_\ind\le\tfrac{1}{3}\log(\epsilon_N^{-1})+\kappa$.  Therefore,
\begin{equation}
\mathbf{V}_{\mathbf{Q}_\ind}(\xi)=\begin{bmatrix}1 & Y_\ind^+(\xi)\\0 & 1\end{bmatrix}d_N(\xi)^{-\sigma_3} +
\bo\left(\frac{\epsilon_N^{1/3}}{|y-\mathscr{P}(\pu_\ind)|}\right),
\quad \xi\in\partial U,\quad \text{$(y,s)\in\Omega_\ind^+$ with $y$ bounded,}
\label{eq:EjumppartialU1}
\end{equation}
where, defining
\begin{equation}
h_\ind^+(w):=\eta_\ind(w)^2, 
\label{eq:hindplus}
\end{equation}
the explicit function appearing in \eqref{eq:EjumppartialU1} is
\begin{equation}
\begin{split}
Y^+_{\ind}(w):=&\;e^{p_\ind}h_\ind^+(w)\left(A_{\ind,12}(y)W(w)^{-1}+\epsilon_N^{1/3}B_{\ind,12}(y)W(w)^{-2}\right)
\\
{}=&\;e^{p_\ind}h_\ind^+(w)\left(\pu_\ind(y)W(w)^{-1} +
\epsilon_N^{1/3}[2H_\ind(y)\pu_\ind(y)-\pu_\ind'(y)]W(w)^{-2}\right).
\end{split}
\label{eq:Ynplusdef}
\end{equation}
Note that $Y_\ind^+(w)$ extends from $\xi\in\partial U$ as a meromorphic function of $w\in \overline{U}$ whose only singularity
is a double pole at $w=w_*$, and that $Y_\ind^+(\xi)$ also depends parametrically on $p_\ind$ and $y$.  Since $d_N(\xi)=1+\bo(\epsilon_N)$ holds uniformly for $\xi\in\partial U$, another way to write \eqref{eq:EjumppartialU1}
is
\begin{multline}
\mathbf{V}_{\mathbf{Q}_\ind}(\xi)\begin{bmatrix}1 & Y_\ind^+(\xi)\\0 & 1\end{bmatrix}^{-1} = 
\mathbb{I} + \bo\left(\frac{\epsilon_N^{1/3}}{|y-\mathscr{P}(\pu_\ind)|}\right) + \bo\left(\frac{Y_\ind^+(\xi)\epsilon_N^{1/3}}{|y-\mathscr{P}(\pu_\ind)|}\right),\\
\xi\in\partial U,\quad\text{$(y,s)\in\Omega_\ind^+$ with $y$ bounded.}
\label{eq:plusjumpestimate}
\end{multline}

\begin{proposition}
The following three estimates hold uniformly for $\xi\in\partial U$ and $y$ bounded.\\
In the upward-pointing ``teeth'' of $\Omega_\ind^+$:
\begin{multline}
\mathbf{V}_{\mathbf{Q}_\ind}(\xi)\begin{bmatrix}1 & Y_\ind^+(\xi)\\ 0 & 1\end{bmatrix}^{-1}=
\mathbb{I}+\bo(\epsilon_N^{1/3}e^{p_\ind}|y-\mathscr{Z}(\pu_\ind)|),\\(y,s)\in\Omega_\ind^+,
\quad p_\ind>0\;\text{and}\; |y-\mathscr{Z}(\pu_\ind)|\ge\delta e^{-p_\ind}.
\label{eq:plusestimateteeth}
\end{multline}
In the part of $\Omega_\ind^+$ resembling a ``strip with notches'':
\begin{equation}
\mathbf{V}_{\mathbf{Q}_\ind}(\xi)\begin{bmatrix}1 & Y_\ind^+(\xi)\\ 0 & 1\end{bmatrix}^{-1}=
\mathbb{I}+\bo\left(\frac{\epsilon_N^{1/3}}{|y-\mathscr{P}(\pu_\ind)|}\right),\quad
(y,s)\in\Omega_\ind^+,\quad |y-\mathscr{P}(\pu_\ind)|\le\delta.
\label{eq:plusestimatenotches}
\end{equation}
Elsewhere in $\Omega_\ind^+$:
\begin{equation}
\mathbf{V}_{\mathbf{Q}_\ind}(\xi)\begin{bmatrix}1 & Y_\ind^+(\xi)\\ 0 & 1\end{bmatrix}^{-1}=
\mathbb{I}+\bo(\epsilon_N^{1/3}).
\label{eq:plusestimateelsewhere}
\end{equation}
\label{prop:VQerrorplus}
\end{proposition}
\begin{proof}
First suppose that $0<p_\ind\le\frac{1}{3}\log(\epsilon_N^{-1})+\kappa$ (so we are considering $(y,s)$ to lie in
the upward-pointing ``teeth''  of $\Omega_\ind^+$) and that $y$ is bounded.    In this sub-region, $|y-\mathscr{P}(\pu_\ind)|$ is bounded away from zero because the teeth are localized near the zeros $\mathscr{Z}(\pu_\ind)$ of $\pu_\ind(y)$.  Therefore
$2H_\ind\pu_\ind-\pu_\ind'$ is bounded and $\pu_\ind(y)=\bo(|y-\mathscr{Z}(\pu_\ind)|)$, so $Y_\ind^+(\xi)=\bo(e^{p_\ind}|y-\mathscr{Z}(\pu_\ind)|) + \bo(1)$ holds uniformly for $\xi\in\partial U$.  Putting this estimate into \eqref{eq:plusjumpestimate} gives
\begin{multline}
\mathbf{V}_{\mathbf{Q}_\ind}(\xi)\begin{bmatrix}1 & Y_\ind^+(\xi)\\0 & 1\end{bmatrix}^{-1}=\mathbb{I}
+\bo(\epsilon_N^{1/3}) + \bo(\epsilon_N^{1/3}e^{p_\ind}|y-\mathscr{Z}(\pu_\ind)|),\\
\xi\in\partial U,\quad \text{$(y,s)\in\Omega_\ind^+$ with $y$ bounded and $p_\ind>0$.}
\end{multline}
This proves the estimate \eqref{eq:plusestimateteeth}, and it also proves the estimate \eqref{eq:plusestimateelsewhere} in the case that $p_\ind>0$.

Now suppose instead that $-\frac{1}{3}\log(\epsilon_N^{-1})-\kappa\le p_\ind\le 0$ (so we are considering $(y,s)$ to lie in the complementary part of $\Omega_\ind^+$ that resembles a ``strip with notches''), and that $y$ is bounded.  Since according to Proposition~\ref{prop:Rnlocgeneral}, both $\pu_\ind(y)$ 
and $2H_\ind(y)\pu_\ind(y)-\pu_\ind'(y)$ have simple poles at the points of $\mathscr{P}(\pu_\ind)$, it follows that in this sub-region,
$Y_\ind^+(\xi)=\bo(e^{p_\ind}|y-\mathscr{P}(\pu_\ind)|^{-1})$.  But in this subregion we also have the inequality $|y-\mathscr{P}(\pu_\ind)|\ge\delta e^{\frac{1}{2}p_\ind}$, so with $p_\ind\le 0$ in fact we have $Y_\ind^+(\xi)=\bo(1)$.  Putting this information into \eqref{eq:plusjumpestimate}
gives
\begin{equation}
\mathbf{V}_\mathbf{Q}(\xi)\begin{bmatrix}1 & Y_\ind^+(\xi)\\0 & 1\end{bmatrix}^{-1}=
\mathbb{I} + \bo\left(\frac{\epsilon_N^{1/3}}{|y-\mathscr{P}(\pu_\ind)|}\right),\;\;
\xi\in\partial U,\quad\text{$(y,s)\in\Omega_\ind^+$ with $y$ bounded and $p_\ind\le 0$.}
\end{equation}
This proves \eqref{eq:plusestimatenotches}, and it also proves \eqref{eq:plusestimateelsewhere} 
in the case that $p_\ind\le 0$.
\end{proof}

\begin{corollary}
The following estimate holds uniformly for $\xi\in\partial U$ and $(y,s)\in\Omega_\ind^+$ with
$y$ bounded:
\begin{equation}
\mathbf{V}_{\mathbf{Q}_\ind}(\xi)\begin{bmatrix}1 & Y_\ind^+(\xi)\\0 & 1\end{bmatrix}^{-1}=
\mathbb{I}+\bo(\epsilon_N^{1/6}).
\end{equation}
While uniform, this estimate fails to be sharp except possibly at the vertical extremes of $\Omega_\ind^+$ where $|p_\ind|\sim \frac{1}{3}\log(\epsilon_N^{-1})$.
\label{cor:VQerrorplus}
\end{corollary}

As $\epsilon_N\downarrow 0$, the function $Y^+_\ind(\xi)$ is not uniformly small on $\partial U$ as $(y,s)$ ranges over
$\Omega_\ind^+$ with $y$ bounded.
This implies that to 
approximate $\mathbf{Q}_\ind(w)$ for $(y,s)\in \Omega_\ind^+$,
it will be necessary to build a parametrix for it
that takes into account the explicit terms in
\eqref{eq:EjumppartialU1}.

\begin{rhp}[Parametrix for $\mathbf{Q}_\ind(w)$ in $\Omega_\ind^+$]
Let real numbers $y$ and $p_\ind$ be given, 
and assume that $y\not\in\mathscr{P}(\pu_\ind)$.  Seek a matrix $\dot{\mathbf{Q}}_\ind^+(w)$
with the following properties:
\begin{itemize}
\item[]\textbf{Analyticity:}  $\dot{\mathbf{Q}}_\ind^+(w)$ is analytic for
$w\in\mathbb{C}\setminus(\mathbb{R}_+\cup\partial U)$, and is uniformly
H\"older continuous up to the boundary of its domain of
analyticity.  
\item[]\textbf{Jump conditions:}  The boundary values taken by 
$\dot{\mathbf{Q}}_\ind^+(w)$ on the contour 
$\mathbb{R}_+\cup\partial U$ are related
as follows:
\begin{equation}
\dot{\mathbf{Q}}^+_{\ind+}(\xi)=\sigma_2\dot{\mathbf{Q}}^+_{\ind-}(\xi)\sigma_2,\quad \xi>0,
\end{equation}
and
\begin{equation}
\dot{\mathbf{Q}}^+_{\ind+}(\xi)=\dot{\mathbf{Q}}^+_{\ind-}(\xi)
\begin{bmatrix}1 & Y^+_\ind(\xi)\\
0 & 1\end{bmatrix},\quad \xi\in\partial U
\end{equation}
where $\partial U$ is taken to be negatively (clockwise) oriented and
$Y_\ind^+(\xi)$ is defined by \eqref{eq:Ynplusdef}.
\item[]\textbf{Normalization:}  the matrix $\dot{\mathbf{Q}}_\ind^+(w)$ satisfies
the condition
\begin{equation}
\lim_{w\to\infty}\dot{\mathbf{Q}}_\ind^+(w)=\mathbb{I}.
\end{equation}
\end{itemize}
\label{rhp:Eparametrix}
\end{rhp}
We solve this problem uniquely for $\dot{\mathbf{Q}}_\ind^+(w)$ as follows.  
Firstly,
define a related unknown $\mathbf{F}_\ind^+(w)$ as follows:
\begin{equation}
\mathbf{F}_\ind^+(w):=\begin{cases}
\dot{\mathbf{Q}}_\ind^+(w),\quad & w\not\in \overline{U}\\
\displaystyle \dot{\mathbf{Q}}_\ind^+(w)
\begin{bmatrix}1 & Y^+_\ind(w)\\ 0 & 1\end{bmatrix},
\quad & w\in U.
\end{cases}
\label{eq:Fnplusdef}
\end{equation}
Clearly, $\mathbf{F}_\ind^+(w)$ is continuous across the contour $\partial U$.
Now, $\eta_\ind(w)$ is analytic in $U$, but $W(w)$ has a simple zero at 
$w=w_*\approx -1$, so $\mathbf{F}_\ind^+(w)$ generally has a double pole at 
$w=w_*$.  
To handle the only remaining jump condition for $\mathbf{F}^+_\ind(w)$, namely
$\mathbf{F}^+_{\ind+}(\xi)=
\sigma_2\mathbf{F}^+_{\ind-}(\xi)
\sigma_2$ for $\xi>0$, we introduce $\sv:=i(-w)^{1/2}$ and set
\begin{equation}
\mathbf{G}_\ind^+(\sv):=\begin{cases} \mathbf{F}_\ind^+(\sv^2),\quad &\Im\{\sv\}>0\\
\sigma_2\mathbf{F}_\ind^+(\sv^2)\sigma_2,\quad &\Im\{\sv\}<0.
\end{cases}
\label{eq:GFdef}
\end{equation}
It is now evident that $\mathbf{G}_\ind^+(\sv)$ is a rational function of
$\sv\in\mathbb{C}$ with simple poles at $\sv=\pm iS$ where $S:=(-w_*)^{1/2}>0$, normalized
to the identity at $\sv=\infty$, and satisfying the symmetry relation
$\mathbf{G}_\ind^+(-\sv)=\sigma_2\mathbf{G}_\ind^+(\sv)\sigma_2$.  

To fully characterize the singularities of $\mathbf{G}_\ind^+(\sv)$, we will
relate its principal part at $\sv=iS$ to its regular part.
Note that in a neighborhood of $\sv=iS$, $Y^+_\ind(\sv^2)$ has the 
Laurent expansion
\begin{equation}
Y^+_\ind(\sv^2)=\frac{J^+_\ind}{(\sv-iS)^{2}}+\frac{iK^+_\ind}{\sv-iS} + \bo(1)
\end{equation}
where $\bo(1)$ represents a function analytic at $\sv=iS$ and where
\begin{equation}
\begin{split}
J^+_\ind&:=\frac{e^{p_\ind}h^+_\ind(w_*)\epsilon_N^{1/3}}{4w_*W'(w_*)^2}
\Big(2H_\ind(y)\pu_\ind(y)-\pu_\ind'(y)\Big)\\
K^+_\ind&:=-\frac{e^{p_\ind}}{4w_*SW'(w_*)^3}
\Bigg[2w_*h_\ind^+(w_*)W'(w_*)^2\pu_\ind(y) \\
&\quad{}+\epsilon_N^{1/3}\Big(2H_\ind(y)\pu_\ind(y)-\pu_\ind'(y)\Big)\Big(2w_*h^{+\prime}_\ind(w_*)W'(w_*)-2w_*h^+_\ind(w_*)W''(w_*)-h^+_\ind(w_*)W'(w_*)\Big)\Bigg].
\end{split}
\label{eq:JKplus}
\end{equation}
Note also that $J^+_\ind$ and $K^+_\ind$ are both real when $y\in\mathbb{R}$ and 
$p_\ind\in\mathbb{R}$.
Let us write the Laurent expansion of $\mathbf{G}_\ind^+(\sv)$ about 
$\sv=iS$ in the form
\begin{equation}
\mathbf{G}_\ind^+(\sv)=\frac{(\mathbf{a}^+_\ind,\mathbf{b}^+_\ind)}
{(\sv-iS)^{2}} + 
\frac{(\mathbf{c}^{+}_\ind,\mathbf{d}^+_\ind)}{\sv-iS} 
+
(\mathbf{e}^+_\ind,\mathbf{f}^+_\ind) +
(\mathbf{g}^+_\ind,\mathbf{h}^+_\ind)(\sv-iS) + \bo((\sv-iS)^2),
\label{eq:LaurentGnplus}
\end{equation}
where $\bo((\sv-iS)^2)$ represents a matrix function analytic
and vanishing to second order at $\sv=iS$.  Since
$\dot{\mathbf{Q}}_\ind^+(w)$ is meant to be analytic at $w=w_*$, from
\eqref{eq:Fnplusdef} we see that the principal part of
$\mathbf{G}_\ind^+(\sv)$ at $\sv=iS$ is determined from its holomorphic part
as follows.  The first column vanishes:
\begin{equation}
\mathbf{a}^+_\ind=\mathbf{c}^+_\ind=\mathbf{0},
\label{eq:GplusPPfirstcol}
\end{equation}
and the second column satisfies the linear relations
\begin{equation}
\mathbf{b}^+_\ind=J^+_\ind\mathbf{e}^+_\ind\quad\text{and}\quad
\mathbf{d}^+_\ind=J^+_\ind\mathbf{g}^+_\ind + iK^+_\ind\mathbf{e}^+_\ind.
\label{eq:GplusPPsecondcol}
\end{equation}
Taking into account the symmetry condition
$\mathbf{G}_\ind^+(-\sv)=\sigma_2\mathbf{G}_\ind^+(\sv)\sigma_2$ and the
notation of the principal part of the Laurent expansion
\eqref{eq:LaurentGnplus}, the matrix $\mathbf{G}_\ind^+(\sv)$ necessarily
can be expressed in partial fractions as follows:
\begin{equation}
\mathbf{G}_\ind^+(\sv)=\mathbb{I} +
\frac{(\mathbf{0},\mathbf{b}^+_\ind)}{(\sv-iS)^2}
+\frac{(\mathbf{0},\mathbf{d}^+_\ind)}{\sv-iS} +
\frac{(i\sigma_2\mathbf{b}^+_\ind,\mathbf{0})}{(\sv+iS)^2}+
\frac{(-i\sigma_2\mathbf{d}^+_\ind,\mathbf{0})}{\sv+iS},
\label{eq:GnplusPF}
\end{equation}
where we have also used \eqref{eq:GplusPPfirstcol}.
Evaluating the first column of this expression and its derivative at
the regular point (for the first column) $\sv=iS$ we find that, 
in the notation of the regular part of the Laurent expansion 
\eqref{eq:LaurentGnplus},
\begin{equation}
\mathbf{e}^+_\ind = \begin{bmatrix}1\\0\end{bmatrix}
+\frac{i\sigma_2}{4w_*}\mathbf{b}^+_\ind -
\frac{\sigma_2}{2S}\mathbf{d}^+_\ind\quad\text{and}\quad
\mathbf{g}^+_\ind = -\frac{\sigma_2}{4w_*S}\mathbf{b}^+_\ind
+\frac{i\sigma_2}{4w_*}\mathbf{d}^+_\ind,\quad S:=(-w_*)^{1/2}>0.
\end{equation}
Substituting these expressions into the relations \eqref{eq:GplusPPsecondcol}
yields a closed inhomogeneous linear system for the vectors $\mathbf{b}^+_\ind$
and $\mathbf{d}^+_\ind$:
\begin{equation}
\begin{split}
\left(\mathbb{I}-\frac{iJ^+_\ind}{4w_*}\sigma_2\right)\mathbf{b}^+_\ind
+\frac{J^+_\ind}{2S}\sigma_2\mathbf{d}^+_\ind &= \begin{bmatrix}
J^+_\ind\\0\end{bmatrix}\\
\frac{J^+_\ind+SK^+_\ind}{4w_*S}\sigma_2
\mathbf{b}^+_\ind +
\left(\mathbb{I}-\frac{iJ^+_\ind+2iSK^+_\ind}{4w_*}\sigma_2
\right)\mathbf{d}^+_\ind&=\begin{bmatrix}iK^+_\ind\\0\end{bmatrix}.
\end{split}
\end{equation}
This system has a unique solution:
\begin{equation}
\begin{split}
\mathbf{b}^+_\ind&=\frac{4w_*J^+_\ind}{D^+_\ind}\begin{bmatrix}
4w_*(16w_*^2+2SJ^+_\ind K^+_\ind+3[J_\ind^+]^{2})\\
[J_\ind^+]^{3}-16w_*^2J^+_\ind-32w_*^2SK^+_\ind
\end{bmatrix}\\
\mathbf{d}^+_\ind&=\frac{4iS}{D^+_\ind}\begin{bmatrix}
4w_*(2[J_\ind^+]^{3}+S[J_\ind^+]^{2}K^+_\ind-16w_*^2SK^+_\ind)\\
[J_\ind^+]^{4}+16w_*^2[J_\ind^+]^{2}+32w_*^2SJ^+_\ind K^+_\ind-32w_*^3[K_\ind^+]^{2}
\end{bmatrix},
\end{split}
\label{eq:bdplus}
\end{equation}
where the denominator 
\begin{equation}
D^+_\ind:=(16w_*^2+[J_\ind^+]^{2})^2 + 64w_*^2(J^+_\ind+SK^+_\ind)^2\ge 256w_*^4
\label{eq:denomplus}
\end{equation}
is strictly positive, because $w_*\approx -1$.  Substituting
\eqref{eq:bdplus}--\eqref{eq:denomplus} into the partial fraction representation
\eqref{eq:GnplusPF} completes the construction of $\mathbf{G}_\ind^+(\sv)$.
By writing $\sv=i(-w)^{1/2}$ and restricting $\mathbf{G}_\ind^+(\sv)$ to the upper
half-plane we then recover $\mathbf{F}_\ind^+(w)$, and then from
\eqref{eq:Fnplusdef} we recover the solution
$\dot{\mathbf{Q}}_\ind^+(w)$ of Riemann-Hilbert Problem~\ref{rhp:Eparametrix}.

\begin{proposition}
  For each fixed integer $\ind$ with $|\ind|\le B$,
  Riemann-Hilbert Problem~\ref{rhp:Eparametrix} has a unique and explicit solution
  $\dot{\mathbf{Q}}_\ind^+(w)$ satisfying
  $\det(\dot{\mathbf{Q}}_\ind^+(w))=1$ where defined.  Also, $\dot{\mathbf{Q}}_\ind^+(w)$
  is uniformly bounded for $w\in\mathbb{C}\setminus (U\cup\mathbb{R}_+)$ and $(y,s)\in\Omega_\ind^+$
  with $y$ bounded  as $\epsilon_N\to 0$.
\label{prop:parametrixplus}
\end{proposition}
\begin{proof}
The above construction clearly yields a solution of Riemann-Hilbert 
Problem~\ref{rhp:Eparametrix}, and uniqueness follows easily from the conditions
of that problem via a Liouville argument, which also shows that 
$\dot{\mathbf{Q}}_\ind^+(w)$ is unimodular.  Since $\dot{\mathbf{Q}}_\ind^+(w)=\mathbf{F}_\ind^+(w)$
for $w\in\mathbb{C}\setminus U$, the uniform bound on $\dot{\mathbf{Q}}_\ind^+(w)$
will follow from a corresponding bound on $\mathbf{F}_\ind^+(w)$ valid when $w$ is bounded away
from the singularity at $w_*$, or what is the same, a bound on $\mathbf{G}_\ind^+(\sv)$ valid when
$\sv$ is bounded away from both $iS$ and $-iS$ where $S:=(-w_*)^{1/2}>0$.  According to the partial-fraction
representation \eqref{eq:GnplusPF} of $\mathbf{G}_\ind^+(\sv)$, it is therefore enough to establish
that the elements of the vectors $\mathbf{b}^+_\ind$ and $\mathbf{d}^+_\ind$ remain bounded as $\epsilon_N\downarrow 0$ uniformly with respect to $(y,s)\in\Omega_\ind^+$ with $y$ bounded.
But from \eqref{eq:bdplus} we see that the components of $\mathbf{b}^+_\ind$ and $\mathbf{d}^+_\ind$
really only depend on $J_\ind^+$ and $K_\ind^+$ (since $w_*\approx -1$) and while the latter
can become unbounded, it is still obvious that $\mathbf{b}^+_\ind$ and $\mathbf{d}^+_\ind$ cannot.
%
\end{proof}

\subsection{Parametrix for $\mathbf{Q}_\ind(w)$ in $\Omega_\ind^-$.}
Fix an integer $\ind$ with $|\ind|\le B$, and suppose that $(y,s)\in\Omega_\ind^-$ with $y\in\mathbb{R}$ bounded. By introducing the scaled and shifted coordinate $p_{\ind-1}$ defined by \eqref{eq:spind}
and analyzing the jump matrix $\mathbf{V}_{\mathbf{Q}_\ind}(\xi)$ for $\xi\in\partial U$ as in
the beginning of \S\ref{sec:parametrixplus}, we see that in the current situation a lower-triangular model
for the jump matrix is more appropriate.  In particular, defining 
\begin{equation}
h_\ind^-(w):=\eta_\ind(w)^{-2} 
\label{eq:hindminus}
\end{equation}
and then
\begin{equation}
Y^-_\ind(w):=e^{-p_{\ind-1}}h_\ind^-(w)\left(\pv_\ind(y)W(w)^{-1} +
\epsilon_N^{1/3}\left[\pv_\ind'(y)-2H_\ind(y)\pv_\ind(y)\right]W(w)^{-2}\right),
\label{eq:Ynminusdef}
\end{equation}
(note that $Y_\ind^-(\xi)$ extends from $\xi\in\partial U$ to a meromorphic function for $\xi\in\overline{U}$ whose only singularity is a double pole at $w=w_*$) we arrive at the following.
\begin{proposition}
The following three estimates hold uniformly for $\xi\in\partial U$ and $y$ bounded.\\
In the downward-pointing ``teeth'' of $\Omega_\ind^-$:
\begin{multline}
\mathbf{V}_{\mathbf{Q}_\ind}(\xi)\begin{bmatrix}1 & 0\\Y_\ind^-(\xi) & 1\end{bmatrix}^{-1}=
\mathbb{I} + \bo(\epsilon_N^{1/3}e^{-p_{\ind-1}}|y-\mathscr{Z}(\pv_\ind)|),\\ (y,s)\in\Omega_\ind^-,
\quad p_{\ind-1}<0\;\text{and}\;|y-\mathscr{Z}(\pv_\ind)|\ge\delta e^{p_{\ind-1}}.
\end{multline}
In the part of $\Omega_\ind^-$ resembling a ``strip with notches'':
\begin{equation}
\mathbf{V}_{\mathbf{Q}_\ind}(\xi)\begin{bmatrix}1 & 0\\Y_\ind^-(\xi) & 1\end{bmatrix}^{-1}=
\mathbb{I} + \bo\left(\frac{\epsilon_N^{1/3}}{|y-\mathscr{P}(\pv_\ind)|}\right),\quad (y,s)\in\Omega_\ind^-,
\quad |y-\mathscr{P}(\pv_\ind)|\le\delta.
\end{equation}
Elsewhere in $\Omega_\ind^-$:
\begin{equation}
\mathbf{V}_{\mathbf{Q}_\ind}(\xi)\begin{bmatrix}1 & 0\\Y_\ind^-(\xi) & 1\end{bmatrix}^{-1}=
\mathbb{I} + \bo(\epsilon_N^{1/3}).
\end{equation}
\label{prop:VQerrorminus}
\end{proposition}
The proof of this statement is virtually the same as that of Proposition~\ref{prop:VQerrorplus}.  
The analogue for $\Omega_\ind^-$ of Corollary~\ref{cor:VQerrorplus} is then this result.
\begin{corollary}
The following estimate holds uniformly for $\xi\in\partial U$ and $(y,s)\in\Omega_\ind^-$ with $y$
bounded:
\begin{equation}
\mathbf{V}_{\mathbf{Q}_\ind}(\xi)\begin{bmatrix}1 & 0\\Y_\ind^-(\xi) & 1\end{bmatrix}^{-1}=
\mathbb{I}+\bo(\epsilon_N^{1/6}).
\end{equation}
While uniform, this estimate fails to be sharp except possibly at the vertical extremes of $\Omega_\ind^-$ where $|p_{\ind-1}|\sim \frac{1}{3}\log(\epsilon_N^{-1})$.
\label{cor:VQerrorminus}
\end{corollary}


Based on this approximation result for $\mathbf{V}_{\mathbf{Q}_\ind}(\xi)$, we
build a parametrix for $\mathbf{Q}_\ind(w)$ to satisfy the following
problem.
\begin{rhp}[Parametrix for $\mathbf{Q}_\ind(w)$ in $\Omega_\ind^-$]
Let real numbers $y$ and $p_{\ind-1}$ be given, and assume that $y\not\in\mathscr{P}(\pv_\ind)$.  Seek a matrix $\dot{\mathbf{Q}}_\ind^-(w)$
with the following properties:
\begin{itemize}
\item[]\textbf{Analyticity:}  $\dot{\mathbf{Q}}_\ind^-(w)$ is analytic for
$w\in\mathbb{C}\setminus(\mathbb{R}_+\cup\partial U)$, and is uniformly
H\"older continuous up to the boundary of its domain of
analyticity.  
\item[]\textbf{Jump conditions:}  The boundary values taken by 
$\dot{\mathbf{Q}}_\ind^-(w)$ on the contour 
$\mathbb{R}_+\cup\partial U$ are related
as follows:
\begin{equation}
\dot{\mathbf{Q}}^-_{\ind+}(\xi)=\sigma_2
\dot{\mathbf{Q}}^-_{\ind-}(\xi)\sigma_2,\quad \xi>0,
\end{equation}
and
\begin{equation}
\dot{\mathbf{Q}}^-_{\ind+}(\xi)=\dot{\mathbf{Q}}^-_{\ind-}(\xi)
\begin{bmatrix}1 & 0\\
Y_\ind^-(\xi) & 1\end{bmatrix},\quad \xi\in\partial U
\end{equation}
where $\partial U$ is taken to be negatively (clockwise) oriented and
$Y_\ind^-(\xi)$ is defined by \eqref{eq:Ynminusdef}.
\item[]\textbf{Normalization:}  the matrix $\dot{\mathbf{Q}}_\ind^-(w)$ satisfies
the condition
\begin{equation}
\lim_{w\to\infty}\dot{\mathbf{Q}}_\ind^-(w)=\mathbb{I}.
\end{equation}
\end{itemize}
\label{rhp:EparametrixII}
\end{rhp}
This problem is solved completely analogously to Riemann-Hilbert 
Problem~\ref{rhp:Eparametrix}.  Namely, we introduce a matrix 
$\mathbf{G}_\ind^-(\sv)$
by 
\begin{equation}
\mathbf{G}_\ind^-(\sv):=\begin{cases}\mathbf{F}_\ind^-(\sv^2),\quad &\Im\{\sv\}>0\\
\sigma_2\mathbf{F}_\ind^-(\sv^2)\sigma_2,\quad &\Im\{\sv\}<0,
\end{cases}
\end{equation}
where we define $\mathbf{F}_\ind^-(w)$ by
\begin{equation}
\mathbf{F}_\ind^-(w):=\begin{cases}
\dot{\mathbf{Q}}_\ind^-(w),&\quad w\not\in\overline{U}\\
\displaystyle \dot{\mathbf{Q}}_\ind^-(w)\begin{bmatrix}1 & 0\\
Y_\ind^-(w) & 1\end{bmatrix},&\quad
w\in U.
\end{cases}
\label{eq:Fnminusdef}
\end{equation}
Then, $\mathbf{G}_\ind^-(\sv)$ is a rational function with 
$\mathbf{G}_\ind^-(\infty)=\mathbb{I}$ and double poles at $\sv=\pm iS$ where $S:=(-w_*)^{1/2}>0$
only, satisfying the symmetry relation 
$\mathbf{G}_\ind^-(-\sv)=\sigma_2\mathbf{G}_\ind^-(\sv)\sigma_2$.

The singularities of $\mathbf{G}_\ind^-(\sv)$ are characterized by first noting
the Laurent expansion of $Y_\ind^-(\sv^2)$ about $\sv=iS$:
\begin{equation}
Y_\ind^-(\sv^2)=\frac{J_\ind^-}{(\sv-iS)^2} + \frac{iK_\ind^-}{\sv-iS}
+\bo(1)
\end{equation}
where $\bo(1)$ denotes a function analytic at $\sv=iS$ and where
\begin{equation}
\begin{split}
J_\ind^-&:=-\frac{e^{-p_{\ind-1}}h_\ind^-(w_*)\epsilon_N^{1/3}}{4w_*W'(w_*)^2}
\Big(2H_\ind(y)\pv_\ind(y)-\pv_\ind'(y)\Big)\\
K_\ind^-&:=-\frac{e^{-p_{\ind-1}}}{4w_*SW'(w_*)^3}
\Bigg[2w_*h_\ind^-(w_*)W'(w_*)^2\pv_\ind(y)\\
&\quad{}+\epsilon_N^{1/3}\Big(2H_\ind(y)\pv_\ind(y)-\pv_\ind'(y)\Big)
\Big(-2w_*h_\ind^{-\prime}(w_*)W'(w_*)+2w_*h_\ind^-(w_*)W''(w_*)+h_\ind^-(w_*)W'(w_*)
\Big)\Bigg]
\end{split}
\label{eq:JKminus}
\end{equation}
are real-valued coefficients when both $y\in\mathbb{R}$ and 
$p_{\ind-1}\in\mathbb{R}$.  The Laurent series of  $\mathbf{G}_\ind^-(\sv)$
about $\sv=iS$ has the form
\begin{equation}
\mathbf{G}_\ind^-(\sv)=\frac{(\mathbf{a}_\ind^-,\mathbf{b}_\ind^-)}{(\sv-iS)^2} +
\frac{(\mathbf{c}_\ind^-,\mathbf{d}_\ind^-)}{\sv-iS} +
(\mathbf{e}_\ind^-,\mathbf{f}_\ind^-) +(\mathbf{g}_\ind^-,\mathbf{h}_\ind^-)(\sv-iS)
+\bo((\sv-iS)^2),
\end{equation}
where $\bo((\sv-iS)^2)$ represents an analytic function vanishing
to second order at $\sv=iS$.  Since $\dot{\mathbf{Q}}_\ind^-(w)$ is
meant to be analytic at $w=w_*$, from \eqref{eq:Fnminusdef} we learn
that the principal part of the second column of $\mathbf{G}_\ind^-(\sv)$ vanishes:
\begin{equation}
\mathbf{b}_\ind^-=\mathbf{d}_\ind^-=\mathbf{0},
\label{eq:GminusPPsecondcol}
\end{equation}
and the principal part of the first column satisfies the relations
\begin{equation}
\mathbf{a}_\ind^- = J_\ind^-\mathbf{f}_\ind^-\quad\text{and}\quad
\mathbf{c}_\ind^- = J_\ind^-\mathbf{h}_\ind^-+iK_\ind^-\mathbf{f}_\ind^-.
\label{eq:GminusPPfirstcol}
\end{equation}
Using the symmetry condition 
$\mathbf{G}_\ind^-(-\sv)=\sigma_2\mathbf{G}_\ind^-(\sv)\sigma_2$
and taking into account \eqref{eq:GminusPPsecondcol}, we see that the
partial fraction expansion of $\mathbf{G}_\ind^-(\sv)$ necessarily has the form:
\begin{equation}
\mathbf{G}_\ind^-(\sv)=\mathbb{I}+
\frac{(\mathbf{a}_\ind^-,\mathbf{0})}{(\sv-iS)^2} +
\frac{(\mathbf{c}_\ind^-,\mathbf{0})}{\sv-iS} +
\frac{(\mathbf{0},-i\sigma_2\mathbf{a}_\ind^-)}{(\sv+iS)^2} +
\frac{(\mathbf{0},i\sigma_2\mathbf{c}_\ind^-)}{\sv+iS}.
\label{eq:GnminusPF}
\end{equation}
By evaluating the second column of this expression and its derivative at 
$\sv=iS$ we can express $\mathbf{f}_\ind^-$ and $\mathbf{h}_\ind^-$ in terms
of $\mathbf{a}_\ind^-$ and $\mathbf{c}_\ind^-$:
\begin{equation}
\mathbf{f}_\ind^- = \begin{bmatrix}0\\1\end{bmatrix} -
\frac{i\sigma_2}{4w_*}\mathbf{a}_\ind^- +\frac{\sigma_2}{2S}\mathbf{c}_\ind^-
\quad\text{and}\quad
\mathbf{h}_\ind^- = \frac{\sigma_2}{4w_*S}\mathbf{a}_\ind^--\frac{i\sigma_2}
{4w_*}\mathbf{c}_\ind^-,\quad S:=(-w_*)^{1/2}>0.
\end{equation}
Combining this with \eqref{eq:GminusPPfirstcol} yields a closed inhomogeneous
linear system of equations for the components of the vectors $\mathbf{a}_\ind^-$
and $\mathbf{c}_\ind^-$:
\begin{equation}
\begin{split}
\left(\mathbb{I}+\frac{iJ_\ind^-}{4w_*}\sigma_2\right)\mathbf{a}_\ind^- -
\frac{J_\ind^-}{2S}\sigma_2\mathbf{c}_\ind^- &=\displaystyle
\begin{bmatrix}0\\J_\ind^-\end{bmatrix}\\
-\frac{J_\ind^-+SK_\ind^-}{4w_*S}\sigma_2\mathbf{a}_\ind^- +
\left(\mathbb{I}+\frac{iJ_\ind^-+2iSK_\ind^-}{4w_*}\sigma_2\right)\mathbf{c}_\ind^- &\displaystyle = \begin{bmatrix}0 \\ iK_\ind^-\end{bmatrix}.
\end{split}
\end{equation}
As before, this system has a unique solution:
\begin{equation}
\begin{split}
\mathbf{a}_\ind^-&=\frac{4w_*J_\ind^-}{D_\ind^-}\begin{bmatrix}
[J_\ind^-]^3-16w_*^2J_\ind^--32w_*^2SK_\ind^-\\
4w_*(16w_*^2+2SJ_\ind^-K_\ind^-+3[J_\ind^-]^2)
\end{bmatrix}\\
\mathbf{c}_\ind^-&=\frac{4iS}{D_\ind^-}
\begin{bmatrix}
[J_\ind^-]^4+16w_*^2[J_\ind^-]^2+32w_*^2SJ_\ind^-K_\ind^--32w_*^3[K_\ind^-]^2\\
4w_*(2[J_\ind^-]^3+S[J_\ind^-]^2K_\ind^--16w_*^2SK_\ind^-)
\end{bmatrix}
\end{split}
\label{eq:acminus}
\end{equation}
where the denominator
\begin{equation}
D_\ind^-:=(16w_*^2+[J_\ind^-]^2)^2+64w_*^2(J_\ind^-+SK_\ind^-)^2\ge 256w_*^4
\label{eq:denomminus}
\end{equation}
is again strictly positive because $w_*\approx -1$.
This essentially completes the construction of $\dot{\mathbf{Q}}_\ind^-(w)$
solving Riemann-Hilbert Problem~\ref{rhp:EparametrixII}.

The analogue of Proposition~\ref{prop:parametrixplus} in this case is the following.  (Its proof
is also nearly the same.)
\begin{proposition}
  For each fixed integer $\ind$ with $|\ind|\le B$,
  Riemann-Hilbert Problem~\ref{rhp:EparametrixII} has a unique and explicit solution
  $\dot{\mathbf{Q}}_\ind^-(w)$ satisfying
  $\det(\dot{\mathbf{Q}}_\ind^-(w))=1$ where defined.  Also, $\dot{\mathbf{Q}}_\ind^-(w)$
  is uniformly bounded for $w\in\mathbb{C}\setminus (U\cup\mathbb{R}_+)$ and $(y,s)\in\Omega_\ind^-$
  with $y$ bounded as $\epsilon_N\to 0$.
\label{prop:parametrixminus}
\end{proposition}

\subsection{Accuracy of the parametrices for $\mathbf{Q}_\ind(w)$.  Error analysis.}
The \emph{error} in approximating $\mathbf{Q}_\ind(w)$ by the parametrix $\dot{\mathbf{Q}}_\ind^\pm(w)$
is the matrix $\mathbf{E}_\ind^\pm(w):=\mathbf{Q}_\ind(w)\dot{\mathbf{Q}}_\ind^\pm(w)^{-1}$.   Since the
jump contour $\partial U$ for $\dot{\mathbf{Q}}_\ind^\pm(w)$ is a subset of the jump contour $\Sigma$
of $\mathbf{Q}_\ind(w)$, and since $\det(\dot{\mathbf{Q}}_\ind^\pm(w))=1$ for $w\in\mathbb{C}\setminus
\partial U$, it is clear that $\mathbf{E}_\ind^\pm(w)$ is, like $\mathbf{Q}_\ind(w)$, analytic for
$w\in\mathbb{C}\setminus\Sigma$, and takes its boundary values on $\Sigma$ in the same
classical sense as we require of $\mathbf{Q}_\ind(w)$.  Also, $\mathbf{E}_\ind^\pm(w)\to\mathbb{I}$
as $w\to\infty$ since this is true of both factors.

We now consider the jump conditions satisfied by $\mathbf{E}_\ind^\pm(w)$ as a consequence of
those known to be satisfied by $\mathbf{Q}_\ind(w)$ from the statement of Riemann-Hilbert Problem~\ref{rhp:critQ} and those satisfied by $\dot{\mathbf{Q}}_\ind^\pm(w)$ from either Riemann-Hilbert
Problem~\ref{rhp:Eparametrix} or \ref{rhp:EparametrixII}.  If $\xi\in\mathbb{R}_+$, then
\begin{equation}
\begin{split}
\mathbf{E}_{\ind+}^\pm(\xi)&=\mathbf{Q}_{\ind+}(\xi)\dot{\mathbf{Q}}_{\ind+}^\pm(\xi)^{-1}\\
&=
\sigma_2\mathbf{Q}_{\ind-}(\xi)\sigma_2(\mathbb{I}+\mathbf{X}_\ind(\xi))\sigma_2
\dot{\mathbf{Q}}_{\ind-}^\pm(\xi)^{-1}\sigma_2 \\
&= 
\sigma_2\mathbf{E}_{\ind-}^\pm(\xi)\sigma_2\cdot\left[\sigma_2\dot{\mathbf{Q}}_{\ind-}^\pm(\xi)
\sigma_2(\mathbb{I}+\mathbf{X}_\ind(\xi))\sigma_2\dot{\mathbf{Q}}_{\ind-}^\pm(\xi)^{-1}\sigma_2
\right].
\end{split}
\end{equation}
If we write the matrix in brackets above as $\mathbb{I}+\tilde{\mathbf{X}}_\ind^\pm(\xi)$, we see that
$\tilde{\mathbf{X}}_\ind^\pm(\xi)=\mathbf{0}$ wherever $\mathbf{X}_\ind(\xi)=\mathbf{0}$ (that is, everywhere
except in the interval $I$ near $\xi=1$).  Using Proposition~\ref{prop:parametrixplus} and the
fact that $\mathbf{X}_\ind(\xi)=\bo(\epsilon_N)$ for $\xi\in I$, we see that under the condition
that $(y,s)\in\Omega_\ind^+$ with $y$ bounded we also have $\tilde{\mathbf{X}}_\ind^\pm(\xi)=\bo(\epsilon_N)$
for $\xi\in I$.  Similarly, Proposition~\ref{prop:parametrixminus} shows that under the
condition that $(y,s)\in\Omega_\ind^-$ with $y$ bounded the same estimate holds for $\tilde{\mathbf{X}}_\ind^\pm(\xi)$.  If $\xi\in\Sigma\setminus(\mathbb{R}_+\cup\partial U)$, then the parametrix has no discontinuity,
so
\begin{equation}
\mathbf{E}_{\ind+}^\pm(\xi)=\mathbf{Q}_{\ind+}(\xi)\dot{\mathbf{Q}}_\ind^\pm(\xi)^{-1} =
\mathbf{E}_{\ind-}^\pm(\xi)\dot{\mathbf{Q}}_\ind^\pm(\xi)(\mathbb{I}+\bo(\epsilon_N))
\dot{\mathbf{Q}}_\ind^\pm(\xi)^{-1}
\end{equation}
Therefore according to Proposition~\ref{prop:parametrixplus}, if $(y,s)\in\Omega_\ind^+$ with
$y$ bounded, we will have $\mathbf{E}_{\ind+}^+(\xi)=\mathbf{E}_{\ind-}^+(\xi)(\mathbb{I}+\bo(\epsilon_N))$
for such $\xi$.  Analogously, Proposition~\ref{prop:parametrixminus} shows that $(y,s)\in\Omega_\ind^-$
with $y$ bounded implies that $\mathbf{E}_{\ind+}^-(\xi)=\mathbf{E}_{\ind-}^-(\xi)(\mathbb{I}+\bo(\epsilon_N))$ for such $\xi$.  Finally, if $\xi\in\partial U$, then 
\begin{equation}
\begin{split}
\mathbf{E}_{\ind+}^+(\xi)&=\mathbf{Q}_{\ind+}(\xi)\dot{\mathbf{Q}}_{\ind+}^+(\xi)^{-1}\\
&=\mathbf{Q}_{\ind-}(\xi)\mathbf{V}_{\mathbf{Q}_\ind}(\xi)\begin{bmatrix}1 & Y_\ind^+(\xi)\\
0 & 1\end{bmatrix}^{-1}\dot{\mathbf{Q}}_{\ind-}^+(\xi)^{-1}\\
&=\mathbf{E}_{\ind-}^+(\xi)\left(
\dot{\mathbf{Q}}_{\ind-}^+(\xi)\mathbf{V}_{\mathbf{Q}_\ind}(\xi)\begin{bmatrix}1 & Y_\ind^+(\xi)\\
0 & 1\end{bmatrix}^{-1}\dot{\mathbf{Q}}_{\ind-}^+(\xi)^{-1}\right).
\end{split}
\end{equation}
If $(y,s)\in\Omega_\ind^+$ with $y$ bounded, then according to Propositions~\ref{prop:VQerrorplus}
and \ref{prop:parametrixplus} we find $\mathbf{E}_{\ind+}^+(\xi)=\mathbf{E}_{\ind-}^+(\xi)(\mathbb{I}+
\lo(1))$ where the size of the error term is $\bo(\epsilon_N^{1/6})$ at the worst (according to Corollary~\ref{cor:VQerrorplus}) but is more typically smaller according to the more precise estimates \eqref{eq:plusestimateteeth}--\eqref{eq:plusestimateelsewhere} enumerated in
Proposition~\ref{prop:VQerrorplus}.  Similarly, for $\xi\in\partial U$,
\begin{equation}
\begin{split}
\mathbf{E}_{\ind+}^-(\xi)&=\mathbf{Q}_{\ind+}(\xi)\dot{\mathbf{Q}}_{\ind+}^-(\xi)^{-1}\\
&=\mathbf{Q}_{\ind-}(\xi)\mathbf{V}_{\mathbf{Q}_\ind}(\xi)\begin{bmatrix}1 & 0\\
Y_\ind^-(\xi) & 1\end{bmatrix}^{-1}\dot{\mathbf{Q}}_{\ind-}^-(\xi)^{-1}\\
&=\mathbf{E}_{\ind-}^-(\xi)\left(
\dot{\mathbf{Q}}_{\ind-}^-(\xi)\mathbf{V}_{\mathbf{Q}_\ind}(\xi)\begin{bmatrix}1 &0\\
Y_\ind^-(\xi) & 1\end{bmatrix}^{-1}\dot{\mathbf{Q}}_{\ind-}^-(\xi)^{-1}\right),
\end{split}
\end{equation}
so if $(y,s)\in\Omega_\ind^-$ with $y$ bounded, then from
Propositions~\ref{prop:VQerrorminus} and \ref{prop:parametrixminus} we find
$\mathbf{E}_{\ind+}^-(\xi)=\mathbf{E}_{\ind-}^-(\xi)(\mathbb{I}+
\lo(1))$ with a similar characterization of the error term.  

The dominant terms in the deviation of the jump matrix from the identity come from $\xi\in\partial U$.  Therefore, we observe that the matrix $\mathbf{E}_\ind^\pm$ satisfies the conditions of the following
type of Riemann-Hilbert problem.
\begin{rhp}[Small-Norm Problem for the Error]
Let $\ind$ be an integer with $|\ind|\le B$, and suppose that $(y,s)\in\Omega_\ind^\pm$ with $y$ bounded.   Seek a matrix $\mathbf{E}_\ind^\pm(w)$
with the following properties:
\begin{itemize}
\item[]\textbf{Analyticity:}  $\mathbf{E}_\ind^\pm(w)$ is analytic for
$w\in\mathbb{C}\setminus\Sigma$, and is uniformly
H\"older continuous up to the boundary of 
its domain of
analyticity.  
\item[]\textbf{Jump conditions:}  The boundary values taken by 
$\mathbf{E}_\ind^\pm(w)$ on the contour 
$\Sigma$ are related
as follows:
\begin{equation}
\mathbf{E}^\pm_{\ind+}(\xi)=\sigma_2
\mathbf{E}^\pm_{\ind-}(\xi)\sigma_2\mathbf{V}_{\mathbf{E}_\ind^\pm}(\xi),\quad \xi>0,
\end{equation}
and
\begin{equation}
\mathbf{E}^\pm_{\ind+}(\xi)=\mathbf{E}^\pm_{\ind-}(\xi)\mathbf{V}_{\mathbf{E}_\ind^\pm}(\xi),
\quad \xi\in\Sigma\setminus\mathbb{R}_+,
\end{equation}
where the following estimates hold for the jump matrix.  Firstly, we have
$\mathbf{V}_{\mathbf{E}_\ind^\pm}(\xi)=\mathbb{I}$ for $\xi>0$ with $\log(|\xi|)$ sufficiently large, so
$\mathbf{V}_{\mathbf{E}_\ind^\pm}-\mathbb{I}$ is compactly supported and vanishes identically in a neighborhood
of $\xi=0$.  Next, we have $\|\mathbf{V}_{\mathbf{E}_\ind^\pm}-\mathbb{I}\|_{L^\infty(\Sigma)} = \bo(\epsilon_N^{1/6})$, or, more precisely, $\|\mathbf{V}_{\mathbf{E}_\ind^\pm}-\mathbb{I}\|_{L^\infty(\Sigma)} = \bo(e_\ind^\pm(y,s;\epsilon_N))$, where for $(y,s)\in\Omega_\ind^+$,
\begin{equation}
e_\ind^+(y,s;\epsilon_N):= 
\begin{cases}\epsilon_N^{1/3}e^{p_\ind}|y-\mathscr{Z}(\pu_\ind)|,&\quad
p_\ind>0\;\text{and}\;|y-\mathscr{Z}(\pu_\ind)|\ge\delta e^{-p_\ind}\\
\epsilon_N^{1/3}|y-\mathscr{P}(\pu_\ind)|^{-1},&\quad
|y-\mathscr{P}(\pu_\ind)|\le\delta\\
\epsilon_N^{1/3},&\quad\text{elsewhere in $\Omega_\ind^+$},
\end{cases}
\label{eq:VEestimateplus}
\end{equation}
and for $(y,s)\in\Omega_\ind^-$,
\begin{equation}
e_\ind^-(y,s;\epsilon_N):= 
\begin{cases}\epsilon_N^{1/3}e^{-p_{\ind-1}}|y-\mathscr{Z}(\pv_\ind)|,&\quad
p_{\ind-1}<0\;\text{and}\;|y-\mathscr{Z}(\pv_\ind)|\ge\delta e^{p_{\ind-1}}\\
\epsilon_N^{1/3}|y-\mathscr{P}(\pv_\ind)|^{-1},&\quad
|y-\mathscr{P}(\pv_\ind)|\le\delta\\
\epsilon_N^{1/3},&\quad\text{elsewhere in $\Omega_\ind^-$}.
\end{cases}
\label{eq:VEestimateminus}
\end{equation}
\item[]\textbf{Normalization:}  the matrix $\mathbf{E}_\ind^\pm(w)$ satisfies
the condition
\begin{equation}
\lim_{w\to\infty}\mathbf{E}_\ind^\pm(w)=\mathbb{I}.
\end{equation}
\end{itemize}
\label{rhp:Error}
\end{rhp}

\begin{proposition}
Suppose that $\ind$ is an integer with $|\ind|\le B$, and that $(y,s)\in\Omega_\ind^\pm$ with
$y$ bounded.  Then for sufficiently small $\epsilon_N$, there exists a unique solution $\mathbf{E}_\ind^\pm(w)$
of Riemann-Hilbert Problem~\ref{rhp:Error}.  The solution has expansions for small and large $w$
of the form
\begin{equation}
\mathbf{E}^\pm_\ind(w)=\begin{cases}
[\fourIdx{0}{0}{\pm}{\ind}{\mathbf{E}}]+ [\fourIdx{0}{1}{\pm}{\ind}{\mathbf{E}}](-w)^{1/2} + \bo(w),&\quad
w\to 0\\
\mathbb{I}+[\fourIdx{\infty}{1}{\pm}{\ind}{\mathbf{E}}](-w)^{-1/2} + \bo(w^{-1}),&\quad w\to\infty,
\end{cases}
\label{eq:Eexpand}
\end{equation}
and we have the estimates
\begin{equation}
\vphantom{o} [\fourIdx{0}{0}{\pm}{\ind}{\mathbf{E}}]=\mathbb{I}+
 \bo(e_\ind^\pm(y,s;\epsilon_N)),\quad
[\fourIdx{0}{1}{\pm}{\ind}{\mathbf{E}}]=\bo(e_\ind^\pm(y,s;\epsilon_N)),\quad\text{and}\quad
[\fourIdx{\infty}{1}{\pm}{\ind}{\mathbf{E}}]=\bo(e_\ind^\pm(y,s;\epsilon_N))
\end{equation}
holding uniformly for $(y,s)\in\Omega_\ind^\pm$ with $y$ bounded, 
where $e_\ind^+(y,s;\epsilon_N)$ and $e_\ind^-(y,s;\epsilon_N)$ are defined by
\eqref{eq:VEestimateplus} and \eqref{eq:VEestimateminus} respectively.
\label{prop:Error}
\end{proposition}

\begin{proof}
By the substitution $\mathbf{F}_\ind^\pm(\sv):=\mathbf{E}_\ind^\pm(\sv^2)$ for $\Im\{\sv\}>0$ and 
$\mathbf{F}_\ind^\pm(\sv):=\sigma_2\mathbf{E}_\ind^\pm(\sv^2)\sigma_2$ for $\Im\{\sv\}<0$, we translate
the conditions of Riemann-Hilbert Problem~\ref{rhp:Error} into a list of conditions satisfied
by $\mathbf{F}_\ind^\pm(\sv)$.  This ``unfolding'' of $w$ to the $\sv$-plane implies that $\mathbf{F}_\ind^\pm(\sv)$ is analytic on the complement of a compact contour $\Sigma'$, along which it satisfies jump conditions of the form $\mathbf{F}_{\ind+}^\pm=\mathbf{F}_{\ind-}^\pm\mathbf{V}_{\mathbf{F}_\ind^\pm}$ where
 $\mathbf{V}_{\mathbf{F}_\ind^\pm}-\mathbb{I}$ satisfies the uniform estimate $\|\mathbf{V}_{\mathbf{F}_\ind^\pm}-\mathbb{I}\|_{L^\infty(\Sigma')} = \bo(e_\ind^\pm(y,s;\epsilon_N))$ and $e_\ind^\pm(y,s;\epsilon_N)$ is defined by \eqref{eq:VEestimateplus} for $(y,s)\in\Omega_\ind^+$
and by \eqref{eq:VEestimateminus} for $(y,s)\in\Omega_\ind^-$.  Since these estimates
are all dominated by $\bo(\epsilon_N^{1/6})$, for sufficiently small $\epsilon_N$, the small-norm
theory of matrix Riemann-Hilbert problems in, say, the $L^2$-sense implies that $\mathbf{F}_\ind^\pm(\sv)$
may be constructed by iteration applied to a suitable singular
integral equation, and it follows from that theory that an estimate of the form
$\|\mathbf{F}_\ind^\pm-\mathbb{I}\|_{L^\infty(K)} = \bo(e_\ind^\pm(y,s;\epsilon_N))$ actually holds whenever $K$ is a compact set disjoint from the jump contour $\Sigma'$.  
In particular, $\mathbf{F}_\ind^\pm(\sv)$ is analytic on some compact neighborhood $K$ of
the origin $\sv=0$, so it has a Taylor expansion that we write in the form
\begin{equation}
\mathbf{F}_\ind^\pm(\sv)=[\fourIdx{0}{0}{\pm}{\ind}{\mathbf{F}}]+ [\fourIdx{0}{1}{\pm}{\ind}{\mathbf{F}}]\sv + \bo(\sv^2),\quad \sv\to 0,
\end{equation}
and by the uniform estimate on $\mathbf{F}_\ind^\pm(\sv)-\mathbb{I}$ we have $[\fourIdx{0}{0}{\pm}{\ind}{\mathbf{E}}]=[\fourIdx{0}{0}{\pm}{\ind}{\mathbf{F}}]=\mathbb{I} +\bo(\|\mathbf{F}_\ind^\pm-\mathbb{I}\|_{L^\infty(K)})$.  Applying the Cauchy integral formula to $\mathbf{F}_\ind^\pm(\sv)-\mathbb{I}$ on a small loop in $K$ surrounding the origin 
then shows that  $\vphantom{o}[\fourIdx{0}{1}{\pm}{\ind}{\mathbf{E}}]=i[\fourIdx{0}{1}{\pm}{\ind}{\mathbf{F}}]=
\bo(\|\mathbf{F}_\ind^\pm-\mathbb{I}\|_{L^\infty(K)})$.  Also, $\mathbf{F}_\ind^\pm(\sv)$ tends to the identity as $\sv\to\infty$
and it is analytic for sufficiently large $|\sv|$ so it has a Laurent expansion of the form
\begin{equation}
\mathbf{F}_\ind^\pm(\sv)=\mathbb{I}+[\fourIdx{\infty}{1}{\pm}{\ind}{\mathbf{F}}]\sv^{-1}+\bo(\sv^{-2}),\quad \sv\to\infty,
\end{equation}
Now we can take $K$ as an annulus enclosing a large circle.  Applying the Cauchy integral
formula on such a circle to $\mathbf{F}_\ind^\pm(\sv)-\mathbb{I}$ we find that $\vphantom{o}[\fourIdx{\infty}{1}{\pm}{\ind}{\mathbf{E}}]=-i[\fourIdx{\infty}{1}{\pm}{\ind}{\mathbf{F}}]=\bo(\|\mathbf{F}_\ind^\pm-\mathbb{I}\|_{L^\infty(K)})$.  This completes the proof.
\end{proof}

The magnitude of the expansion coefficients of $\mathbf{E}_\ind^+(w)\approx\mathbb{I}$ and of $\mathbf{E}_{\ind+1}^-(w)\approx\mathbb{I}$ in the interlocking regions $\Omega_\ind^+$ and $\Omega_{\ind+1}^-$ respectively is illustrated in Figure~\ref{fig:EminusIEstimates}.  In particular it is clear that the order of accuracy is $\bo(\epsilon_N^{1/3})$
except very close to the common boundaries between teeth and notches.  Moreover, the order of accuracy is worst right along these common boundaries where it is $\bo(\epsilon_N^{1/3}e^{|p_\ind|/2})$ which in turn is largest just at the tips of the
teeth ($|p_\ind|=\frac{1}{3}\log(\epsilon_N^{-1})$), where the estimate reduces to $\bo(\epsilon_N^{1/6})$.
\begin{figure}[h]
\begin{center}
\includegraphics{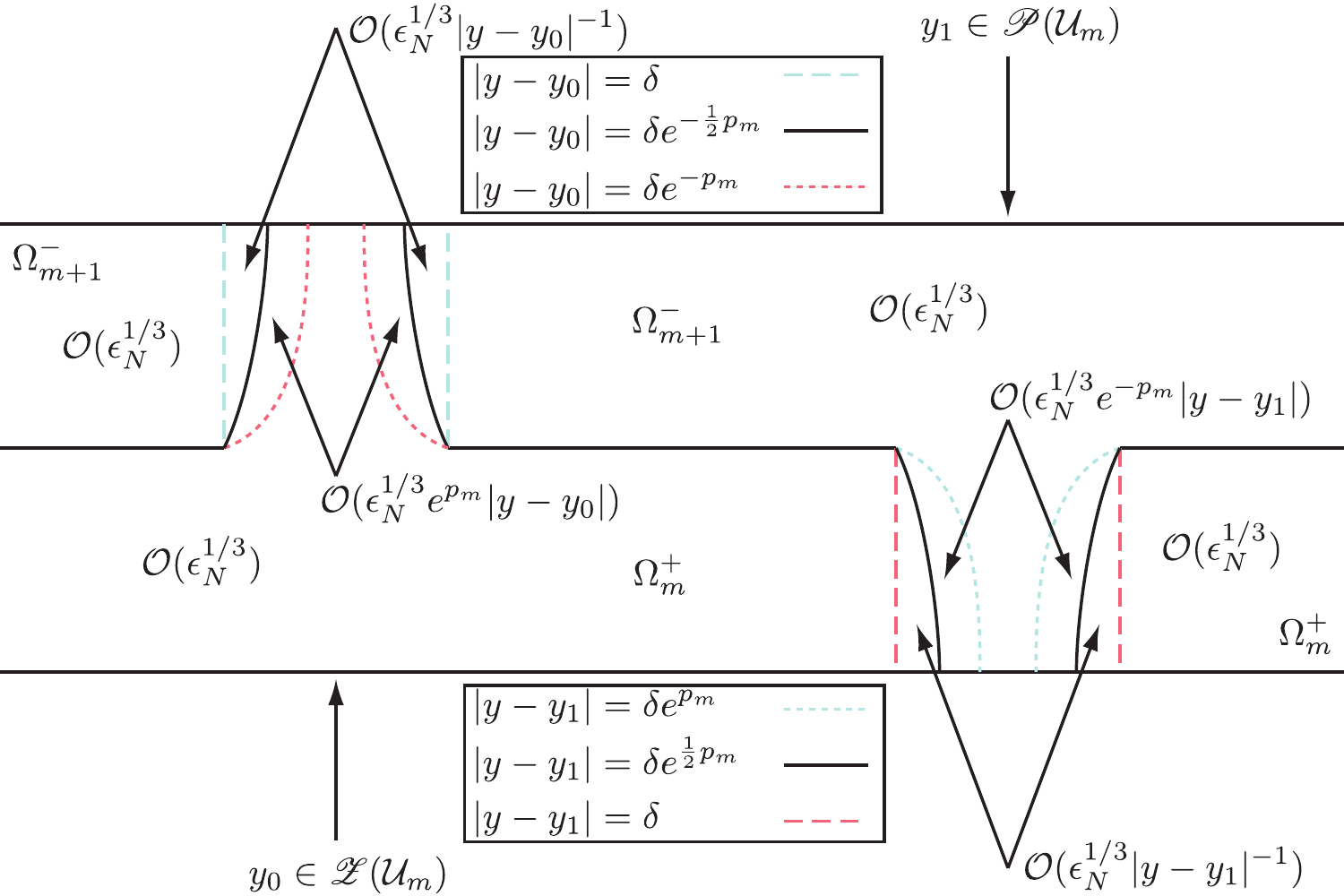}
\end{center}
\caption{\emph{Estimates $e_\ind^+(y,s;\epsilon_N)$ for expansion coefficients of $\mathbf{E}_\ind^+(w)-\mathbb{I}$ with $(y,s)\in\Omega_\ind^+$ and $e_{\ind+1}^-(y,s;\epsilon_N)$ for
expansion coefficients of $\mathbf{E}_{\ind+1}^-(w)-\mathbb{I}$ with $(y,s)\in\Omega_{\ind+1}^-$.  The boundaries of the two interlocking regions are drawn with solid black curves.}}
\label{fig:EminusIEstimates}
\end{figure}

\section{Proofs of Theorems}
\label{sec:potentials}
\subsection{Exact formulae for the potentials.}
Beginning with the matrix $\mathbf{N}(w)=\mathbf{N}^\mathfrak{g}_N(w;x,t)$ we may now express $\mathbf{N}(w)$ in terms of explicitly known quantities and the error matrix $\mathbf{E}_\ind^\pm(w)$ characterized in Proposition~\ref{prop:Error}.  For example, if $|w|$ is sufficiently small then for $\Im\{w\}>0$ we have
\begin{equation}
\mathbf{N}(w)=\mathbf{O}(w)=
\mathbf{P}(w)i\sigma_1=\mathbf{Q}_\ind(w)\dot{\mathbf{P}}^\mathrm{out}_{\ind}(w)i\sigma_1=\mathbf{E}_\ind^\pm(w)\dot{\mathbf{Q}}_\ind^\pm(w)
\dot{\mathbf{P}}^\mathrm{out}_{m}(w)i\sigma_1,
\end{equation}
while for $\Im\{w\}<0$ we have the same formula just replacing $\mathbf{P}(w)$ and $\dot{\mathbf{P}}^\mathrm{out}_\ind(w)$ by $-\mathbf{P}(w)$ and $-\dot{\mathbf{P}}^\mathrm{out}_\ind(w)$ respectively.  
But from \eqref{eq:PoutDS} we see that
\begin{equation}
\mathrm{sgn}(\Im\{w\})\dot{\mathbf{P}}_\ind^\mathrm{out}(w)=\left(\mathbb{I}
+\frac{1-2\ind}{S}\sigma_3 (-w)^{1/2} + \bo(w)\right)(-1)^{\ind+1}i\sigma_2,\quad w\to 0,
\end{equation}
where $S:=(-w_*)^{1/2}>0$.  Therefore, in terms of the expansion \eqref{eq:Eexpand} of $\mathbf{E}_\ind^\pm(w)$ for
small $|w|$, and the corresponding expansion of $\dot{\mathbf{Q}}_\ind^\pm(w)$:
\begin{equation}
\dot{\mathbf{Q}}_\ind^\pm(w)=[\fourIdx{0}{0}{\pm}{\ind}{\smash{\dot{\mathbf{Q}}}}] + 
[\fourIdx{0}{1}{\pm}{\ind}{\smash{\dot{\mathbf{Q}}}}](-w)^{1/2} + \bo(w),\quad w\to 0,
\label{eq:Qwsmall}
\end{equation}
we find the following formulae for certain expansion coefficients of $\mathbf{N}(w)$
(see \eqref{eq:Nexpandwsmall}):
\begin{equation}
[\fourIdx{0}{0}{\mathfrak{g}}{N}{\mathbf{N}}](x,t)=[\fourIdx{0}{0}{\pm}{\ind}{\mathbf{E}}][\fourIdx{0}{0}{\pm}{\ind}{\smash{\dot{\mathbf{Q}}}}](-1)^{\ind+1}i\sigma_2,
\label{eq:N00}
\end{equation}
and
\begin{equation}
\vphantom{\fourIdx{0}{0}{\mathfrak{g}}{N}{\mathbf{N}}(x,t)^{-1}}[\fourIdx{0}{0}{\mathfrak{g}}{N}{\mathbf{N}}](x,t)^{-1}[\fourIdx{0}{1}{\mathfrak{g}}{N}{\mathbf{N}}](x,t)=
\sigma_2[\fourIdx{0}{0}{\pm}{\ind}{\smash{\dot{\mathbf{Q}}}}]^{-1}
[\fourIdx{0}{0}{\pm}{\ind}{\mathbf{E}}]^{-1}
[\fourIdx{0}{1}{\pm}{\ind}{\mathbf{E}}]
[\fourIdx{0}{0}{\pm}{\ind}{\smash{\dot{\mathbf{Q}}}}]\sigma_2 +
\sigma_2[\fourIdx{0}{0}{\pm}{\ind}{\smash{\dot{\mathbf{Q}}}}]^{-1}
[\fourIdx{0}{1}{\pm}{\ind}{\smash{\dot{\mathbf{Q}}}}]\sigma_2 -
\frac{1-2m}{S}\sigma_3.
\label{eq:N01}
\end{equation}

Similarly, if $|w|$ is sufficiently large, then we have
\begin{equation}
\mathbf{N}(w)=\mathbf{O}(w)=\mathbf{P}(w)=\mathbf{Q}_\ind(w)\dot{\mathbf{P}}^\mathrm{out}_\ind(w) = \mathbf{E}_\ind^\pm(w)\dot{\mathbf{Q}}_\ind^\pm(w)\dot{\mathbf{P}}_\ind^\mathrm{out}(w),
\end{equation}
and from \eqref{eq:PoutDS} we have
\begin{equation}
\dot{\mathbf{P}}^\mathrm{out}_\ind(w)=\mathbb{I}+(1-2\ind)S\sigma_3(-w)^{-1/2}+\bo(w^{-1}),\quad w\to\infty,
\end{equation}
so in terms of the expansion \eqref{eq:Eexpand} of $\mathbf{E}_\ind^\pm(w)$ for large $|w|$,
and the corresponding expansion of $\dot{\mathbf{Q}}_\ind^\pm(w)$:
\begin{equation}
\dot{\mathbf{Q}}_\ind^\pm(w)=\mathbb{I} + [\fourIdx{\infty}{1}{\pm}{\ind}{\smash{\dot{\mathbf{Q}}}}](-w)^{-1/2} + \bo(w^{-1}),\quad w\to\infty,
\label{eq:Qwlarge}
\end{equation}
we obtain an exact formula for the coefficient $[\fourIdx{\infty}{1}{\mathfrak{g}}{N}{\mathbf{N}}](x,t)$ (see \eqref{eq:Nexpandwlarge}):
\begin{equation}
[\fourIdx{\infty}{1}{\mathfrak{g}}{N}{\mathbf{N}}](x,t)=[\fourIdx{\infty}{1}{\pm}{\ind}{\mathbf{E}}] +
[\fourIdx{\infty}{1}{\pm}{\ind}{\smash{\dot{\mathbf{Q}}}}] +(1-2\ind)S\sigma_3.
\label{eq:Ninfty1}
\end{equation}

Combining \eqref{eq:N00} with \eqref{eq:fluxoncondensate} from Definition~\ref{def:fluxoncondensate} we
obtain the exact formulae:
\begin{equation}
\cos(\tfrac{1}{2}u_N(x,t))=(-1)^\ind\left[[\fourIdx{0}{0}{\pm}{\ind}{\mathbf{E}}][\fourIdx{0}{0}{\pm}{\ind}{\smash{\dot{\mathbf{Q}}}}]\right]_{12}\quad\text{and}\quad
\sin(\tfrac{1}{2}u_N(x,t))=(-1)^\ind\left[[\fourIdx{0}{0}{\pm}{\ind}{\mathbf{E}}][\fourIdx{0}{0}{\pm}{\ind}{\smash{\dot{\mathbf{Q}}}}]\right]_{22}.
\end{equation}
Similarly, combining \eqref{eq:N01} and \eqref{eq:Ninfty1} with \eqref{eq:epsuNt} gives
\begin{equation}
\epsilon_N\frac{\partial u_N}{\partial t}(x,t)=[\fourIdx{\infty}{1}{\pm}{\ind}{\mathbf{E}}]_{12} +
[\fourIdx{\infty}{1}{\pm}{\ind}{\smash{\dot{\mathbf{Q}}}}]_{12} -
\left[[\fourIdx{0}{0}{\pm}{\ind}{\smash{\dot{\mathbf{Q}}}}]^{-1}
[\fourIdx{0}{0}{\pm}{\ind}{\mathbf{E}}]^{-1}
[\fourIdx{0}{1}{\pm}{\ind}{\mathbf{E}}]
[\fourIdx{0}{0}{\pm}{\ind}{\smash{\dot{\mathbf{Q}}}}]\right]_{21}-
\left[[\fourIdx{0}{0}{\pm}{\ind}{\smash{\dot{\mathbf{Q}}}}]^{-1}
[\fourIdx{0}{1}{\pm}{\ind}{\smash{\dot{\mathbf{Q}}}}]\right]_{21}.
\end{equation}

Now suppose that $(x,t)$ are such that $(y,s)=(\epsilon^{-2/3}r(x,t),s(x,t))\in\Omega_\ind^\pm$
with $y$ bounded.  Then we can apply Proposition~\ref{prop:Error} and the boundedness
of $[\fourIdx{0}{0}{\pm}{\ind}{\smash{\dot{\mathbf{Q}}}}]=\dot{\mathbf{Q}}_\ind^\pm(0)$ and
its inverse as guaranteed by Propositions~\ref{prop:parametrixplus} and \ref{prop:parametrixminus} to obtain:
\begin{equation}
\begin{split}
\cos(\tfrac{1}{2}u_N(x,t))&=\dot{C}_{N,\ind}^\pm(x,t) +\bo(e_\ind^\pm(y,s;\epsilon_N))\\
\sin(\tfrac{1}{2}u_N(x,t))&=\dot{S}_{N,\ind}^\pm(x,t)+\bo(e_\ind^\pm(y,s;\epsilon_N))\\
\epsilon_N\frac{\partial u_N}{\partial t}(x,t)&=\dot{G}_{N,\ind}^\pm(x,t)+\bo(e_\ind^\pm(y,s;\epsilon_N)),
\end{split}
\label{eq:cossineE}
\end{equation}
all holding for $(y,s)\in\Omega_\ind^\pm$ with $y$ bounded, where
\begin{equation}
\begin{split}
\dot{C}_{N,\ind}^\pm(x,t)&:=(-1)^\ind[\fourIdx{0}{0}{\pm}{\ind}{\smash{\dot{\mathbf{Q}}}}]_{12}\\
\dot{S}_{N,\ind}^\pm(x,t)&:=(-1)^\ind[\fourIdx{0}{0}{\pm}{\ind}{\smash{\dot{\mathbf{Q}}}}]_{22}\\
\dot{G}_{N,\ind}^\pm(x,t)&:=[\fourIdx{\infty}{1}{\pm}{\ind}{\smash{\dot{\mathbf{Q}}}}]_{12} -
\left[[\fourIdx{0}{0}{\pm}{\ind}{\smash{\dot{\mathbf{Q}}}}]^{-1}
[\fourIdx{0}{1}{\pm}{\ind}{\smash{\dot{\mathbf{Q}}}}]\right]_{21} = 
[\fourIdx{\infty}{1}{\pm}{\ind}{\smash{\dot{\mathbf{Q}}}}]_{12}+
[\fourIdx{0}{0}{\pm}{\ind}{\smash{\dot{\mathbf{Q}}}}]_{21}
[\fourIdx{0}{1}{\pm}{\ind}{\smash{\dot{\mathbf{Q}}}}]_{11} - 
[\fourIdx{0}{0}{\pm}{\ind}{\smash{\dot{\mathbf{Q}}}}]_{11}
[\fourIdx{0}{1}{\pm}{\ind}{\smash{\dot{\mathbf{Q}}}}]_{21}.
\end{split}
\end{equation}
(In the last equality we used the fact that $[\fourIdx{0}{0}{\pm}{\ind}{\smash{\dot{\mathbf{Q}}}}]=
\dot{\mathbf{Q}}_\ind^\pm(0)$ has determinant one.)
Since $\dot{\mathbf{Q}}_\ind^\pm(w)=\mathbf{F}_\ind^\pm(w)=\mathbf{G}_\ind^\pm(i(-w)^{1/2})$, 
the necessary expansion coefficients of $\dot{\mathbf{Q}}_\ind^\pm(w)$ (see \eqref{eq:Qwsmall} and \eqref{eq:Qwlarge}) can be obtained directly from the partial-fractions expansions \eqref{eq:GnplusPF} and \eqref{eq:GnminusPF}.  In this way, we get
\begin{equation}
\begin{split}
[\fourIdx{0}{0}{+}{\ind}{\smash{\dot{\mathbf{Q}}}}]&=\mathbb{I}+\frac{1}{w_*}\left(i\sigma_2[\mathbf{b}_\ind^+-iS\mathbf{d}_\ind^+],
\mathbf{b}_\ind^+-iS\mathbf{d}_\ind^+\right)\\
[\fourIdx{0}{1}{+}{\ind}{\smash{\dot{\mathbf{Q}}}}]&=\frac{1}{w_*S}
\left(-i\sigma_2[2\mathbf{b}_\ind^+-iS\mathbf{d}_\ind^+],
2\mathbf{b}_\ind^+-iS\mathbf{d}_\ind^+\right)\\
[\fourIdx{\infty}{1}{+}{\ind}{\smash{\dot{\mathbf{Q}}}}]&=-i\left(-i\sigma_2\mathbf{d}_\ind^+,\mathbf{d}_\ind^+\right),
\end{split}
\end{equation}
implying that
\begin{equation}
\begin{split}
\dot{C}_{N,\ind}^+(x,t)&=(-1)^\ind\left[\frac{b_{\ind,1}^+}{w_*}+\frac{id_{\ind,1}^+}{S}\right]\\
\dot{S}_{N,\ind}^+(x,t)&=(-1)^\ind\left[1+\frac{b_{\ind,2}^+}{w_*}+\frac{id_{\ind,2}^+}{S}\right]\\
\dot{G}_{N,\ind}^+(x,t)&=\frac{1}{w_*^2}\left[iw_*(1-w_*)d_{\ind,1}^++2Sb_{\ind,1}^++ib_{\ind,1}^+d_{\ind,2}^+-ib_{\ind,2}^+d_{\ind,1}^+\right],
\end{split}
\end{equation}
where the components of the vectors $\mathbf{b}_\ind^+$ and $\mathbf{d}_\ind^+$ are given by 
\eqref{eq:bdplus} with \eqref{eq:denomplus} and we recall that $S:=(-w_*)^{1/2}>0$.   In the same way one sees that
\begin{equation}
\begin{split}
[\fourIdx{0}{0}{-}{\ind}{\smash{\dot{\mathbf{Q}}}}]&=\mathbb{I}+\frac{1}{w_*}
\left(\mathbf{a}_\ind^--iS\mathbf{c}_\ind^-,-i\sigma_2[\mathbf{a}_\ind^--iS\mathbf{c}_\ind^-]\right)\\
[\fourIdx{0}{1}{-}{\ind}{\smash{\dot{\mathbf{Q}}}}]&=\frac{1}{w_*S}
\left(2\mathbf{a}_\ind^--iS\mathbf{c}_\ind^-,i\sigma_2[2\mathbf{a}_\ind^--iS\mathbf{c}_\ind^-]\right)\\
[\fourIdx{\infty}{1}{-}{\ind}{\smash{\dot{\mathbf{Q}}}}]&=-i\left(\mathbf{c}_\ind^-,i\sigma_2\mathbf{c}_\ind^-\right),
\end{split}
\end{equation}
implying that
\begin{equation}
\begin{split}
\dot{C}_{N,\ind}^-(x,t)&=(-1)^{\ind+1}\left[\frac{a_{\ind,2}^-}{w_*} +\frac{ic_{\ind,2}^-}{S}\right]\\
\dot{S}_{N,\ind}^-(x,t)&=(-1)^\ind\left[1+\frac{a_{\ind,1}^-}{w_*}+\frac{ic_{\ind,1}^-}{S}\right]\\
\dot{G}_{N,\ind}^-(x,t)&=\frac{1}{w_*^2}\left[iw_*(1-w_*)c_{\ind,2}^-+2Sa_{\ind,2}^- + ia_{\ind,2}^-c_{\ind,1}^--ia_{\ind,1}^-c_{\ind,2}^-\right],
\end{split}
\end{equation}
where $\mathbf{a}_\ind^-$ and $\mathbf{c}_\ind^-$ are given by \eqref{eq:acminus} with \eqref{eq:denomminus} and we have $S:=(-w_*)^{1/2}>0$.

Finally, substituting from \eqref{eq:bdplus}--\eqref{eq:denomplus} and \eqref{eq:acminus}--\eqref{eq:denomminus} gives
\begin{equation}
\begin{split}
\dot{C}_{N,\ind}^\pm(x,t)&=\pm (-1)^{m+1}\frac{2R_\ind^\pm}{[R_\ind^\pm]^2+1}\\
\dot{S}_{N,\ind}^\pm(x,t)&=(-1)^{m+1}\frac{[R_\ind^\pm]^2-1}{[R_\ind^\pm]^2+1}\\
\dot{G}_{N,\ind}^\pm(x,t)&=16w_*\frac{[J_\ind^\pm]^2K_\ind^\pm+2S[J_\ind^\pm]^3-w_*[J_\ind^\pm]^2K_\ind^\pm+32w_*SJ_\ind^\pm-16w_*^2K_\ind^\pm+16w_*^3K_\ind^\pm}{([J_\ind^\pm]^2+16w_*^2)^2+64w_*^2(J_\ind^\pm+SK_\ind^\pm)^2},
\end{split}
\label{eq:dotCSGplusminus}
\end{equation}
where
\begin{equation}
R_\ind^\pm:=\frac{-8w_*(J_\ind^\pm+SK_\ind^\pm)}{[J_\ind^\pm]^2+16w_*^2},
\label{eq:Rplusminus}
\end{equation}
and where $J_\ind^+$ and $K_\ind^+$ are defined by \eqref{eq:JKplus} while 
$J_\ind^-$ and $K_\ind^-$ are defined by \eqref{eq:JKminus}.

\subsection{Proof of Theorem~\ref{thm:tzero} (Initial accuracy of fluxon condensates).}
The fluxon condensate $u_N(x,t)$ is an even function of $x$, so it suffices to consider $x\approx x_\mathrm{crit}>0$.  The line segment $t=0$ with $\Delta x=\bo(\epsilon_N^{2/3})$ is mapped under $(y,s)=(\epsilon^{-2/3}r(x,t),s(x,t))$
into the union $\Omega_0^+\cup\Omega_1^-$,
assuming that  $\epsilon_N$ sufficiently small.  More specifically, it is mapped into the part of $\Omega_0^+\cup\Omega_1^-$ with $y=\bo(1)$ and $p_0=\bo(\epsilon_N^{1/3})$.  Using the facts that $\pu_0(y)=\pv_1(y)=1$ and $\pv_0(y)=\pu_1(y)=-y/6$, and hence also $H_0(y)=H_1(y)=y^2/24$, it is easy to see from the definitions \eqref{eq:JKplus} and
\eqref{eq:JKminus} (using also $h_0^+(w_*)=h_1^-(w_*)=-4w_*W'(w_*)$) that for $t=0$ and $\Delta x=\bo(\epsilon_N^{2/3})$,
\begin{equation}
J_0^+=\bo(\epsilon_N^{1/3})\quad\text{and}\quad J_1^-=\bo(\epsilon_N^{1/3})
\end{equation}
while
\begin{equation}
K_0^+=-2S + \bo(\epsilon_N^{1/3})\quad\text{and}\quad
K_1^-=-2S+\bo(\epsilon_N^{1/3}),
\end{equation}
where $S:=(-w_*)^{1/2}>0$.
According to \eqref{eq:Rplusminus} we then have
\begin{equation}
R_0^+=-1+\bo(\epsilon_N^{1/3})\quad\text{and}\quad R_1^-=-1+\bo(\epsilon_N^{1/3}).
\end{equation}
It  follows from these facts and \eqref{eq:dotCSGplusminus} that for such $x$ and $t=0$,
\begin{equation}
\dot{C}_{N,0}^+(x,0)=1+\bo(\epsilon_N^{1/3})\quad\text{and}\quad
\dot{C}_{N,1}^-(x,0)=1+\bo(\epsilon_N^{1/3}),
\end{equation}
\begin{equation}
\dot{S}_{N,0}^+(x,0)=\bo(\epsilon_N^{1/3})\quad\text{and}\quad
\dot{S}_{N,1}^-(x,0)=\bo(\epsilon_N^{1/3}),
\end{equation}
and
\begin{equation}
\dot{G}_{N,0}^+(x,0)=-\left(S+\frac{1}{S}\right)+\bo(\epsilon_N^{1/3})\quad\text{and}\quad
\dot{G}_{N,1}^-(x,0)=-\left(S+\frac{1}{S}\right)+\bo(\epsilon_N^{1/3}).
\end{equation}
Furthermore, since $w_*$ is an analytic function of $(x,t)$ with $w_*(x_\mathrm{crit},0)=-1$,
we have $w_*=-1 + \bo(\Delta x,t) = -1+\bo(\Delta x) = -1+\bo(\epsilon_N^{2/3})$.  Therefore,
we may equivalently write the above formulae in the form
\begin{equation}
\dot{G}_{N,0}^+(x,0)=-2 + \bo(\epsilon_N^{1/3})\quad\text{and}\quad
\dot{G}_{N,1}^-(x,0)=-2+\bo(\epsilon_N^{1/3}),\quad \Delta x=\bo(\epsilon_N^{2/3}).
\end{equation}
Also, since the function $G(\cdot)$ in the initial data is differentiable at $x_\mathrm{crit}$
and $G(x_\mathrm{crit})=-2$, we can also write
\begin{equation}
\dot{G}_{N,0}^+(x,0)=G(x)+\bo(\epsilon_N^{1/3})\quad\text{and}\quad
\dot{G}_{N,1}^-(x,0)=G(x)+\bo(\epsilon_N^{1/3}),\quad \Delta x = \bo(\epsilon_N^{2/3}).
\end{equation}
Finally, recalling \eqref{eq:cossineE},
 the fact that $e_0^+(y,s;\epsilon_N)=\bo(\epsilon_N^{1/3})$ for $(y,s)\in\Omega_0^+$
with $y$ bounded and $p_0=\bo(\epsilon_N^{1/3})$, and the fact that
$e_1^-(y,s;\epsilon_N)=\bo(\epsilon_N^{1/3})$ for $(y,s)\in\Omega_1^-$ with $y$ bounded
and $p_0=\bo(\epsilon_N^{1/3})$, we arrive at
\begin{equation}
u_N(x,0)=\bo(\epsilon_N^{1/3})\pmod{4\pi}\quad\text{and}\quad\epsilon_N\frac{\partial u_N}{\partial t}(x,0)=G(x)+\bo(\epsilon_N^{1/3}),\quad \Delta x=\bo(\epsilon_N^{2/3}),
\end{equation}
which completes the proof of the Theorem.

\subsection{Proof of Theorem~\ref{thm:main} (Main approximation theorem).}
\label{sec:ProofMain}
Firstly, we wish to simplify the expressions for $\dot{C}_{N,\ind}^\pm(x,t)$ and $\dot{S}_{N,\ind}^\pm(x,t)$ given in \eqref{eq:dotCSGplusminus} by localizing near the critical point $\Delta x=t=0$, which
allows us to write the coordinates $r(x,t)$ and $s(x,t)$ explicitly in terms of $x$ and $t$ at the
cost of some small error terms according to \eqref{eq:rsexpansions}.  
By analogy with the definition \eqref{eq:spind} of $p_\ind$ for general $\ind\in\mathbb{Z}$, we define a shifted and scaled coordinate $q_\ind$ equivalent to $t$ by the relation
\begin{equation}
t = \frac{2\ind}{3}\epsilon_N\log(\epsilon_N^{-1}) + \epsilon_Nq_\ind,\quad\ind\in\mathbb{Z}.
\label{eq:sqind}
\end{equation}
Recall the spatial coordinate $z$ exactly proportional to $\Delta x$ as defined 
by \eqref{eq:zDeltax}.
We will now show that there is essentially no additional cost (beyond that already introduced in approximating
$\mathbf{E}_\ind^\pm(w)$) in explicitly introducing $\Delta x$ and $t$ as coordinates.
\begin{proposition}
Suppose that $(y,s)\in\Omega_\ind^+$ with $y$ bounded and $\ind$ fixed.  Then
\begin{equation}
\dot{C}_{N,\ind}^+(x,t)=\ddot{C}_{N,\ind}^+(x,t)+
\bo(e_\ind^+(y,s;\epsilon_N))
\quad\text{and}\quad
\dot{S}_{N,\ind}^+(x,t)=\ddot{S}_{N,\ind}^+(x,t)+
\bo(e_\ind^+(y,s;\epsilon_N))
\end{equation}
where
\begin{equation}
\ddot{C}_{N,\ind}^+(x,t):=(-1)^{m+1}\frac{2\dot{R}_\ind^+}{[\dot{R}_\ind^+]^2+1}\quad\text{and}\quad
\ddot{S}_{N,\ind}^+(x,t):=(-1)^{m+1}\frac{[\dot{R}_\ind^+]^2-1}{[\dot{R}_\ind^+]^2+1}
\label{eq:CSddotplus}
\end{equation}
and where
\begin{equation}
\dot{R}_\ind^+:=\frac{-2^{4(1-\ind)}\nu^{-2\ind/3}e^{q_\ind}\pu_\ind(z)}{2^{-8\ind}\nu^{-2(1+2\ind)/3}
\epsilon_N^{2/3}e^{2q_\ind}(2H_\ind(z)\pu_\ind(z)-\pu_\ind'(z))^2 + 16}.
\label{eq:dotRplus}
\end{equation}
Similarly, suppose that $(y,s)\in\Omega_\ind^-$ with $y$ bounded and $\ind$ fixed.  Then
\begin{equation}
\dot{C}_{N,\ind}^-(x,t)=\ddot{C}_{N,\ind}^-(x,t)+
\bo(e_\ind^-(y,s;\epsilon_N))
\quad\text{and}\quad
\dot{S}_{N,\ind}^-(x,t)=\ddot{S}_{N,\ind}^-(x,t)+
\bo(e_\ind^-(y,s;\epsilon_N))
\end{equation}
where
\begin{equation}
\ddot{C}_{N,\ind}^-(x,t):=(-1)^{m}\frac{2\dot{R}_\ind^-}{[\dot{R}_\ind^-]^2+1}\quad\text{and}\quad
\ddot{S}_{N,\ind}^-(x,t):=(-1)^{m+1}\frac{[\dot{R}_\ind^-]^2-1}{[\dot{R}_\ind^-]^2+1}
\label{eq:CSddotminus}
\end{equation}
and where
\begin{equation}
\dot{R}_\ind^-:=\frac{-2^{4\ind}\nu^{2(\ind-1)/3}e^{-q_{\ind-1}}\pv_\ind(z)}
{2^{8(\ind-1)}\nu^{-2+4\ind/3}\epsilon_N^{2/3}e^{-2q_{\ind-1}}(2H_\ind(z)\pv_\ind(z)-\pv_\ind'(z))^2 + 16}.
\label{eq:dotRminus}
\end{equation}
\label{prop:CSapproximate}
\end{proposition}
\begin{proof}
First, we localize near criticality by isolating the terms in $R_\ind^\pm$ depending ``slowly'' on
$x$ and $t$ and replacing $x$ with $x_\mathrm{crit}$ and $t$ with zero in these terms.
The ``slowly varying'' quantities are $w_*=w_*(x,t)$ and the following:
\begin{equation}
\begin{split}
\alpha_\ind^\pm&:=\frac{h_\ind^\pm(w_*)}{4w_*W'(w_*)^2}\\
\beta_\ind^\pm&:=\frac{4w_*h_\ind^\pm(w_*)}{W'(w_*)}\\
\gamma_\ind^\pm&:=\frac{4w_*h_\ind^{\pm\prime}(w_*)W'(w_*)-4w_*h_\ind^\pm(w_*)W''(w_*)-4h_\ind^\pm(w_*)W'(w_*)}{W'(w_*)^3}.
\end{split}
\end{equation}
Recalling the definitions of  $J_\ind^\pm$ and $K_\ind^\pm$ given by \eqref{eq:JKplus} and \eqref{eq:JKminus},
one obtains from \eqref{eq:Rplusminus} the exact expressions
\begin{equation}
\begin{split}
R_\ind^+&=\frac{\beta_\ind^+e^{p_\ind}\pu_\ind(y)+\gamma_\ind^+\epsilon_N^{1/3}e^{p_\ind}B_{\ind,12}(y)}{[\alpha_\ind^+\epsilon_N^{1/3}e^{p_\ind}B_{\ind,12}(y)]^2+16w_*^2},\quad\text{where}\quad B_{\ind,12}(y)=2H_\ind(y)\pu_\ind(y)-\pu_\ind'(y)\\
R_\ind^-&=\frac{\beta_\ind^-e^{-p_{\ind-1}}\pv_\ind(y)+\gamma_\ind^-\epsilon_N^{1/3}e^{-p_{\ind-1}}B_{\ind,21}(y)}{[\alpha_\ind^-\epsilon_N^{1/3}e^{-p_{\ind-1}}B_{\ind,21}(y)]^2+16w_*^2},\quad\text{where}\quad
B_{\ind,21}(y)=\pv_\ind'(y)-2H_\ind(y)\pv_\ind(y).
\end{split}
\end{equation}
Now, for each $\ind\in\mathbb{Z}$, $\alpha_\ind^\pm$, $\beta_\ind^\pm$, $\gamma_\ind^\pm$, and $w_*$ are all independent of $\epsilon_N$, depending only on $(x,t)$.  Moreover, they are all analytic functions
of $(x,t)$ near criticality.  Computing their limits at criticality using \eqref{eq:WprimeWdoubleprimecrit}
and \eqref{eq:Hnformula} together with \eqref{eq:hindplus} and \eqref{eq:hindminus} we obtain:
\begin{equation}
\begin{split}
\alpha_\ind^\pm &= -2^{-2\pm(2-4\ind)}\nu^{(-2\pm(1-2\ind))/3} + \bo(\Delta x,t)\\
\beta_\ind^\pm &=-2^{2\pm(2-4\ind)}\nu^{(-1\pm(1-2\ind))/3} + \bo(\Delta x,t)\\
\gamma_\ind^\pm &=\bo(\Delta x,t),
\end{split}
\end{equation} 
and for $w_*$ we have the corresponding expression \eqref{eq:wstarexpand}.  
Since by assumption $\ind$ and $y$ are bounded, it follows immediately that $s(x,t)=\bo(\epsilon_N\log(\epsilon_N^{-1}))$
and that $r(x,t)=\bo(\epsilon_N^{2/3})$.  Then, since the Jacobian matrix $\partial(r,s)/\partial (x,t)$ is
diagonal and invertible at criticality according to \eqref{eq:rsexpansions}, it also follows that $\Delta x=\bo(\epsilon_N^{2/3})$ and $t=\bo(\epsilon_N\log(\epsilon_N^{-1}))$, and consequently all error terms
of the form $\bo(\Delta x,t)$ may be replaced with $\bo(\epsilon_N^{2/3})$.  Therefore,
\begin{equation}
\begin{split}
R_\ind^+&=\frac{-2^{4(1-\ind)}\nu^{-2\ind/3} e^{p_\ind}\pu_\ind(y)(1+\bo(\epsilon_N^{2/3})) +
\epsilon_N^{1/3}e^{p_\ind}B_{\ind,12}(y)\bo(\epsilon_N^{2/3})}
{[-2^{-4\ind}\nu^{-(1+2\ind)/3}\epsilon_N^{1/3}e^{p_\ind}B_{\ind,12}(y)(1+\bo(\epsilon_N^{2/3}))]^2 + 16(1+\bo(\epsilon_N^{2/3}))}\\
R_\ind^-&=\frac{-2^{4\ind}\nu^{2(m-1)/3}e^{-p_{\ind-1}}\pv_\ind(y)(1+\bo(\epsilon_N^{2/3})) +\epsilon_N^{1/3}e^{-p_{\ind-1}}
B_{\ind,21}(y)\bo(\epsilon_N^{2/3})}
{[-2^{4(\ind-1)}\nu^{-1+2\ind/3}\epsilon_N^{1/3}e^{-p_{\ind-1}}B_{\ind,21}(y)(1+\bo(\epsilon_N^{2/3}))]^2+16(1+\bo(\epsilon_N^{2/3}))}.
\end{split}
\label{eq:Rpmind_intermediate}
\end{equation}
The denominators are bounded away from zero and, according to Proposition~\ref{prop:Rnlocgeneral},
$B_{\ind,12}(y)$ and $B_{\ind,21}(y)$ have simple poles at the points of $\mathscr{P}(\pu_\ind)=\mathscr{P}(\pv_\ind)$ implying that
\begin{equation}
\begin{split}
\epsilon_N^{1/3}e^{p_\ind}B_{\ind,12}(y)&=\bo(1),\quad\text{for $(y,s)\in\Omega_\ind^+$}
\\
\epsilon_N^{1/3}e^{-p_{\ind-1}}B_{\ind,21}(y)&=\bo(1),\quad\text{for $(y,s)\in\Omega_\ind^-$}.
\end{split}
\label{eq:B1221bound}
\end{equation}
Therefore, the formulae \eqref{eq:Rpmind_intermediate} can be written in the form
\begin{equation}
\begin{split}
R_\ind^+&=\frac{-2^{4(1-\ind)}\nu^{-2\ind/3}e^{p_\ind}\pu_\ind(y)}
{2^{-8\ind}\nu^{-2(1+2\ind)/3}\epsilon_N^{2/3}e^{2p_\ind}B_{\ind,12}(y)^2 + 16}
+\bo(\epsilon_N^{2/3}e^{p_\ind}\pu_\ind(y)) + \bo(\epsilon_N^{2/3})\\
R_\ind^-&=\frac{-2^{4\ind}\nu^{2(\ind-1)/3}e^{-p_{\ind-1}}\pv_\ind(y)}
{2^{8(\ind-1)}\nu^{-2+4\ind/3}\epsilon_N^{2/3}e^{-2p_{\ind-1}}B_{\ind,21}(y)^2 + 16}
+\bo(\epsilon_N^{2/3}e^{-p_{\ind-1}}\pv_\ind(y)) + \bo(\epsilon_N^{2/3}).
\end{split}
\end{equation}

The explicit terms we have retained in $R_\ind^\pm$ have the same form as $\dot{R}_\ind^\pm$
 defined by \eqref{eq:dotRplus} and \eqref{eq:dotRminus} but written in terms of the ``fast'' 
 variables $p_\ind$, $p_{\ind-1}$, and $y$ instead of $q_\ind$, $q_{\ind-1}$, and $z$. But, 
recalling the definition \eqref{eq:sqind} we deduce from
\eqref{eq:rsexpansions} that
\begin{equation}
p_\ind = q_\ind + \bo(\epsilon_N^{-1}\Delta x^2,\epsilon_N^{-1}t\Delta x,\epsilon_N^{-1}t^2) = 
q_\ind+\bo(\epsilon_N^{1/3}),
\end{equation}
and from this it follows (using \eqref{eq:B1221bound} and the fact that the denominators are bounded away from zero) that
\begin{equation}
\begin{split}
R_\ind^+&=\frac{-2^{4(1-\ind)}\nu^{-2\ind/3}e^{q_\ind}\pu_\ind(y)}
{2^{-8\ind}\nu^{-2(1+2\ind)/3}\epsilon_N^{2/3}e^{2q_\ind}B_{\ind,12}(y)^2 + 16}
+\bo(\epsilon_N^{1/3}e^{p_\ind}\pu_\ind(y)) + \bo(\epsilon_N^{2/3})\\
R_\ind^-&=\frac{-2^{4\ind}\nu^{2(\ind-1)/3}e^{-q_{\ind-1}}\pv_\ind(y)}
{2^{8(\ind-1)}\nu^{-2+4\ind/3}\epsilon_N^{2/3}e^{-2q_{\ind-1}}B_{\ind,21}(y)^2 + 16}
+\bo(\epsilon_N^{1/3}e^{-p_{\ind-1}}\pv_\ind(y)) + \bo(\epsilon_N^{2/3}).
\end{split}
\end{equation}
Also, from the definition \eqref{eq:zDeltax}, the definition $y=r(x,t)/\epsilon_N^{2/3}$ and 
\eqref{eq:rsexpansions}, we see that
\begin{equation}
y=z + \bo(\epsilon_N^{-2/3}\Delta x^2,\epsilon_N^{-2/3}t\Delta x,\epsilon_N^{-2/3}t^2)=
z + \bo(\epsilon_N^{2/3}).
\end{equation}
Since the poles are simple for $\pu_\ind$ and $\pv_\ind$, 
\begin{equation}
\begin{split}
\pu_\ind(y)-\pu_\ind(z)&=\bo(\epsilon_N^{2/3}\pu_\ind'(y+\bo(\epsilon_N^{2/3}))) =\bo(\epsilon_N^{2/3}|y+\bo(\epsilon_N^{2/3})-\mathscr{P}(\pu_\ind)|^{-2})\\
\pv_\ind(y)-\pv_\ind(z)&=\bo(\epsilon_N^{2/3}\pv_\ind'(y+\bo(\epsilon_N^{2/3}))) = 
\bo(\epsilon_N^{2/3}|y+\bo(\epsilon_N^{2/3})-\mathscr{P}(\pv_\ind)|^{-2}).
\end{split}
\end{equation}
Now, $|y-\mathscr{P}(\pu_\ind)|^{-1}=\bo(\epsilon_N^{-1/6})$ for $(y,s)\in\Omega_\ind^+\cup\Omega_\ind^-$, so
\begin{equation}
\begin{split}
\frac{1}{|y+\bo(\epsilon_N^{2/3})-\mathscr{P}(\pu_\ind)|}&=
\frac{1}{|y-\mathscr{P}(\pu_\ind)|}\frac{|y-\mathscr{P}(\pu_\ind)|}{|y+\bo(\epsilon_N^{2/3})-\mathscr{P}(\pu_\ind)|}\\
&\le
\frac{1}{|y-\mathscr{P}(\pu_\ind)|}\frac{|y-\mathscr{P}(\pu_\ind)|}
{||y-\mathscr{P}(\pu_\ind)| - |\bo(\epsilon_N^{2/3})||} \\ &= 
\frac{1}{|y-\mathscr{P}(\pu_\ind)|}\frac{1}{1-|\bo(\epsilon_N^{2/3})||y-\mathscr{P}(\pu_\ind)|^{-1}}
\\ &= \frac{1}{|y-\mathscr{P}(\pu_\ind)|}\frac{1}{1-|\bo(\epsilon_N^{1/2})|}
\end{split}
\end{equation}
so that $\pu_\ind(y)-\pu_\ind(z)=\bo(\epsilon_N^{2/3}|y-\mathscr{P}(\pu_\ind)|^{-2})$
and similarly $\pv_\ind(y)-\pv_\ind(z)=\bo(\epsilon_N^{2/3}|y-\mathscr{P}(\pv_\ind)|^{-2})$.
For 
$B^2_{\ind,12}$ and $B^2_{\ind,21}$ the poles are double, and it follows by the same reasoning that
$B_{\ind,12}(y)^2-B_{\ind,12}(z)^2=\bo(\epsilon_N^{2/3}|y-\mathscr{P}(\pu_\ind)|^{-3})$
and $B_{\ind,21}(y)^2-B_{\ind,21}(z)^2=\bo(\epsilon_N^{2/3}|y-\mathscr{P}(\pv_\ind)|^{-3})$.
From these considerations we have that
\begin{equation}
\begin{split}
R_\ind^+&=\dot{R}_\ind^+ + \bo(\epsilon_N^{2/3}e^{p_\ind}|y-\mathscr{P}(\pu_\ind)|^{-2})
+\bo(\epsilon_N^{4/3}e^{3p_\ind}\pu_\ind(y)|y-\mathscr{P}(\pu_\ind)|^{-3})\\
&\quad\quad{}
+\bo(\epsilon_N^2e^{3p_\ind}|y-\mathscr{P}(\pu_\ind)|^{-5}) + 
\bo(\epsilon_N^{1/3}e^{p_\ind}\pu_\ind(y)) + \bo(\epsilon_N^{2/3}),\quad (y,s)\in\Omega_\ind^+\\
R_\ind^-&=\dot{R}_\ind^- +
\bo(\epsilon_N^{2/3}e^{-p_{\ind-1}}|y-\mathscr{P}(\pv_\ind)|^{-2}) +
\bo(\epsilon_N^{4/3}e^{-3p_{\ind-1}}\pv_\ind(y)|y-\mathscr{P}(\pv_\ind)|^{-3})\\
&\quad\quad{} +
\bo(\epsilon_N^2e^{-3p_{\ind-1}}|y-\mathscr{P}(\pv_\ind)|^{-5}) +
\bo(\epsilon_N^{1/3}e^{-p_{\ind-1}}\pv_\ind(y))+\bo(\epsilon_N^{2/3}),\quad (y,s)\in\Omega_\ind^-,
\end{split}
\end{equation}
where $\dot{R}_\ind^+$ is defined by \eqref{eq:dotRplus} and $\dot{R}_\ind^-$ is defined
by \eqref{eq:dotRminus}.  By comparing with the definition \eqref{eq:VEestimateplus} of $e_\ind^+(y,s;\epsilon_N)$ for $(y,s)\in\Omega_\ind^+$ and the definition \eqref{eq:VEestimateminus} of $e_\ind^-(y,s;\epsilon_N)$ for $(y,s)\in\Omega_\ind^-$ we see that
\begin{equation}
R_\ind^\pm = \dot{R}_\ind^\pm + \bo(e_\ind^\pm(y,s;\epsilon_N)),\quad (y,s)\in\Omega_\ind^\pm.  
\label{eq:RRdotestimate}
\end{equation}

Finally, we insert the estimate \eqref{eq:RRdotestimate} into $\dot{C}_{N,\ind}^\pm(x,t)$
and $\dot{S}_{N,\ind}^\pm(x,t)$ given by \eqref{eq:dotCSGplusminus} noting that derivatives of the rational expressions in $R_\ind^\pm$ that
appear 
are uniformly bounded to complete the
proof of the proposition.
\end{proof}
Combining Proposition~\ref{prop:CSapproximate} with the formulae \eqref{eq:cossineE}, we  obtain the estimates
\begin{equation}
\cos(\tfrac{1}{2}u_N(x,t))=\ddot{C}_{N,\ind}^\pm(x,t)+\bo(e_\ind^\pm(y,s;\epsilon_N))\quad
\text{and}\quad
\sin(\tfrac{1}{2}u_N(x,t))=\ddot{S}_{N,\ind}^\pm(x,t)+\bo(e_\ind^\pm(y,s;\epsilon_N))
\label{eq:CSapproximateII}
\end{equation}
holding uniformly for $(y,s)\in\Omega_\ind^\pm$ with $y$ bounded.

We are now in a position to define the multiscale asymptotic formulae for $\cos(\tfrac{1}{2}u_N(x,t))$ and $\sin(\tfrac{1}{2}u_N(x,t))$
mentioned in the statement of the Theorem:  
\begin{equation}
\dot{C}(x,t;\epsilon_N):=
\ddot{C}^\pm_{N,\ind}(x,t)\quad\text{and}\quad \dot{S}(x,t;\epsilon_N):=\ddot{S}^\pm_{N,\ind}(x,t),\quad\text{whenever $(\epsilon_N^{-2/3}r(x,t),s(x,t))\in\Omega_\ind^\pm$}.
\end{equation}
Using \eqref{eq:CSapproximateII} we find that
\begin{equation}
\begin{split}
\cos(\tfrac{1}{2}u_N(x,t)) &= \dot{C}(x,t;\epsilon_N)  + \bo(e^\pm_\ind(\epsilon_N^{-2/3}r(x,t),s(x,t);\epsilon_N))=\dot{C}(x,t;\epsilon_N) + \bo(\epsilon_N^{1/6})\\
\sin(\tfrac{1}{2}u_N(x,t))&=\dot{S}(x,t;\epsilon_N)+\bo(e^\pm_\ind(\epsilon_N^{-2/3}r(x,t),s(x,t);\epsilon_N))=\dot{S}(x,t;\epsilon_N)+\bo(\epsilon_N^{1/6})
\end{split}
\end{equation}
hold whenever $(x,t)$ are such that $(\epsilon_N^{-2/3}r(x,t),s(x,t))\in\Omega_\ind^\pm$.  The exponent
of $1/6$ is the worst-case exponent as explained in the paragraph following the proof of Proposition~\ref{prop:Error}.  On the other hand, it is clear that this error estimate is uniformly valid whenever $m$
is bounded and $\epsilon_N^{-2/3}r(x,t)=y$ is bounded, or as is the same, whenever $t=\bo(\epsilon_N\log(\epsilon_N^{-1}))$ and $\Delta x=\bo(\epsilon_N^{2/3})$.  This completes the proof of Theorem~\ref{thm:main}.

\subsection{Proof of Theorem~\ref{thm:kinkapprox} (Superluminal kink asymptotics).}
Choose any integer $m$ with $|m|\le B$, and let $K>0$ be fixed.  Suppose that $(x,t)$ lies in
the rectangle $\mathcal{R}_\ind$ defined by the inequalities $|\Delta x| \le K\epsilon_N^{2/3}$ and
$|t-\tfrac{2}{3}\ind\epsilon_N\log(\epsilon_N^{-1})|\le\tfrac{1}{3}\epsilon_N\log(\epsilon_N^{-1})$. 
Because $y=z+\bo(\epsilon_N^{2/3})$ and $p_\ind=q_\ind+\bo(\epsilon_N^{1/3})$, for any \emph{strictly positive} choice of the parameter $\kappa\ge 0$ entering into Definition~\ref{def:sets}, the image $I(\mathcal{R}_\ind)$ of $\mathcal{R}_\ind$ in the $(y,s)$-plane under the mapping $y=r(x,t)/\epsilon_N^{2/3}$ and $s=s(x,t)$ will be contained in the horizontal strip
$\Omega_\ind^+\cup\Omega_{\ind+1}^-$ (see \eqref{eq:stripunion}) with $y$ bounded, as long as $\epsilon_N$ is sufficiently small.  Also, for $(x,t)$ mapping into $I(\mathcal{R}_\ind)\cap\Omega_\ind^+$ (respectively into $I(\mathcal{R}_\ind)\cap\Omega_{\ind+1}^-$) we have the uniform estimate $e_\ind^+(\epsilon_N^{-2/3}r(x,t),s(x,t);\epsilon_N)=\bo(\epsilon_N^{1/6})$ (respectively we have the uniform estimate $e_{\ind+1}^-(\epsilon_N^{-2/3}r(x,t),s(x,t);\epsilon_N)=\bo(\epsilon_N^{1/6})$).

Moreover, given an interval $[z_-,z_+]$ with $2\nu^{1/3}|z_\pm|\le K$ on which $\log|\pu_\ind(z)|$ is bounded, if the parameter $\delta>0$ entering into Definition~\ref{def:sets} is chosen sufficiently small, the sub-rectangle $\mathcal{R}^{\mathrm{sub}}_\ind\subset\mathcal{R}_\ind$ defined by the inequalities
$z_-\le z\le z_+$ (recall the definition \eqref{eq:zDeltax} of $z$ in terms of $\Delta x$) and $|t-\tfrac{2}{3}\ind\epsilon_N\log(\epsilon_N^{-1})|\le\tfrac{1}{3}\epsilon_N\log(\epsilon_N^{-1})$
will be mapped into a subregion $I(\mathcal{R}_\ind^{\mathrm{sub}})$ of $I(\mathcal{R}_\ind)$ 
in which both inequalities $|y-\mathscr{P}(\pu_\ind)|\ge\delta$ and $|y-\mathscr{P}(\pv_{\ind+1})|\ge\delta$
hold.  It  then follows that
for $(x,t)$ mapping into $I(\mathcal{R}_\ind^{\mathrm{sub}})\cap\Omega_\ind^+$, the estimate
$e_\ind^+(\epsilon_N^{-2/3}r(x,t),s(x,t);\epsilon_N)=\bo(\epsilon_N^{1/3})$ holds uniformly while for $(x,t)$ mapping into
$I(\mathcal{R}_\ind^{\mathrm{sub}})\cap\Omega_{\ind+1}^-$, the estimate $e_{\ind+1}^-(\epsilon_N^{-2/3}r(x,t),s(x,t);\epsilon_N)=\bo(\epsilon_N^{1/3})$ holds uniformly.  See also Figure~\ref{fig:EminusIEstimates}.

Suppose that $(x,t)$ corresponds to a point $(y,s)\in\Omega_\ind^+$.  We wish to neglect the
first term in the  denominator of $\dot{R}_\ind^+$ as defined by \eqref{eq:dotRplus};
we write $\dot{R}_\ind^+$
in the form
\begin{equation}
\dot{R}_\ind^+ = \frac{-\mathrm{sgn}(\pu_\ind(z))e^{T_\kink}}{f^2+1},\quad f=\bo(\epsilon_N^{1/3}e^{q_\ind}B_{\ind,12}(z)) = \bo(\epsilon_N^{1/3}e^{q_\ind}|z-\mathscr{P}(\pu_\ind)|^{-1}),
\end{equation}
where $T_\kink$ is defined in terms of $t$, $\epsilon_N$, $\nu$, and $|\pu_\ind(z)|$ by 
\eqref{eq:kinkTdefine}.  Note that since $z=y+\bo(\epsilon_N^{2/3})$ and $q_\ind=p_\ind+\bo(\epsilon_N^{1/3})$, and since $|y-\mathscr{P}(\pu_\ind)|\ge\delta\epsilon_N^{1/6}$ for $(y,s)\in\Omega_\ind^+$, the upper bound for $f$ may be replaced with $f= 
\bo(\epsilon_N^{1/3}e^{p_\ind}|y-\mathscr{P}(\pu_\ind)|^{-1})$. Therefore, as pointed out in the proof of Proposition~\ref{prop:CSapproximate},
$f=\bo(1)$ uniformly for $(y,s)\in\Omega_\ind^+$, although it is typically smaller pointwise.
By using this result in \eqref{eq:CSddotplus} and recalling \eqref{eq:CSapproximateII}, it follows that whenever $(x,t)\in\mathcal{R}_\ind$ maps into $\Omega_\ind^+$,
\begin{equation}
\begin{split}
\cos(\tfrac{1}{2}u_N(x,t))&=(-1)^m\,\mathrm{sgn}(\pu_\ind(z))\,\mathrm{sech}(T_\kink) + 
\bo(e^{-|T_\kink|}f^2)+
\bo(e_\ind^+(\epsilon_N^{-2/3}r(x,t),s(x,t);\epsilon_N))\\
\sin(\tfrac{1}{2}u_N(x,t))&=(-1)^{m+1}\tanh(T_\kink) + 
\bo(e^{-2|T_\kink|}f^2)+
\bo(e_\ind^+(\epsilon_N^{-2/3}r(x,t),s(x,t);\epsilon_N)).
\end{split}
\end{equation}
As shown at the beginning of the proof, the error terms involving $e_\ind^+$ are either of magnitude $\bo(\epsilon_N^{1/6})$ and in particular $\lo(1)$, or, if $(x,t)\in\mathcal{R}_\ind^{\mathrm{sub}}$, of smaller magnitude $\bo(\epsilon_N^{1/3})$.  The remaining error terms are bounded, and 
they are negligible unless \emph{both} $T_\kink=\bo(1)$ and $f$ fails to be small.  But since 
$f=\bo(\epsilon_N^{1/3}e^{p_\ind}|y-\mathscr{P}(\pu_\ind)|^{-1})$ and $(y,s)\in\Omega_\ind^+$,
we see from Definition~\ref{def:sets} that $f$ fails to be small only if $e^{-p_\ind}=\bo(\epsilon_N^{1/3})$, or equivalently, $e^{-q_\ind}=\epsilon_N^{-2\ind/3}e^{-t/\epsilon_N}=\bo(\epsilon_N^{1/3})$.  Now, under this condition we have from \eqref{eq:kinkTdefine} that
$T_\kink\ge \bo(1) + \tfrac{1}{3}\log(\epsilon_N^{-1})+\log|\pu_\ind(z)|$, so to have $T_\kink=\bo(1)$ requires
that $|z-\mathscr{Z}(\pu_\ind)|=\bo(\epsilon_N^{1/3})$.
Moreover, if $(x,t)\in\mathcal{R}_\ind^{\mathrm{sub}}$ and its image $(y,s)=(\epsilon_N^{-2/3}r(x,t),s(x,t))$ lies in $\Omega_\ind^+$, 
$e^{-|T_\kink|} = \bo(e^{-|q_\ind|})$ and $f = \bo(\epsilon_N^{1/3}e^{q_\ind})$.  The inequality $|q_\ind|\le\tfrac{1}{3}\log(\epsilon_N^{-1})$ then implies that if $(x,t)\in\mathcal{R}_\ind^{\mathrm{sub}}$ with image in $\Omega_\ind^+$, then $e^{-|T_\kink|}f^2$ and $e^{-2|T_\kink|}f^2$ are both 
dominated by the $\bo(\epsilon_N^{1/3})$ estimate of the error terms involving $e_\ind^+$.

%
%

Now suppose instead that $(x,t)$ corresponds to a point $(y,s)\in\Omega_{\ind+1}^-$. 
As before, we wish to neglect the first term in the denominator 
 of $\dot{R}_{\ind+1}^-$ as defined by \eqref{eq:dotRminus}, so we write $\dot{R}_{\ind+1}^-$ in the form (using the identity $\pv_{\ind+1}(z)=\pu_\ind(z)^{-1}$)
\begin{equation}
\dot{R}_{\ind+1}^- = \frac{-\mathrm{sgn}(\pu_\ind(z))e^{-T_\kink}}{h^2+1},\quad
h = \bo(\epsilon_N^{1/3}e^{-q_\ind}B_{\ind+1,21}(z))=\bo(\epsilon_N^{1/3}e^{-q_\ind}|z-\mathscr{Z}(\pu_\ind)|^{-1}).
\end{equation}
Again, the error bound for $h$ may be replaced with $\bo(\epsilon_N^{1/3}e^{-p_\ind}|y-\mathscr{Z}(\pu_\ind)|^{-1})$ and so $h=\bo(1)$ holds uniformly for $(y,s)\in\Omega_{\ind+1}^-$;
it then follows
by using this result in \eqref{eq:CSddotminus} and recalling \eqref{eq:CSapproximateII}, that whenever $(x,t)\in\mathcal{R}_\ind$ maps into $\Omega_{\ind+1}^-$,
\begin{equation}
\begin{split}
\cos(\tfrac{1}{2}u_N(x,t)) &= (-1)^\ind\,\mathrm{sgn}(\pu_\ind(z))\,\mathrm{sech}(T_\kink) +
\bo(e^{-|T_\kink|}h^2) + \bo(e_{\ind+1}^-(\epsilon_N^{-2/3}r(x,t),s(x,t);\epsilon_N))\\
\sin(\tfrac{1}{2}u_N(x,t)) &= (-1)^{\ind+1}\tanh(T_\kink) + \bo(e^{-2|T_\kink|}h^2) +
\bo(e_{\ind+1}^-(\epsilon_N^{-2/3}r(x,t),s(x,t);\epsilon_N)).
\end{split}
\end{equation}
In particular, we note that the leading terms are \emph{identical} with those arising from
the analysis in the case that the image of $(x,t)$ lies in $\Omega_\ind^+$.  By completely
analogous reasoning as in that case, one learns that the sum of error terms is $\lo(1)$
unless both $T_\kink=\bo(1)$ and also $h$ fails to be small, which occurs in this case only
if $e^{p_\ind}=\bo(\epsilon_N^{1/3})$, or equivalently, if $e^{q_\ind} = \epsilon_N^{2\ind/3}e^{t/\epsilon_N}=\bo(\epsilon_N^{1/3})$.  Under this condition, the definition \eqref{eq:kinkTdefine} implies that $T_\kink\le\bo(1)-\tfrac{1}{3}\log(\epsilon_N^{-1})+\log|\pu_\ind(z)|$, so to have $T_\kink=\bo(1)$ requires that $|z-\mathscr{P}(\pu_\ind)|=\bo(\epsilon_N^{1/3})$.  Moreover, if $(x,t)\in\mathcal{R}_\ind^{\mathrm{sub}}$
with image in $\Omega_{\ind+1}^-$, then both error terms are $\bo(\epsilon_N^{1/3})$.
This completes the proof of Theorem~\ref{thm:kinkapprox}.

\subsection{Proof of Theorem~\ref{thm:grazingapprox} (Grazing kink collisional asymptotics).}
Recall that $z_0\in\mathscr{Z}(\pu_{\ind-1})$ is a simple zero of $\pu_{\ind-1}$, and hence also a simple zero of $\pv_{\ind+1}$.  Note that from the recurrence relation \eqref{eq:Baecklundplus} it follows that $\pv_{\ind+1}'(z_0)=-1/\pu_{\ind-1}'(z_0)$.  Since $z=y+\bo(\epsilon_N^{2/3})$ and $t=s + \bo(\epsilon_N^{4/3})$ for the range of $x$ and
$t$ under consideration, the region in the $(x,t)$-plane defined by the inequalities
\begin{equation}
|t-(\tfrac{2}{3}\ind-\tfrac{1}{3})\epsilon_N\log(\epsilon_N^{-1})|\le\tfrac{1}{3}\epsilon_N\log(\epsilon_N^{-1})\quad\text{and}\quad |z-z_0|\le \mu\epsilon_N^{1/6}\exp(|t-(\tfrac{2}{3}\ind-\tfrac{1}{3})\epsilon_N\log(\epsilon_N^{-1})/(2\epsilon_N))
\end{equation}
is, for both $\mu>0$ and $\epsilon_N$ sufficiently small, mapped into the union $\Omega_{\ind-1}^+\cup\Omega_{\ind+1}^-$ in the $(y,s)$-plane, where $(y,s)=(\epsilon_N^{-2/3}r(x,t),s(x,t))$.  In particular, each point of the image lies either
in the upward-pointing ``tooth'' of $\Omega_{\ind-1}^+$ centered at $y=z_0$ or in the downward-pointing ``tooth'' of $\Omega_{\ind+1}^-$ centered at $y=z_0$.  It follows that
for the points whose images lie in $\Omega_{\ind-1}^+$, the error term $e_{\ind-1}^+(\epsilon_N^{-2/3}r(x,t),s(x,t);\epsilon)$ is $\bo(\epsilon_N^{1/6})$ while for the points whose images lie in $\Omega_{\ind+1}^-$, the error term $e_{\ind+1}^-(\epsilon_N^{-2/3}r(x,t),s(x,t);\epsilon_N)$ is also $\bo(\epsilon_N^{1/6})$ so that in particular these error terms arising
in \eqref{eq:CSapproximateII} will always tend to zero with $\epsilon_N$.  Moreover, the inequality $|z-z_0|=\bo(\epsilon_N^{1/3})$ implies that $|y-z_0|=\bo(\epsilon_N^{1/3})$, from which it follows
that under this condition $(x,t)$ will be mapped into the parts of $\Omega_{\ind-1}^+$
and $\Omega_{\ind+1}^-$ (the centers of two matching teeth) in which the better error estimate of $\bo(\epsilon_N^{1/3})$ holds for 
$e_{\ind-1}^+$ and $e_{\ind+1}^-$ respectively.

Suppose that the image of $(x,t)$ in the $(y,s)$-plane lies in the upward-pointing tooth of $\Omega_{\ind-1}^+$ centered at $y=z_0$.  In this part of the image, the uniform estimate $e^{T_\grazing}=\bo(1)$ holds.   By Taylor expansion,
\begin{equation}
\begin{split}
\pu_{\ind-1}(z)&=\pu_{\ind-1}'(z_0)(z-z_0) (1+ \bo(z-z_0))=\frac{\epsilon_N^{1/3}\pu_{\ind-1}'(z_0)}{2\nu^{1/3}}X_\grazing (1+\bo(z-z_0))\\
B_{\ind-1,12}(z)&=2H_{\ind-1}(z)\pu_{\ind-1}(z)-\pu_{\ind-1}'(z)=-\pu_{\ind-1}'(z_0)+\bo(z-z_0),
\end{split}
\end{equation}
where the spatial coordinate $X_\grazing$ is defined by \eqref{eq:grazingXT}.  It follows from \eqref{eq:dotRplus} using $e^{T_\grazing}=\bo(1)$ that
\begin{equation}
\dot{R}_{\ind-1}^+=-\,\mathrm{sgn}(\pu_{\ind-1}'(z_0))X_\grazing\,\mathrm{sech}(T_\grazing)(1+\bo(z-z_0)),
\end{equation}
and then that
\begin{equation}
\begin{split}
\ddot{C}_{N,\ind-1}^+(x,t)&=(-1)^{\ind-1}\,\mathrm{sgn}(\pu_{\ind-1}'(z_0))\frac{2X_\grazing\,\mathrm{sech}(T_\grazing)}
{1+X_\grazing^2\,\mathrm{sech}^2(T_\grazing)} + \bo(z-z_0)\\
\ddot{S}_{N,\ind-1}^+(x,t)&=(-1)^{\ind-1}\frac{1-X_\grazing^2\,\mathrm{sech}^2(T_\grazing)}{1+X_\grazing^2\,\mathrm{sech}^2(T_\grazing)} + \bo(z-z_0).
\end{split}
\end{equation}
Combining this formula with the above estimates of $e_{\ind-1}^+(x,t;\epsilon_N)$ and
the formulae \eqref{eq:CSapproximateII} proves the Theorem for those $(x,t)$ whose image
in the $(y,s)$-plane lies within the upward-pointing tooth of $\Omega_{\ind-1}^+$ centered
at $y=z_0$.

Now suppose instead that the image of $(x,t)$ in the $(y,s)$-plane lies in the downward-pointing
tooth of $\Omega_{\ind+1}^-$ centered at $y=z_0$.  In this part of the image, the uniform estimate $e^{-T_\grazing}=\bo(1)$ holds.  By Taylor expansion and the fact that $\pv_{\ind+1}'(z_0)=-1/\pu_{\ind-1}'(z_0)$,
\begin{equation}
\begin{split}
\pv_{\ind+1}(z)&=-\frac{1}{\pu_{\ind-1}'(z_0)}(z-z_0)(1+\bo(z-z_0))=
-\frac{\epsilon_N^{1/3}}{2\nu^{1/3}\pu_{\ind-1}'(z_0)}X_\grazing(1+\bo(z-z_0))\\
B_{\ind+1,21}(z)&=\pv_{\ind+1}'(z)-2H_{\ind+1}(z)\pv_{\ind+1}(z)=
-\frac{1}{\pu_{\ind-1}'(z_0)} + \bo(z-z_0).
\end{split}
\end{equation}
It then follows from \eqref{eq:dotRminus} using $e^{-T_\grazing}=\bo(1)$ that
\begin{equation}
\dot{R}_{\ind+1}^- = \mathrm{sgn}(\pu_{\ind-1}'(z_0))X_\grazing\,\mathrm{sech}(T_\grazing)(1+\bo(z-z_0)),
\end{equation}
and then that
\begin{equation}
\begin{split}
\ddot{C}_{N,\ind+1}^-(x,t)&=(-1)^{\ind-1}\,\mathrm{sgn}(\pu_{\ind-1}'(z_0))
\frac{2X_\grazing\,\mathrm{sech}(T_\grazing)}{1+X_\grazing^2\,\mathrm{sech}^2(T_\grazing)}+\bo(z-z_0)\\
\ddot{S}_{N,\ind+1}^-(x,t)&=(-1)^{\ind-1}\frac{1-X_\grazing^2\,\mathrm{sech}^2(T_\grazing)}{1+X_\grazing^2\,\mathrm{sech}^2(T_\grazing)}+\bo(z-z_0).
\end{split}
\end{equation}
Combining this formula with the above estimates of $e_{\ind+1}^-(x,t;\epsilon_N)$ and
the formulae \eqref{eq:CSapproximateII} proves the Theorem for those $(x,t)$ whose image in the $(y,s)$-plane lies within the downward-pointing tooth of $\Omega_{\ind+1}^-$ centered 
at $y=z_0$.

\end{document}